\documentclass[a4paper, pdftex, 12pt]{article} % for arxiv

\usepackage{amsmath,amssymb,amsthm}
\usepackage{hyperref}
\usepackage{graphicx}
\usepackage{multirow}
\usepackage{here}
\usepackage{bm}
\usepackage{float}
\usepackage{layout}
\usepackage{latexsym}
\usepackage{subfigure}
\usepackage{booktabs}
\usepackage{lscape}
\usepackage{xcolor}
\hypersetup{
    colorlinks=true,
    citecolor=black,
    linkcolor=black,
    urlcolor=black
}
\usepackage{authblk}
\usepackage{colortbl}
\usepackage[top=30mm, bottom=30mm, left = 25mm, right = 25mm]{geometry}
\usepackage{natbib}
\usepackage{url}

\theoremstyle{definition}

\newtheorem{lemma}{Lemma}[section]
\newtheorem{proposition}{Proposition}

\newtheorem{algorithm}{Algorithm}

\author{Toshihiro Hirano}
\affil{Kanto Gakuin University}

\title{\textbf{A multi-resolution approximation via linear projection for large spatial datasets}}
\date{}

\begin{document}
\setlength{\baselineskip}{18.5pt}

\maketitle

\begin{abstract}
Recent technical advances in collecting spatial data have been increasing the demand for
 methods to analyze large spatial datasets. The statistical analysis for these types of datasets can provide useful knowledge in various fields. However, conventional spatial statistical methods, such as maximum likelihood estimation and kriging, are impractically time-consuming for large spatial datasets due to the necessary matrix inversions. To cope with this problem, we propose a multi-resolution approximation via linear projection ($M$-RA-lp). The $M$-RA-lp conducts a linear projection approach on each subregion whenever a spatial domain is subdivided, which leads to an approximated covariance function capturing both the large- and small-scale spatial variations. Moreover, we elicit the algorithms for fast computation of the log-likelihood function and predictive distribution with the approximated covariance function obtained by the $M$-RA-lp. Simulation studies and a real data analysis for air dose rates demonstrate that our proposed $M$-RA-lp works well relative to the related existing methods.

\footnote[0]{E-mail: 1hirano2@kanto-gakuin.ac.jp}

\end{abstract}

\section{Introduction\label{introduction}}

Advances in Global Navigation Satellite System (GNSS) and compact sensing devices have made it easy to collect a large volume of spatial data with coordinates in various fields such as environmental science, traffic, and urban engineering. The statistical analysis for these types  of spatial datasets would assist in an evidence-based environmental policy and the efficient management of a smart city.

In spatial statistics, this type of statistical analysis, including model fitting and spatial prediction, has been conducted based on Gaussian processes \citep[see, e.g.,][]{Cressie_2011}. However, traditional spatial statistical methods, such as maximum likelihood estimation and kriging, are computationally infeasible for large spatial datasets, requiring $O(n^3)$ operations for a dataset of size $n$. This is because these methods involve the inversion of an $n \times n$ covariance matrix.

This difficulty has encouraged the development of many efficient statistical techniques for large spatial datasets. 
\cite{Heaton_2019} comprehensively reviews recent developments of these techniques. \cite{Liu_2020} is a detailed survey on current state-of-the-art scalable Gaussian processes in the machine learning literature. Efficient statistical techniques are generally categorized into four types: a sparse approach, a low rank approach, a spectral approach, and an algorithmic approach. 
%These techniques are generally categorized into three types: a sparse approach, a low rank approach, and a spectral approach.
The main idea of the sparse approach is to model either the covariance matrix or its inverse matrix as a sparse matrix. The former method is typically called covariance tapering \citep[][]{Furrer_2006, Kaufman_2008}. 
\cite{Du_2009}, \cite{Chu_2011}, \cite{Wang_2011}, \cite{Hirano_2013}, \cite{Stein_2013}, and \cite{Furrer_2016} discussed further statistical properties of the covariance tapering. However, the covariance tapering ignores the large-scale spatial variation. The latter one includes the approximation of the likelihood function by using products of the lower-dimensional conditional distributions \citep[e.g.,][]{Vecchia_1988, Stein_2004}, an approximation by the Gaussian Markov random field by using a particular type of stochastic partial differential equation \citep[][]{Lindgren_2011}, the representation of a field by using a multiresolution basis \citep[][]{Nychka_2015}, 
and the nearest-neighbor Gaussian process by using a directed acyclic graph \citep[][]{Datta_2016}. 

The low rank approach includes the following two techniques: fixed rank kriging \citep[][]{Cressie_2008} and predictive process \citep[][]{Banerjee_2008}. \cite{Finley_2009} corrected a bias in the predictive process, and \cite{Banerjee_2013} proposed a linear projection approach that is an extension of the predictive process and has the advantage of alleviating the complicated knot selection problem. However, the predictive process and the linear projection are effective for fitting the large-scale spatial variation,
 whereas they are disadvantageous for capturing the small-scale spatial variation. To overcome this problem, \cite{Sang_2012} and \cite{Katzfuss_2017} developed improvements of the predictive process, and \cite{Hirano_2017a} proposed a modification of the linear projection by the covariance tapering based on the idea of \cite{Sang_2012}. 

For the spectral approach, \cite{Fuentes_2007}, \cite{Matsuda_2009}, and \cite{Matsuda_2018} considered the Whittle estimation for either spatial or spatio-temporal data. The Whittle estimation requires no huge matrix inversions. \cite{Fuentes_2007}, \cite{Matsuda_2009}, and \cite{Matsuda_2018} revealed the statistical properties of the estimation by the spectral approach. \cite{Guinness_2019} developed a computationally efficient method for estimating the spectral density from incomplete gridded data based on imputing missing values.

The algorithmic approach focuses more on using schemes than model building and includes \cite{Gramacy_2015}, \cite{Gerber_2018}, and \cite{Guhaniyogi_2018}. 
%\textcolor{red}{For a more detailed recent developments of the efficient statistical techniques, the reader is referred to \cite{Heaton_2019} and \cite{Liu_2020} which are comprehensive review papers.\cite{Gramacy_2015} and \cite{Gerber_2018} and \cite{Guhaniyogi_2018}}

In this paper, we propose a multi-resolution approximation via linear projection ($M$-RA-lp) of Gaussian processes observed at irregularly spaced locations. The $M$-RA-lp implements the linear projection on each subregion obtained by partitioning the spatial domain recursively, resulting in an approximated covariance function that captures both the large- and small-scale spatial variations unlike the covariance tapering and some low rank approaches. Additionally, we derive algorithms for fast computation of the log-likelihood function and predictive distribution with the approximated covariance function obtained by the $M$-RA-lp. Also, these algorithms can be parallelized. 
Our proposed $M$-RA-lp is regarded as a combination of the two recent low rank approaches: a modified linear projection (MLP) \citep[][]{Hirano_2017a} and a multi-resolution approximation ($M$-RA) \citep[][]{Katzfuss_2017}. The $M$-RA-lp extends the MLP by introducing multiple resolutions based on the idea of \cite{Katzfuss_2017}, leading to better approximation accuracy of the covariance function than that by the MLP. Particularly, when the variation of the spatial correlation around the origin is smooth like the Gaussian covariance function, the approximation accuracy of the covariance function by the MLP often degrades. In contrast, the $M$-RA-lp avoids this problem. 
Additionally, the $M$-RA-lp is regarded as an extension of the $M$-RA and enables not only to alleviate the knot selection problem but also to increase empirically numerical stability in specific steps of fast computation algorithms of the $M$-RA. Simulation studies and a real data analysis for air dose rates generally support the effectiveness of our proposed $M$-RA-lp in terms of  computational time, estimation of model parameters, and prediction at unobserved locations when compared with the MLP and $M$-RA.

The remainder of this paper is organized as follows. We introduce a Gaussian process model for spatial datasets in Section \ref{sec:Gaussian_process_model}. Section \ref{sec:m_ra_lp} describes our proposed $M$-RA-lp. In Section \ref{sec:inference}, we present the algorithms for fast computation of the log-likelihood function and predictive distribution. In Section \ref{sec:illustrations}, we provide the results of the simulation studies and real data analysis. Our conclusions and future studies are discussed in Section \ref{sec:conclusion_future_studies}. 
%The appendices contain technical lemmas, the proof of the proposition, and the derivation of the algorithms.
The appendices contain technical lemmas, the proof of the proposition, and the derivation and distributed computing of the algorithms.

\section{Gaussian process model for spatial datasets\label{sec:Gaussian_process_model}}

For $\bm{s} = (s_1,\ldots,s_d)^{\top} \in D_0 \subset \mathbb{R}^d$ $(d \in \mathbb{N}^{+})$, we consider the following model
\begin{align*}
%\label{eq:model}
Z(\bm{s}) = Y_0(\bm{s}) + \varepsilon(\bm{s}),
\end{align*}
where $Z(\bm{s})$ is a response variable observed at location $\bm{s}$. $Y_0(\bm{s}) \sim \mbox{GP}(0, C_0)$ is a zero-mean Gaussian process with a covariance function $C_0(\bm{s},\bm{s}^{*})$ ($\bm{s}$, $\bm{s}^{*} \in D_0$), which is a positive definite function. $C_{0}(\bm{s},\bm{s}^{*})$ is specified as $\sigma^2 \rho_{0} (\bm{s},\bm{s}^{*}; \bm{\theta})$ where $\sigma^2 = \mbox{Var}(Y_0(\bm{s}))$, and $\rho_{0}$ means a correlation function of $Y_0(\bm{s})$ with a parameter vector $\bm{\theta}$. For example, $\bm{\theta}$ may include a range parameter. $\varepsilon(\bm{s})$ is a zero-mean independent process following a normal distribution with a variance $\tau^2$ and expresses a measurement error that is often referred to as a nugget effect \citep[see, e.g.,][]{Cressie_1993}. It is assumed that $\{ Y_0(\bm{s}) \}$ and $\{ \varepsilon(\bm{s}) \}$ are independent. 

In what follows, for a generic Gaussian process $X(\bm{s}) \sim \mbox{GP}(0, C)$
 and sets of the vectors, that is, $A = \{ \bm{a}_1, \ldots,\bm{a}_N \}$ and $B = \{ \bm{b}_1, \ldots,\bm{b}_M \}$ ($\bm{a}_i, \bm{b}_j \in \mathbb{R}^{d'}$, $i= 1,\ldots,N,\; j=1,\ldots,M$),  we write $\bm{X}(A) = \left( X(\bm{a}_1), \ldots, X(\bm{a}_N) \right)^{\top}$ and $(C(A,B))_{ij} = C(\bm{a}_i,\bm{b}_j)$ $(i= 1,\ldots,N,\; j=1,\ldots,M)$. 

Suppose that we observe the response variable $Z(\bm{s})$ at a set of $n$ spatial locations $S_0 = \{ \bm{s}_1,\ldots,\bm{s}_n \}$.  
The observation vector is denoted by $\bm{Z}(S_0)=(Z(\bm{s}_1),\ldots,\\Z(\bm{s}_n))^{\top}$. The major goal in the spatial statistical analysis is to estimate the parameters $\Omega = (\sigma^2, \bm{\theta}, \tau^2)$ and to predict $\bm{Y}_0 \left(S_0^P \right)=(Y_0(\bm{s}^{P}_1),\ldots, Y_0(\bm{s}^{P}_{n'}))^{\top}$ at a set of $n'$ unobserved locations $S_0^P = \{ \bm{s}^{P}_1,\ldots,\bm{s}^{P}_{n'} \}$.

We adopt the maximum likelihood method to estimate the unknown parameters $\Omega$. The log-likelihood function is
\begin{align}
\label{eq:log-lik}
l(\Omega) =& - \frac{n}{2} \log(2 \pi) - \frac{1}{2} \log \left[ \mbox{det} \left\{ C_0(S_0,S_0) + \tau^2 \bm{\mbox{I}}_n \right\} \right] \notag\\
& - \frac{1}{2} \bm{Z}(S_0)^{\top}
 \left\{ C_0(S_0,S_0) + \tau^2 \bm{\mbox{I}}_n \right\}^{-1} \bm{Z}(S_0),
\end{align}
where $(C_0(S_0,S_0))_{ij} = C_{0}(\bm{s}_i,\bm{s}_j)$ $(i,j=1,\ldots,n)$ and $\bm{\mbox{I}}_n$ is an $n \times n$ identity matrix. After the parameter inference is completed, the spatial prediction is conducted by using the resulting maximum likelihood estimates. For the spatial prediction, we aim to obtain the predictive distribution
\begin{align}
\label{eq:pred_dist}
\bm{Y}_0 \left(S_0^P \right) | \bm{Z}(S_0) \sim & \mathcal{N} \left( C_0 \left(S_0^P, S_0 \right)  \left\{ C_0(S_0,S_0) + \tau^2 \bm{\mbox{I}}_n \right\}^{-1} \bm{Z}(S_0),  C_0 \left(S_0^P, S_0^P \right) \right. \notag\\
& \left. - C_0 \left(S_0^P, S_0 \right) \left\{ C_0(S_0,S_0) + \tau^2 \bm{\mbox{I}}_n \right\}^{-1} C_0 \left(S_0^P, S_0 \right)^{\top} \right).
\end{align}

\eqref{eq:log-lik} and \eqref{eq:pred_dist} involve the determinant and/or inverse matrix of the $n \times n$ matrix $C_0(S_0,S_0) + \tau^2 \bm{\mbox{I}}_n$. The inverse matrix calculation requires $O(n^3)$ operations, which causes a formidable computation when evaluating the log-likelihood function \eqref{eq:log-lik} and calculating both the mean vector and the covariance matrix in \eqref{eq:pred_dist} for large spatial datasets. 
Furthermore, \eqref{eq:log-lik} and \eqref{eq:pred_dist} require $O(n^2)$ memory, which often causes a lack of memory for large spatial datasets.

\section{Multi-resolution approximation via linear projection\label{sec:m_ra_lp}}

To address the computational burden, we propose the $M$-RA-lp. First, some notations are  defined based on \cite{Katzfuss_2017} in order to describe the $M$-RA-lp concisely. Let $m$ ($m = 0,\ldots,M$) denote a resolution. 
For $m = 0,\ldots,M$, $D_{j_1,\ldots,j_m}$ ($1 \le j_i \le J_i$, $2 \le J_i$, $i = 1,\ldots, M$) is obtained by partitioning the entire spatial domain $D_0$ and denotes a numbered subregion at the $m$th resolution. Throughout this paper, the index ($j_1 , \ldots, j_m$) and the index ($j_1 , \ldots, j_m,a$) for $m=0$ correspond to the index $0$ and the index $a$, respectively. For example, $D_{j_1,\ldots,j_m}$ for $m=0$ is $D_0$. The domain partitioning must satisfy the following assumption
\begin{align*}
D_{j_1,\ldots,j_m} = \bigcup_{j_{m+1} = 1, \ldots, J_{m+1}} D_{j_1,\ldots,j_m,j_{m+1}}, \quad 
D_{j_1,\ldots,j_m,k} \cap D_{j_1,\ldots,j_m,l} = \emptyset,&\\
 \quad 1 \le k \neq l \le J_{m+1},&
\end{align*}
for $m=0,\ldots,M-1$. This assumption implies that each subregion is recursively divided into smaller disjoint subregions while increasing the resolution. We need to prespecify $M$ and how to partition each $D_{j_1,\ldots,j_m}$ ($m=0,\ldots,M-1$). Let $S_{j_1,\ldots, j_m}$ be a subset of observed locations on $D_{j_1,\ldots,j_m}$ ($m=0,\ldots,M$). 

Hereafter, for a generic notation $X_{j_1, \ldots, j_m}$ of a set, vector, or matrix, we assume that the stacked one of $X_{j_1, \ldots, j_m}$ is arranged in ascending order by the index ($j_1, \ldots, j_m$) ($m=1,\ldots,M$). 
When comparing the number in order from the left of the index ($j_1, \ldots, j_m$), the first determined magnitude relationship is adopted as that of the index ($j_1, \ldots, j_m$). 
For example, if $\left| S_{j_1,\ldots,j_M} \right| \ge 1$ where $| \cdot |$ denotes the size of the set, we can have the recursive expression $\bm{Z}(S_{j_1,\ldots, j_m}) = ( \bm{Z}(S_{j_1,\ldots, j_m,1})^{\top}, \bm{Z}(S_{j_1,\ldots, j_m,2})^{\top}, \ldots, \bm{Z}(S_{j_1,\ldots, j_m,J_{m+1}})^{\top}  )^{\top}$ for $m=0,\ldots,M-1$.

We also need to select a set of knots on each subregion $D_{j_1,\ldots,j_m}$ which is denoted by $Q_{j_1,\ldots, j_m}$ ($m=0,\ldots,M$). 
Based on \cite{Katzfuss_2017}, it is assumed that $Q_{j_1,\ldots, j_M} = S_{j_1,\ldots, j_M}$. The set of knots at the $M$th resolution is restrictive, but we can select $Q_{j_1,\ldots, j_m}$, for example, as lattice points on $D_{j_1,\ldots,j_m}$ and a subset of $S_{j_1,\ldots, j_m}$ ($m=0,\ldots,M-1$). 
In simulation studies and the real data analysis of this paper, we select $Q_{j_1,\ldots, j_m}$ randomly from $S_{j_1,\ldots, j_m}$ ($m=0,\ldots, M-1$). Moreover, we define $Q^{(m)}=\left\{ Q_{j_1,\ldots, j_m} \middle| 1 \le j_1 \le J_1, \ldots, 1 \le j_m \le J_m \right\}$ ($m=0,\ldots,M$).

Finally, we introduce an $r_{j_1,\ldots, j_m} \times \left| Q_{j_1,\ldots, j_m} \right|$ matrix  $\Phi_{j_1,\ldots, j_m}$ ($m=0,\ldots,M-1$, $1 \le r_{j_1,\ldots, j_m}$) where $\text{rank} \left( \Phi_{j_1,\ldots, j_m} \right) = r_{j_1,\ldots, j_m}$ and its row-norm is equal to 1. $r_{j_1,\ldots, j_m}$ is much smaller than the sample size $n$ to avoid the computational burden. For $m=0,\ldots,M-1$, we define
\begin{align*}
\Phi^{(m)} = 
\begin{pmatrix}
\Phi_{1,\ldots,1} & & O \\
 & \ddots & \\
 O & & \Phi_{J_1,\ldots,J_m}
\end{pmatrix}
.
\end{align*}
The selection of $\Phi_{j_1,\ldots, j_m}$ will be discussed in Section \ref{subsec:select_phi}.

$\Phi_{j_1,\ldots, j_m}$ plays a critical role in the linear projection \citep[see][]{Banerjee_2013,Hirano_2017a}. In the linear projection, for $\bm{s} \in D_0$, we define 
\begin{align*}
\tau_0(\bm{s}) &= E[Y_0(\bm{s}) | \Phi_0 \bm{Y}_0(Q_0)] \\
&=C_0(\bm{s},Q_0) \Phi_0^{\top} \left\{ \Phi_0  C_0(Q_0,Q_0) \Phi_0^{\top} \right\}^{-1} \Phi_0 \bm{Y}_0(Q_0).
\end{align*}
Then, it follows that
\begin{align*}
C_{\tau_0}(\bm{s}_1,\bm{s}_2) &= \mbox{Cov}(\tau_0(\bm{s}_1),\tau_0(\bm{s}_2))\\
&=C_0(\bm{s}_1,Q_0) \Phi_0^{\top} \widehat{K}_{0}^0 \Phi_0 C_0(\bm{s}_2,Q_0)^{\top},
%&=C_0(\bm{s}_1,Q_0) \Phi_0^{\top} (\Phi_0  K_0^0 \Phi_0^{\top})^{-1} \Phi_0 C_0(\bm{s}_2,Q_0)^{\top}.
\end{align*}
where $\widehat{K}_{0}^0 = (\Phi_0  K_0^0 \Phi_0^{\top})^{-1}$ and $K_0^0 = C_0(Q_0,Q_0)$. The linear projection uses $C_{\tau_0}$ as the main approximation of $C_0$ and is identical with the predictive process in the case of $\Phi_0 =\bm{\mbox{I}}_{\left| Q_0 \right|}$. 
The simulation studies and real data analyses in \cite{Banerjee_2013} demonstrated that it achieved better performance efficiently than that of the predictive process.

\subsection{Algorithm for approximating the covariance function\label{subsec:m_ra_lp_alg}}

In our proposed $M$-RA-lp, the calculation of $C_{\tau_0}$ is regarded as the linear projection at resolution 0, and the linear projection is applied repeatedly to its approximation error at resolutions $m=1,\ldots,M-1$. 
We will state the details of our proposed algorithm for approximating the covariance function.

\begin{algorithm}[Approximation of the covariance function $C_0(\bm{s}_1^{*},\bm{s}_2^{*})$]
\label{alg:approx_cov}
Given $\bm{s}_1^{*},\bm{s}_2^{*} \in D_0$, $M \ge 0$, $D_{j_1,\ldots,j_m}$ ($m = 1,\ldots, M$, $1 \le j_1 \le J_1, \ldots, 1 \le j_M \le J_M$), $Q_{j_1,\ldots,j_m}$ ($m = 0,\ldots, M-1$, $1 \le j_1 \le J_1, \ldots, 1 \le j_{M-1} \le J_{M-1}$), and $\gamma > 0$, find the approximated covariance function $C_{M \text{-RA-lp}}(\bm{s}_1^{*},\bm{s}_2^{*})$. If $M=0$, output $C_{M \text{-RA-lp}}(\bm{s}_1^{*},\bm{s}_2^{*})=C_0(\bm{s}_1^{*},\bm{s}_2^{*})$.  Otherwise, set $m=0$ and $C_{M \text{-RA-lp}}(\bm{s}_1^{*},\bm{s}_2^{*})=0$ initially.

\bigskip
\noindent
\textit{Step} 1. %Calculate $C_{\tau_0}(\bm{s}_1,\bm{s}_2)$ ($\bm{s}_1,\bm{s}_2 \in D_0$) and 
Set $C_{M \text{-RA-lp}}(\bm{s}_1^{*},\bm{s}_2^{*}) = C_{M \text{-RA-lp}}(\bm{s}_1^{*},\bm{s}_2^{*}) + C_{\tau_0}(\bm{s}_1^{*},\bm{s}_2^{*})$.

\noindent
\textit{Step} 2. When $m+1 < M$, if $\bm{s}_1^{*}$ and $\bm{s}_2^{*}$ are in the same subregion $D_{j^{*}_1,\ldots,j^{*}_{m+1}}$, 
set $m = m+1$ and go to Step 3. When $m+1 = M$, if $\bm{s}_1^{*}$ and $\bm{s}_2^{*}$ are in the same subregion $D_{j^{*}_1,\ldots,j^{*}_{M}}$, go to Step 4. Otherwise, go to Step  5.

\noindent
\textit{Step} 3. Define
\begin{align}
\label{eq:cov_resolution_m}
C_m(\bm{s}_1,\bm{s}_2) =
\begin{cases}
C_{m-1}(\bm{s}_1,\bm{s}_2) - C_{\tau_{m-1}}(\bm{s}_1,\bm{s}_2), \quad \bm{s}_1,\bm{s}_2 \in D_{j_1,\ldots,j_m}\\
\hspace{3.5cm} (1 \le j_i \le J_i,\; i=1,\ldots,m),\\
0, \quad \mbox{otherwise}.
\end{cases}
\end{align}
Next, for $\bm{s} \in D_0$, let $\delta_m(\bm{s}) \sim \mbox{GP}(0,C_m)$ be a zero-mean Gaussian process with the degenerate covariance function $C_m$. By conducting the linear projection for $\delta_m(\bm{s})$ at the $m$th resolution, we obtain
\begin{align*}
\tau_m(\bm{s}) &= E[\delta_m(\bm{s}) | \Phi^{(m)} \bm{\delta}_m(Q^{(m)})]\\
&=C_m(\bm{s},Q^{(m)}) \Phi^{{(m)}^{\top}} \left\{ \Phi^{(m)} C_m(Q^{(m)},Q^{(m)}) \Phi^{{(m)}^{\top}} \right\}^{-1} \Phi^{(m)}\bm{\delta}_m(Q^{(m)})
\end{align*}
and
\begin{align}
\label{eq:lp_resolution_m}
C_{\tau_m}(\bm{s}_1,\bm{s}_2) &= \mbox{Cov}(\tau_m(\bm{s}_1),\tau_m(\bm{s}_2)) \notag \\
&=
\begin{cases}
C_m(\bm{s}_1,Q_{j_1,\ldots,j_m}) \Phi_{j_1,\ldots,j_m}^{\top} \widehat{K}_{j_1,\ldots,j_m}^m \Phi_{j_1,\ldots,j_m} C_m(\bm{s}_2,Q_{j_1,\ldots,j_m})^{\top}, \\
\hspace{3cm}  \bm{s}_1,\bm{s}_2 \in D_{j_1,\ldots,j_m} \; (1 \le j_i \le J_i,\; i=1,\ldots,m),\\
0, \quad \mbox{otherwise},
\end{cases}
\end{align}
where $\widehat{K}_{j_1,\ldots,j_m}^m = \left( \Phi_{j_1,\ldots,j_m} K_{j_1,\ldots,j_m}^m  \Phi_{j_1,\ldots,j_m}^{\top} \right)^{-1}$ and $K_{j_1,\ldots,j_m}^m = C_m(Q_{j_1,\ldots,j_m},\\
Q_{j_1,\ldots,j_m})$. Set $C_{M \text{-RA-lp}}(\bm{s}_1^{*},\bm{s}_2^{*}) = C_{M \text{-RA-lp}}(\bm{s}_1^{*},\bm{s}_2^{*}) + C_{\tau_m}(\bm{s}_1^{*},\bm{s}_2^{*})$ and go to Step 2.

\noindent
\textit{Step} 4. Define 
\begin{align}
\label{eq:cov_resolution_M}
C_M(\bm{s}_1,\bm{s}_2) =
\begin{cases}
C_{M-1}(\bm{s}_1,\bm{s}_2) - C_{\tau_{M-1}}(\bm{s}_1,\bm{s}_2), \quad \bm{s}_1,\bm{s}_2 \in D_{j_1,\ldots,j_M} \\ 
\hspace{3.5cm} (1 \le j_i \le J_i,\; i=1,\ldots,M),\\
0, \quad \mbox{otherwise}.
\end{cases}
\end{align}
Set $C_{M \text{-RA-lp}}(\bm{s}_1^{*},\bm{s}_2^{*}) = C_{M \text{-RA-lp}}(\bm{s}_1^{*},\bm{s}_2^{*}) + C_{M}(\bm{s}_1^{*},\bm{s}_2^{*}) T_{\gamma}(\bm{s}_1^{*},\bm{s}_2^{*})$ where $T_{\gamma}(\bm{s}_1,\bm{s}_2)$ ($\gamma > 0$, $\bm{s}_1,\bm{s}_2 \in D_0$) is a compactly supported correlation function with $T_{\gamma}(\bm{s}_1,\bm{s}_2) = 0$ for $\| \bm{s}_1 - \bm{s}_2 \| \ge \gamma$, and $\| \cdot \|$ denotes the Euclidean norm.

\noindent
\textit{Step} 5. Output $C_{M \text{-RA-lp}}(\bm{s}_1^{*},\bm{s}_2^{*})$. 
\end{algorithm}

Step 3 represents the linear projection at the $m$th resolution. In order to derive the fast computation algorithms in Section \ref{sec:inference}, $C_m(\bm{s}_1,\bm{s}_2)$ is defined as 0 if $\bm{s}_1$ and $\bm{s}_2$ do not belong to the same subregion at the $m$th resolution for $m=1,\ldots,M$. For the same reason, we introduce $T_{\gamma}$ in Step 4. Some compactly supported correlation functions have been developed \citep[see, e.g.,][]{Wendland_1995,Gneiting_2002,Bevilacqua_2019}. Examples of these types of functions include the spherical covariance function
\begin{align*}
T_{\gamma}(\bm{s}_1,\bm{s}_2) = \left(1-\frac{ \| \bm{s}_1-\bm{s}_2 \| }{\gamma} \right)^2_{+} \left( 1 + \frac{\| \bm{s}_1-\bm{s}_2 \|}{2\gamma} \right)
\end{align*}
and the Wendland$_2$ taper function \citep[see][]{Wendland_1995,Furrer_2006}:
\begin{align*}
T_{\gamma}(\bm{s}_1,\bm{s}_2) = \left(1-\frac{\| \bm{s}_1-\bm{s}_2 \|}{\gamma} \right)^6_{+} \left(1+\frac{6 \| \bm{s}_1-\bm{s}_2 \|}{\gamma}+\frac{35 \| \bm{s}_1-\bm{s}_2 \|^2}{3\gamma^2} \right).
\end{align*}
For simplicity, we use the spherical covariance function as $T_{\gamma}$ in this paper.

If $\left| S_{j_1,\ldots,j_M} \right| \ge 1$, the covariance matrix by using $C_{M \text{-RA-lp}}$ defined by Algorithm \ref{alg:approx_cov} is
\begin{align}
\label{eq:cov_mat_m-ra-lp}
C_{M \text{-RA-lp}}&(S_0,S_0) = \sum_{m=0}^{M-1} C_{\tau_{m}}(S_0,S_0) + C_M(S_0,S_0) \circ T_{\gamma}(S_0,S_0) \notag \\
=& B_0^0 \Phi_0^{\top} \widehat{K}_{0}^0 \Phi_0 B_0^{{0}^{\top}} \notag \\ 
&+
\scalebox{1.0}{$
\begin{pmatrix}
B_1^1 & & O \\
 & \ddots & \\
 O & & B_{J_1}^1
\end{pmatrix}
\begin{pmatrix}
\Phi_1^{\top} & & O \\
 & \ddots & \\
 O & & \Phi_{J_1}^{\top}
\end{pmatrix}
\begin{pmatrix}
\widehat{K}_{1}^1 & & O \\
 & \ddots & \\
 O & & \widehat{K}_{J_1}^1
\end{pmatrix}
\begin{pmatrix}
\Phi_1 & & O \\
 & \ddots & \\
 O & & \Phi_{J_1}
\end{pmatrix}
$}\notag
\\
&\times
\scalebox{1.0}{$
\begin{pmatrix}
B_1^{1^{\top}} & & O \\
 & \ddots & \\
 O & & B_{J_1}^{1^{\top}}
\end{pmatrix}
$} \notag \\
&\;\; \vdots \notag \\
&+
\scalebox{0.85}{$
\begin{pmatrix}
B_{1,\ldots,1}^{M-1} & & O \\
 & \ddots & \\
 O & & B_{J_1,\ldots,J_{M-1}}^{M-1}
\end{pmatrix}
\begin{pmatrix}
\Phi_{1,\ldots,1}^{\top} & & O \\
 & \ddots & \\
 O & & \Phi_{J_1,\ldots,J_{M-1}}^{\top}
\end{pmatrix}
\begin{pmatrix}
\widehat{K}_{1,\ldots,1}^{M-1} & & O \\
 & \ddots & \\
 O & & \widehat{K}_{J_1,\ldots,J_{M-1}}^{M-1}
\end{pmatrix}
$}\notag
\\
&\times
\scalebox{0.85}{$
\begin{pmatrix}
\Phi_{1,\ldots,1} & & O \\
 & \ddots & \\
 O & & \Phi_{J_1,\ldots,J_{M-1}}
\end{pmatrix}
\begin{pmatrix}
B_{1,\ldots,1}^{{M-1}^{\top}} & & O \\
 & \ddots & \\
 O & & B_{J_1,\ldots,J_{M-1}}^{{M-1}^{\top}}
\end{pmatrix}
$} \notag \\
&+
\scalebox{1.0}{$
\begin{pmatrix}
C_M (S_{1,\ldots,1},S_{1,\ldots,1}) & & O\\
 & \ddots & \\
 O & & C_M (S_{J_1,\ldots,J_M},S_{J_1,\ldots,J_M})
\end{pmatrix}
\circ
T_{\gamma}(S_0,S_0),
$}
\end{align}
where $B_{j_1,\ldots,j_m}^m = C_m(S_{j_1,\ldots,j_m},Q_{j_1,\ldots,j_m})$ ($m=0,\ldots,M-1$), 
and the symbol ``$\circ$'' refers to the Hadamard product. For $m=0,\ldots,M-1$, the $(m+1)$st term in \eqref{eq:cov_mat_m-ra-lp} corresponds to the linear projection at the $m$th resolution. We observed that the linear projection at the higher resolution improved the approximation of the original covariance function on smaller and smaller scales.  Consequently, the overlap between the covariance tapering at the highest resolution and the effect of iterative approximation in the $M$-RA-lp can occur. 
By selecting low $M$, we may be able to bypass the redundant overlap. 
Moreover, for large $M$, the approximated covariance functions up to resolution $m$ ($m=0,\ldots,M$) in Algorithm \ref{alg:approx_cov}, that is, the summation up to the $(m+1)$st term in \eqref{eq:cov_mat_m-ra-lp}, are often almost unchanged at high resolutions. This fact provides us with suggestions on selecting an appropriate $M$. 
%In this case, the magnitude of the overlap between the covariance tapering at the highest resolution and the effect of iterative approximation in the $M$-RA-lp is large. We can often bypass the redundant overlap by selecting low $M$. 

%For large $M$, the overlap between the covariance tapering at the highest resolution and the effect of iterative approximation in the $M$-RA-lp can occur because we observe that  the linear projection at the high resolution approximates the small-scale spatial variations of the original covariance function. 
%By selecting low $M$, we may be able to bypass the redundant overlap.

The following proposition proves the theoretical properties associated with Algorithm \ref{alg:approx_cov}. Note that the case of $M=0$ in the following proposition is excluded because the validity of Algorithm \ref{alg:approx_cov} is clear from $C_{M \text{-RA-lp}}=C_0$ when $M=0$.

\begin{proposition}
\label{prop:m-ra-lp}
Given $M \ge 1$, $D_{j_1,\ldots,j_m}$ ($m = 1,\ldots, M$, $1 \le j_1 \le J_1, \ldots, 1 \le j_M \le J_M$), $Q_{j_1,\ldots,j_m}$ ($m = 0,\ldots, M-1$, $1 \le j_1 \le J_1, \ldots, 1 \le j_{M-1} \le J_{M-1}$), and $\gamma > 0$, suppose that $\Phi_{j_1,\ldots,j_m}$, which satisfies $R(\Phi_{j_1,\ldots,j_m}^{\top}) \cap R(K_{j_1,\ldots,j_m}^m)^{\perp} = \{ \bm{0} \}$ ($m=1,\ldots,M-1$) if $M \ge 2$, is selected where $R(\cdot)$ means the column space of $\cdot$, the symbol ``$\perp$'' refers to the orthogonal complement, and $\bm{0}$ is the zero vector.
\par
\noindent
(a) For $m=1,\ldots,M$, $C_m$ is a positive semidefinite function.
\par
\noindent
(b) For $m=0,\ldots,M-1$, $\Phi_{j_1,\ldots,j_m} K_{j_1,\ldots,j_m}^m \Phi_{j_1,\ldots,j_m}^{\top}$ is positive definite.
\par
\noindent
(c) $C_{M \text{-RA-lp}}$ is a positive semidefinite function.
\par
\noindent
(d) If $\bm{s}_1 = \bm{s}_2$, then $C_{M \text{-RA-lp}}(\bm{s}_1,\bm{s}_2)=C_0(\bm{s}_1,\bm{s}_2)$.
\end{proposition}

For example, if the normalized vectors selected from $R(K_{j_1,\ldots,j_m}^m)$ can be used as the linearly independent column vectors of $\Phi_{j_1,\ldots,j_m}^{\top}$ 
in Step 3 of Algorithm \ref{alg:approx_cov}, the assumption of Proposition \ref{prop:m-ra-lp} is satisfied. For $m=1,\ldots,M-1$, if $K_{j_1,\ldots,j_m}^m = O$, this assumption does not hold  because of the definition of $\Phi_{j_1,\ldots,j_m}$. 
Proposition \ref{prop:m-ra-lp} (b) guarantees the existence of inverse matrices in Step 3 of Algorithm \ref{alg:approx_cov}. In Section \ref{sec:inference}, we propose the two fast computation algorithms of the log-likelihood function and predictive distribution defined by replacing $C_0$ in \eqref{eq:log-lik} and \eqref{eq:pred_dist} with $C_{M \text{-RA-lp}}$. 
Proposition \ref{prop:m-ra-lp} (c) guarantees the existence of the inverse matrix of $C_{M \text{-RA-lp}}(S_0,S_0) + \tau^2 \bm{\mbox{I}}_n$ appearing in this replacement. 
Propositions \ref{prop:m-ra-lp} (a) and (c) are the linear projection versions of the results in the proof of Proposition 1 of \cite{Katzfuss_2017}. 
Proposition \ref{prop:m-ra-lp} (d) states that 
Algorithm \ref{alg:approx_cov} completely recovers the variance of the original Gaussian process.

\subsection{Selection of $\Phi_{j_1,\ldots, j_m}$\label{subsec:select_phi}}

We will discuss how to select $\Phi_{j_1,\ldots, j_m}$ in the linear projection at each resolution based on the argument of Section 3.1 of \cite{Banerjee_2013}. 
Now, we consider the case of $\Phi_{j_1,\ldots, j_m}^{\top} = U_{j_1,\ldots, j_m}^{(r_{j_1,\ldots, j_m})}$ where $U_{j_1,\ldots, j_m}^{(r_{j_1,\ldots, j_m})}$ is a $\left| Q_{j_1,\ldots, j_m} \right| \times r_{j_1,\ldots, j_m}$ matrix whose $i$th column vector is the eigenvector corresponding to the $i$th eigenvalue of the positive semidefinite matrix $K_{j_1,\ldots,j_m}^m$ in descending order of magnitude ($i = 1, \ldots, r_{j_1,\ldots, j_m}$). Suppose that $r_{j_1,\ldots, j_m} < \mbox{rank}(K_{j_1,\ldots,j_m}^m)$ is satisfied. Since $R(\Phi_{j_1,\ldots, j_m}^{\top}) = R \left( U_{j_1,\ldots, j_m}^{(r_{j_1,\ldots, j_m})} \right) \subset R(K_{j_1,\ldots,j_m}^m)$ and $R(K_{j_1,\ldots,j_m}^m) \cap R(K_{j_1,\ldots,j_m}^m)^{\perp} = \{ \bm{0} \}$, it follows that $R(\Phi_{j_1,\ldots,j_m}^{\top}) \cap R(K_{j_1,\ldots,j_m}^m)^{\perp} = \{ \bm{0} \}$. In addition, from Schmidt's approximation theorem \citep[see][]{Stewart_1993, Puntanen_2011}, 
$C_{\tau_m}(Q_{j_1,\ldots, j_m},Q_{j_1,\ldots, j_m}) = K_{j_1,\ldots,j_m}^m \Phi_{j_1,\ldots, j_m}^{\top} \widehat{K}_{j_1,\ldots,j_m}^m \Phi_{j_1,\ldots, j_m} K_{j_1,\ldots,j_m}^m$ is the best rank-$r_{j_1,\ldots, j_m}$ approximation of $C_m(Q_{j_1,\ldots, j_m},Q_{j_1,\ldots, j_m})=K_{j_1,\ldots,j_m}^m$ in the sense of the Frobenius norm for matrices. Therefore, one reasonable selection is $\Phi_{j_1,\ldots, j_m}^{\top} = U_{j_1,\ldots, j_m}^{(r_{j_1,\ldots, j_m})}$, but the derivation of eigenvalues and eigenvectors of $K_{j_1,\ldots,j_m}^m$ involves  $O\left(\left| Q_{j_1,\ldots, j_m} \right|^3 \right)$ computations \citep[][]{Golub_2012}.

To address this problem, \cite{Banerjee_2013} used a stochastic matrix approximation technique to find $\Phi_{0}$ in the linear projection at resolution 0 on the basis of Algorithm 4.2 of \cite{Halko_2011}. \cite{Banerjee_2013} and \cite{Hirano_2017a} demonstrated its effectiveness in practice through simulation studies and real data analyses. Thus, in this paper, we implement this technique at each resolution, which enables us to obtain $\Phi_{j_1,\ldots, j_m}$ efficiently. However, whether the selected $\Phi_{j_1,\ldots, j_m}$ satisfies the assumption of Proposition \ref{prop:m-ra-lp} rigorously is a future study.

The following algorithm corresponds to Algorithm 2 of \cite{Banerjee_2013} at each resolution.

\begin{algorithm}[Selection of $\Phi_{j_1,\ldots, j_m}$ \citep{Halko_2011, Banerjee_2013}]
\label{alg:stochastic_matrix_approx}
Given $K_{j_1,\ldots,j_m}^m$, a target error $\varepsilon > 0$, and $c \in \mathbb{N}^+$, find the $r_{j_1,\ldots, j_m} \times \left| Q_{j_1,\ldots, j_m} \right|$ matrix $\Phi_{j_1,\ldots, j_m}$ for $m=0,\ldots,M-1$. The selected $\Phi_{j_1,\ldots, j_m}$ satisfies $\| K_{j_1,\ldots,j_m}^m - \Phi_{j_1,\ldots, j_m}^{\top} \Phi_{j_1,\ldots, j_m} K_{j_1,\ldots,j_m}^m \|_F < \varepsilon$ with probability $1- \left| Q_{j_1,\ldots, j_m} \right|/10^c$ where $\| \cdot \|_F$ denotes the Frobenius norm for matrices.

\bigskip
\noindent
\textit{Step} 1. Initially, set $j = 0$ and $\Phi_{j_1,\ldots, j_m}^{(0)}=[ \; ]$, which is the $0 \times \left| Q_{j_1,\ldots, j_m} \right|$ empty matrix.

\noindent
\textit{Step} 2. Draw $c$ length-$\left| Q_{j_1,\ldots, j_m} \right|$ random vectors $\bm{\omega}^{(1)},\ldots,\bm{\omega}^{(c)}$ with independent entries from ${\cal N}(0,1)$.

\noindent
\textit{Step} 3. Calculate $\bm{\kappa}^{(i)} = K_{j_1,\ldots,j_m}^m \bm{\omega}^{(i)}$ for $i = 1,\ldots,c$.

\noindent
\textit{Step} 4. Check whether $\max_{i = 1,\ldots,c}(\| \bm{\kappa}^{(i+j)} \|) < \{ (\pi/2)^{1/2} \varepsilon \}/10$. If it holds, go to Step 11. Otherwise, go to Step 5.

\noindent
\textit{Step} 5. Set $ j = j + 1$. Recalculate $\bm{\kappa}^{(j)} = \left( \bm{\mbox{I}}_{\left| Q_{j_1,\ldots, j_m} \right|} -  \Phi_{j_1,\ldots, j_m}^{(j-1)^{\top}}   \Phi_{j_1,\ldots, j_m}^{(j-1)} \right) \bm{\kappa}^{(j)}$ and $\bm{\phi}^{(j)} = \bm{\kappa}^{(j)}/\| \bm{\kappa}^{(j)} \|$.

\noindent
\textit{Step} 6. Set $\Phi_{j_1,\ldots, j_m}^{(j)} = \left[  \Phi_{j_1,\ldots, j_m}^{(j-1)^{\top}} \;\;\; \bm{\phi}^{(j)} \right]^{\top}$, which stands for the concatenation of the matrix and row vector. 

\noindent
\textit{Step} 7. Draw a length-$\left| Q_{j_1,\ldots, j_m} \right|$ random vector $\bm{\omega}^{(j + c)}$ with independent entries from ${\cal N}(0,1)$.

\noindent
\textit{Step} 8. Calculate $\bm{\kappa}^{(j+c)} = \left( \bm{\mbox{I}}_{\left| Q_{j_1,\ldots, j_m} \right|} - \Phi_{j_1,\ldots, j_m}^{(j)^{\top}} \Phi_{j_1,\ldots, j_m}^{(j)} \right) K_{j_1,\ldots,j_m}^m \bm{\omega}^{(j + c)}$.

\noindent
\textit{Step} 9. Recalculate $\bm{\kappa}^{(i)} = \bm{\kappa}^{(i)} - \bm{\phi}^{(j)} 
\left(\bm{\phi}^{(j)^{\top}} \bm{\kappa}^{(i)} \right)$ for $i = j+1, \ldots, j+c-1$.

\noindent
\textit{Step} 10. Go back to the target error check in Step 4.

\noindent
\textit{Step} 11. If $j = 0$, output $\Phi_{j_1,\ldots, j_m} = \left( \bm{\kappa}^{(1)}/\| \bm{\kappa}^{(1)} \| \right)^{\top}$. Otherwise, output $\Phi_{j_1,\ldots, j_m} = \Phi_{j_1,\ldots, j_m}^{(j)}$.
\end{algorithm}

From $U_{j_1,\ldots, j_m}^{(r_{j_1,\ldots, j_m})} U_{j_1,\ldots, j_m}^{{(r_{j_1,\ldots, j_m})}^{\top}} K_{j_1,\ldots,j_m}^m = K_{j_1,\ldots,j_m}^m U_{j_1,\ldots, j_m}^{(r_{j_1,\ldots, j_m})} \left( U_{j_1,\ldots, j_m}^{{(r_{j_1,\ldots, j_m})}^{\top}} K_{j_1,\ldots,j_m}^m U_{j_1,\ldots, j_m}^{(r_{j_1,\ldots, j_m})} \right)^{-1} \\
\times U_{j_1,\ldots, j_m}^{{(r_{j_1,\ldots, j_m})}^{\top}} K_{j_1,\ldots,j_m}^m$, Algorithm \ref{alg:stochastic_matrix_approx} aims to select the appropriate $\Phi_{j_1,\ldots, j_m}$ by diminishing $\| K_{j_1,\ldots,j_m}^m - \Phi_{j_1,\ldots, j_m}^{\top} \Phi_{j_1,\ldots, j_m} K_{j_1,\ldots,j_m}^m \|_F$ for any target error level. 
The projection of Step 5 is introduced in order to ensure better numerical stability 
\citep[see][]{Halko_2011}. In our implementation of Algorithm \ref{alg:stochastic_matrix_approx}, we used $c$ such that $\left| Q_{j_1,\ldots, j_m} \right|/10^c \approx 0.1$.

\section{Inference\label{sec:inference}}

In this section, we propose the two algorithms to conduct fast computation of \eqref{eq:log-lik} and \eqref{eq:pred_dist} where $C_0$ is replaced with $C_{M \text{-RA-lp}}$ defined by Algorithm \ref{alg:approx_cov}.  Consequently, just by using the subsequent two algorithms, we can conduct the likelihood-based inference on the parameters $\Omega$ and obtain the spatial predictive distribution. 
In what follows, it is assumed that $\left| S_{j_1,\ldots,j_M} \right| \ge 1$ for simplicity.

\subsection{Parameter estimation\label{subsec:parameter_estimation}}

The log-likelihood function replaced by the approximated covariance function $C_{M \text{-RA-lp}}$ instead of $C_0$ is given by
\begin{align}
\label{eq:log-lik_m-ra-lp}
- \frac{n}{2} \log(2 \pi) - \frac{1}{2} \log \left[ \mbox{det} \left\{ C_{M \text{-RA-lp}}(S_0,S_0) + \tau^2 \bm{\mbox{I}}_n \right\} \right] \notag \\
- \frac{1}{2} \bm{Z}(S_0)^{\top}
 \left\{ C_{M \text{-RA-lp}}(S_0,S_0) + \tau^2 \bm{\mbox{I}}_n \right\}^{-1} \bm{Z}(S_0).
\end{align}
We will elicit the algorithm to calculate \eqref{eq:log-lik_m-ra-lp} efficiently 
in accordance with the arguments in Sections 3.1--3.3 of \cite{Katzfuss_2017}. For $m = 0,\ldots,M-1$, we define
\begin{align}
\Sigma_{j_1,\ldots,j_m} &= B_{j_1,\ldots,j_m}^m \Phi_{j_1,\ldots,j_m}^{\top} \widehat{K}_{j_1,\ldots,j_m}^m \Phi_{j_1,\ldots,j_m} B_{j_1,\ldots,j_m}^{m^{\top}} + V_{j_1,\ldots,j_m}, \label{eq:def_Sigma_m}\\
%V_{j_1,\ldots,j_m} &= \mbox{blockdiag}\left\{ \Sigma_{j_1,\ldots,j_m,1}, \ldots,\Sigma_{j_1,\ldots,j_m,J_{m+1}} \right\},\label{eq:def_V_m}
V_{j_1,\ldots,j_m} &= 
\begin{pmatrix}
\Sigma_{j_1,\ldots,j_m,1} & & O \\
 & \ddots & \\
 O & & \Sigma_{j_1,\ldots,j_m,J_{m+1}}
\end{pmatrix}
,\label{eq:def_V_m}
\end{align}
where $\Sigma_{j_1,\ldots,j_M}=K_{j_1,\ldots,j_M}^M \circ T_{\gamma}(S_{j_1,\ldots,j_M},S_{j_1,\ldots,j_M}) + \tau^2 \bm{\mbox{I}}_{\left| S_{j_1,\ldots, j_M}\right|}$ and $K_{j_1,\ldots,j_M}^M = C_M(S_{j_1,\ldots,j_M},S_{j_1,\ldots,j_M})\\ = C_M(Q_{j_1,\ldots,j_M},Q_{j_1,\ldots,j_M})$ because $Q_{j_1,\ldots,j_M} = S_{j_1,\ldots,j_M}$. 
From \eqref{eq:cov_mat_m-ra-lp}, \eqref{eq:def_Sigma_m}, and \eqref{eq:def_V_m}, it follows that $\Sigma_0 = C_{M \text{-RA-lp}}(S_0,S_0) + \tau^2 \bm{\mbox{I}}_n$. Moreover, we introduce some comprehensive definitions $K_{j_1,\ldots,j_m}^k = C_k(Q_{j_1,\ldots,j_k},Q_{j_1,\ldots,j_m})$ ($0 \le k \le m,\; m=0,\ldots,M$) and $B_{j_1,\ldots,j_m}^k = C_k(S_{j_1,\ldots,j_m},Q_{j_1,\ldots,j_k})$ ($0 \le k \le m$, $m=0,\ldots,M$).

Next, we describe the algorithm to calculate efficiently the approximated log-likelihood function \eqref{eq:log-lik_m-ra-lp} (see Appendix \ref{subappend:algorithm_log-lik} for the derivation of the algorithm).

\begin{algorithm}[Efficient computation of the approximated log-likelihood function \eqref{eq:log-lik_m-ra-lp}]
\label{alg:log-lik_m-ra-lp}
Given $M > 0$, $D_{j_1,\ldots,j_m}$ ($m = 1,\ldots, M$, $1 \le j_1 \le J_1, \ldots, 1 \le j_M \le J_M$), $Q_{j_1,\ldots,j_m}$ ($m = 0,\ldots, M-1$, $1 \le j_1 \le J_1, \ldots, 1 \le j_{M-1} \le J_{M-1}$), and $\gamma > 0$, find $d_0 = \log \left[ \mbox{det} \left\{ C_{M \text{-RA-lp}}(S_0,S_0) + \tau^2 \bm{\mbox{I}}_n \right\} \right]$ and $u_0 = \bm{Z}(S_0)^{\top} \left\{ C_{M \text{-RA-lp}}(S_0,S_0) + \tau^2 \bm{\mbox{I}}_n \right\}^{-1} \bm{Z}(S_0)$.

\bigskip
\noindent
\textit{Step} 1. For $0 \le k \le m$, $m = 0,\ldots,M$, it follows that
\begin{align}
\label{eq:K_expansion}
K_{j_1,\ldots,j_m}^k =
\begin{cases}
C_0(Q_0,Q_{j_1,\ldots,j_m}), \quad k=0,\\
C_0(Q_{j_1,\ldots,j_k},Q_{j_1,\ldots,j_m}) \\
-\sum_{l=0}^{k-1} K_{j_1,\ldots,j_k}^{l^{\top}} \Phi_{j_1,\ldots,j_l}^{\top}
\widehat{K}_{j_1,\ldots,j_l}^l \Phi_{j_1,\ldots,j_l} K_{j_1,\ldots,j_m}^l, \quad
1 \le k \le m.
\end{cases}
\end{align}
Calculate $K_{j_1,\ldots,j_m}^m$ ($m=0,\ldots,M$), $B_{j_1,\ldots,j_M}^k = K_{j_1,\ldots,j_M}^{k^{\top}}$ ($k=0,\ldots,M-1$), and $\Phi_{j_1,\ldots,j_m}$ ($m=0,\ldots,M-1$) by starting with $K_{j_1,\ldots,j_m}^0$ ($m = 0,\ldots,M$) as the initial matrix.
% as the resolution increases. 
 $\Phi_{j_1,\ldots,j_m}$ is obtained by applying Algorithm \ref{alg:stochastic_matrix_approx} to $K_{j_1,\ldots,j_m}^m$.

\noindent
\textit{Step} 2. Calculate $\widetilde{A}_{j_1,\ldots,j_M}^{k,l} = B_{j_1,\ldots,j_M}^{k^{\top}} \Sigma_{j_1,\ldots,j_M}^{-1} B_{j_1,\ldots,j_M}^l$ ($0 \le k \le l < M$).

\noindent
\textit{Step} 3. For $0 \le k \le l \le m$, $m = 0,\ldots,M-1$, we have
\begin{align}
\label{eq:A_recursive}
A_{j_1,\ldots,j_m}^{k,l} = \sum_{j_{m+1}=1}^{J_{m+1}} \widetilde{A}_{j_1,\ldots,j_m,j_{m+1}}^{k,l}.
\end{align}
Now, we define $\widetilde{K}_{j_1,\ldots,j_m}^m = \left( \widehat{K}_{j_1,\ldots,j_m}^{m^{-1}} + \Phi_{j_1,\ldots,j_m} A_{j_1,\ldots,j_m}^{m,m} \Phi_{j_1,\ldots,j_m}^{\top} \right)^{-1}$ for $m=0,\ldots,M-1$. For $0 \le k \le l < m$, $m = 1,\ldots,M-1$, it follows that
\begin{align}
\label{eq:tilA_recursive}
\widetilde{A}_{j_1,\ldots,j_m}^{k,l} = A_{j_1,\ldots,j_m}^{k,l} - A_{j_1,\ldots,j_m}^{k,m} \Phi_{j_1,\ldots,j_m}^{\top} \widetilde{K}_{j_1,\ldots,j_m}^m \Phi_{j_1,\ldots,j_m} A_{j_1,\ldots,j_m}^{m,l},
\end{align}
where $A_{j_1,\ldots,j_m}^{m,l} = A_{j_1,\ldots,j_m}^{{l,m}^{\top}}$. 
Calculate $A_{j_1,\ldots,j_m}^{k,l}$ ($0 \le k \le l \le m$, $m = M-1,\ldots,0$) by using \eqref{eq:A_recursive} and \eqref{eq:tilA_recursive} alternately from $\widetilde{A}_{j_1,\ldots,j_M}^{k,l}$ of Step 2 as the initial matrix.

\noindent
\textit{Step} 4. Calculate $\widetilde{\omega}_{j_1,\ldots,j_M}^{k} = B_{j_1,\ldots,j_M}^{k^{\top}} \Sigma_{j_1,\ldots,j_M}^{-1} \bm{Z}(S_{j_1,\ldots,j_M})$ ($0 \le k < M$).

\noindent
\textit{Step} 5. For $0 \le k \le m$, $m = 0,\ldots,M-1$, we have
\begin{align}
\label{eq:omega_recursive}
\omega_{j_1,\ldots,j_m}^k = \sum_{j_{m+1}=1}^{J_{m+1}} \widetilde{\omega}_{j_1,\ldots,j_m,j_{m+1}}^k.
\end{align}
For $0 \le k < m$, $m = 1,\ldots,M-1$, it follows that
\begin{align}
\label{eq:tilomega_recursive}
\widetilde{\omega}_{j_1,\ldots,j_m}^k = \omega_{j_1,\ldots,j_m}^k - A_{j_1,\ldots,j_m}^{k,m} \Phi_{j_1,\ldots,j_m}^{\top} \widetilde{K}_{j_1,\ldots,j_m}^m \Phi_{j_1,\ldots,j_m} \omega_{j_1,\ldots,j_m}^m.
\end{align}
Calculate $\omega_{j_1,\ldots,j_m}^m$ ($m = M-1,\ldots,0$) 
by using \eqref{eq:omega_recursive} and \eqref{eq:tilomega_recursive} alternately from 
$\widetilde{\omega}_{j_1,\ldots,j_M}^{k}$ of Step 4 as the initial vector.

\noindent
\textit{Step} 6. Calculate $d_{j_1,\ldots,j_M} = \log \left\{ \mbox{det} \left( \Sigma_{j_1,\ldots,j_M} \right) \right\}$ and $u_{j_1,\ldots,j_M} = \bm{Z}(S_{j_1,\ldots,j_M})^{\top} \Sigma_{j_1,\ldots,j_M}^{-1} \bm{Z}(S_{j_1,\ldots,j_M})$.

\noindent
\textit{Step} 7. For $m = 0,\ldots,M-1$, it follows that
\begin{align}
\label{eq:d_recursive}
d_{j_1,\ldots,j_m} = &-\log \left\{ \mbox{det} \left( \widetilde{K}_{j_1,\ldots,j_m}^m \right) \right\} + \log \left\{ \mbox{det} \left( \widehat{K}_{j_1,\ldots,j_m}^m \right) \right\} + \sum_{j_{m+1}=1}^{J_{m+1}} d_{j_1,\ldots,j_m,j_{m+1}},\\
\label{eq:u_recursive}
u_{j_1,\ldots,j_m} = &- \omega_{j_1,\ldots,j_m}^{m^{\top}} \Phi_{j_1,\ldots,j_m}^{\top} \widetilde{K}_{j_1,\ldots,j_m}^m \Phi_{j_1,\ldots,j_m} \omega_{j_1,\ldots,j_m}^m + \sum_{j_{m+1}=1}^{J_{m+1}} u_{j_1,\ldots,j_m,j_{m+1}}.
\end{align}
Calculate $d_0$ and $u_0$ by using \eqref{eq:d_recursive} and \eqref{eq:u_recursive} recursively from $d_{j_1,\ldots,j_M}$ and $u_{j_1,\ldots,j_M}$ of Step 6 as the initial values, respectively.

\noindent
\textit{Step} 8. Output $d_0$ and $u_0$.
\end{algorithm}

Indeed, \eqref{eq:log-lik_m-ra-lp} is evaluated by using only Algorithm \ref{alg:log-lik_m-ra-lp} without Algorithm \ref{alg:approx_cov}. In Steps 1--6 of Algorithm \ref{alg:log-lik_m-ra-lp}, we calculate the matrices required to obtain $d_0$ and $u_0$ from $d_{j_1,\ldots,j_M}$ and $u_{j_1,\ldots,j_M}$. Note that if $M=1$, then \eqref{eq:tilA_recursive} and \eqref{eq:tilomega_recursive} are not calculated. Also, all of $A_{j_1,\ldots,j_m}^{k,l}$ ($0 \le k \le l \le m$, $m = M-1,\ldots,0$) calculated in Step 3 are not necessarily used in the subsequent steps.

Algorithm \ref{alg:log-lik_m-ra-lp} does not include the inverse of an $n \times n$ matrix. There are the inverse and determinant of the $\left| S_{j_1,\ldots,j_M} \right| \times \left| S_{j_1,\ldots,j_M} \right|$ sparse matrices $\Sigma_{j_1,\ldots,j_M}$ and $r_{j_1,\ldots, j_m} \times r_{j_1,\ldots, j_m}$ matrices ($r_{j_1,\ldots, j_m} \ll n$), and we can calculate them efficiently. 

In order to discuss the operation count and storage of Algorithm \ref{alg:log-lik_m-ra-lp}, we assume for simplicity that $J_i = J$, $\left| Q_{j_1,\ldots,j_m} \right| = r$ ($m=0,\ldots,M-1$), $r_{j_1,\ldots, j_m}=O(r)$ ($m=0,\ldots,M-1$), and $\left| S_{j_1,\ldots,j_M} \right| = n/(J^M) > r$ only in the discussion on the time and memory complexity. When $J$, $M$, and $r$ are large, the main computational efforts are the calculations of $\widetilde{A}_{j_1,\ldots,j_m}^{k,l}$ ($0 \le k \le l < m$, $1 \le j_m \le J$, $m = 1,\ldots,M-1$) and $K^k_{j_1,\ldots,j_m}$ ($0 \le k \le m$, $m=0,\ldots,M-1$, $1 \le j_i \le J$, $i=1,\ldots,M-1$). From $\sum_{m=0}^{M-1} J^m < J^M$, they are $O(J^M M^2 r^3)$ for the operation count. Similarly, the computational burden of Algorithm \ref{alg:stochastic_matrix_approx} also increases because we need to implement Algorithm \ref{alg:stochastic_matrix_approx} $O(J^M)$ times. Algorithm \ref{alg:stochastic_matrix_approx} uses the $r \times r$ matrix $K^m_{j_1,\ldots,j_m}$ ($m=0,\ldots,M-1$) as the input matrix, and its operation count does not depend on $n$. 
%(see Section 4.4 of \citealt{Halko_2011} for computational requirements of Algorithm \ref{alg:stochastic_matrix_approx})
%\citep[see Section 4.4 of][for computational requirements of Algorithm \ref{alg:stochastic_matrix_approx}]{Halko_2011}. 
The fourth simulation of Section \ref{subsec:simulation_study} indicates that we can rapidly implement Algorithm \ref{alg:stochastic_matrix_approx} $O(J^M)$ times by selecting small $r$ and $r_{j_1,\ldots, j_m}$ even for large $J$ and $M$. When $n$ is large, the computational bottleneck of Algorithm \ref{alg:log-lik_m-ra-lp} is in obtaining $K^k_{j_1,\ldots,j_M}$ ($0 \le k \le M$, $1 \le j_i \le J$, $i=1,\ldots,M$) which is $O \left(n^3 M^2 / \left(J^{2M} \right) \right)$ for the operation count. This means that large $J$ and $M$ make the computation of the matrices $K^k_{j_1,\ldots,j_M}$ fast. In addition, calculations related to the inverse of the sparse matrix $\Sigma_{j_1,\ldots,j_M}$ could also be the computational bottleneck if the sparsity of $\Sigma_{j_1,\ldots,j_M}$ is insufficient. It is difficult to evaluate the exact  computational cost of the Cholesky decomposition of the sparse matrix because it depends on the number of non-zero elements and on the ordering of locations. 
However, its resulting time complexity can be less than $O(n^3)$ \citep[see Section 3.3 of][for details]{Furrer_2006}. 
%In Section 3.3 of \cite{Furrer_2006}, there is a detailed discussion on the Cholesky decomposition of the sparse matrix. 
Large $J$ and $M$ lead to the small size of $\Sigma_{j_1,\ldots,j_M}$ and $J^M$ Cholesky decompositions of the sparse matrix $\Sigma_{j_1,\ldots,j_M}$. Through some simulations, we observed that large $J$ and $M$ usually reduced  the total computation time related to the inverse of sparse matrix. 
%Although large $J$ and $M$ lead to the small size of $\Sigma_{j_1,\ldots,j_M}$ and a large number of calculations related to the inverse of the sparse matrix $\Sigma_{j_1,\ldots,j_M}$, we observed that large $J$ and $M$ usually reduced the total computation time related to the inverse of sparse matrix through some simulations. 
The unignorable bottlenecks of the memory consumption of Algorithm \ref{alg:log-lik_m-ra-lp} are $B^k_{j_1,\ldots,j_M}$ ($0 \le k \le M-1$, $1 \le j_i \le J$, $i=1,\ldots,M$) and $K^M_{j_1,\ldots,j_M}$ ($1 \le j_i \le J$, $i=1,\ldots,M$) which are $O(nMr)$ and $O \left(n^2/\left(J^{M} \right) \right)$ for the storage, respectively. Furthermore, the memory complexities including $J^M$ as the product do not depend on $n$. 
Also, the memory consumption in Algorithm \ref{alg:stochastic_matrix_approx} is independent of $n$, and the sparse matrix $\Sigma_{j_1,\ldots,j_M}$ requires at most $O(n^2/\left(J^{2M} \right))$ memory. 
%Similarly, the memory consumption in Algorithm \ref{alg:stochastic_matrix_approx} is also independent of $n$. For the sparse matrix $\Sigma_{j_1,\ldots,j_M}$, it requires at most $O(n^2/\left(J^{2M} \right))$ memory. 
Thus, Algorithm \ref{alg:log-lik_m-ra-lp} can avoid  $O(n^3)$ operations and $O(n^2)$ memory.
% by using large $J$ and $M$.

Finally, Algorithm \ref{alg:log-lik_m-ra-lp} can be parallelized except for Step 1. We introduce a parallel version of Algorithm \ref{alg:log-lik_m-ra-lp} (see Algorithm \ref{alg:log-lik_m-ra-lp_parallel} in Appendix  \ref{subappend:parallel_version_alg_log-lik}).

\subsection{Spatial prediction\label{subsec:spatial_prediction}}

Similar to Section \ref{subsec:parameter_estimation}, we will propose the algorithm for fast computation of the predictive distribution replaced by the approximated covariance function $C_{M \text{-RA-lp}}$ instead of $C_0$. Let $S_{j_1,\ldots,j_M}^P$ denote the set of the unobserved locations on $D_{j_1,\ldots,j_M}$. Additionally, we define $B_{j_1,\ldots,j_M}^{l,P} = C_l (S_{j_1,\ldots,j_M}^P, Q_{j_1,\ldots,j_l})$ ($0 \le l \le M$). The approximated predictive distribution for $\bm{Y}_0(S_{j_1,\ldots,j_M}^P)$ given $\bm{Z}(S_0)$ is the normal distribution with the mean vector
\begin{align}
\label{eq:pred_dist_m-ra-lp_mean}
C_{M \text{-RA-lp}} \left(S_{j_1,\ldots,j_M}^P, S_0 \right)  \left\{ C_{M \text{-RA-lp}}(S_0,S_0) + \tau^2 \bm{\mbox{I}}_n \right\}^{-1} \bm{Z}(S_0)
\end{align}
and covariance matrix
\begin{align}
\label{eq:pred_dist_m-ra-lp_matrix}
C_{M \text{-RA-lp}} \left(S_{j_1,\ldots,j_M}^P, S_{j_1,\ldots,j_M}^P \right) - C_{M \text{-RA-lp}} \left(S_{j_1,\ldots,j_M}^P, S_0 \right) \notag \\
\times \left\{ C_{M \text{-RA-lp}}(S_0,S_0) + \tau^2 \bm{\mbox{I}}_n \right\}^{-1} C_{M \text{-RA-lp}} \left(S_{j_1,\ldots,j_M}^P, S_0 \right)^{\top}.
\end{align}

The following algorithm allows us to calculate efficiently \eqref{eq:pred_dist_m-ra-lp_mean} and \eqref{eq:pred_dist_m-ra-lp_matrix} (see Appendix \ref{subappend:algorithm_pred_dist} for the derivation of the algorithm). Note that the index $(j_1,\ldots,j_M)$ of $S_{j_1,\ldots,j_M}^P$ is fixed.

\begin{algorithm}[Efficient computation of \eqref{eq:pred_dist_m-ra-lp_mean} and \eqref{eq:pred_dist_m-ra-lp_matrix}]
\label{alg:pred_dist_m-ra-lp}
Given $M > 0$, $D_{j^{\prime}_1,\ldots,j^{\prime}_m}$ ($m = 1,\ldots, M$, $1 \le j^{\prime}_1 \le J_1, \ldots, 1 \le j^{\prime}_M \le J_M$), $Q_{j^{\prime}_1,\ldots,j^{\prime}_m}$ ($m = 0,\ldots, M-1$, $1 \le j^{\prime}_1 \le J_1, \ldots, 1 \le j^{\prime}_{M-1} \le J_{M-1}$), $\gamma > 0$, and $S_{j_1,\ldots,j_M}^P$, find $\bm{\mu}_{j_1,\ldots,j_M}^0 = C_{M \text{-RA-lp}}(S_{j_1,\ldots,j_M}^P,S_0) \\
\times \left\{ C_{M \text{-RA-lp}}(S_0,S_0) + \tau^2 \bm{\mbox{I}}_n \right\}^{-1} \bm{Z}(S_0)$  and $\Psi_{j_1,\ldots,j_M}^0 =C_{M \text{-RA-lp}}(S_{j_1,\ldots,j_M}^P,S_{j_1,\ldots,j_M}^P)\\
-C_{M \text{-RA-lp}}(S_{j_1,\ldots,j_M}^P,S_0) \left\{ C_{M \text{-RA-lp}}(S_0,S_0) + \tau^2 \bm{\mbox{I}}_n \right\}^{-1} C_{M \text{-RA-lp}}(S_{j_1,\ldots,j_M}^P,S_0)^{\top}$.

\bigskip
\noindent
\textit{Step} 1. Conduct Steps 1--3 in Algorithm \ref{alg:log-lik_m-ra-lp}. Moreover, if $M \ge 2$, calculate also $K_{j_1,\ldots,j_l}^k$ ($0 \le k \le l-1$, $1 \le l \le M-1$) for the fixed $(j_1,\ldots,j_M)$.

\par
\noindent
\textit{Step} 2. For $0 \le l \le M$, it follows that
\begin{align*}
B_{j_1,\ldots,j_M}^{l,P} =
\begin{cases}
C_0(S_{j_1,\ldots,j_M}^P, Q_0), \quad l=0, \\
C_0(S_{j_1,\ldots,j_M}^P,Q_{j_1,\ldots,j_l})\\
-\sum_{k=0}^{l-1} B_{j_1,\ldots,j_M}^{k, P} \Phi_{j_1,\ldots,j_k}^{\top}
\widehat{K}_{j_1,\ldots,j_k}^k \Phi_{j_1,\ldots,j_k} K_{j_1,\ldots,j_l}^k, \quad 1 \le l \le M.
\end{cases}
\end{align*}
Calculate $B_{j_1,\ldots,j_M}^{l,P}$ ($l=0,\ldots,M$) by starting with $B_{j_1,\ldots,j_M}^{0,P}$ as the initial matrix. Furthermore, calculate 
\begin{align*}
L_{j_1,\ldots,j_M}^M =& B_{j_1,\ldots,j_M}^{M,P} \circ T_{\gamma}(S_{j_1,\ldots,j_M}^P,S_{j_1,\ldots,j_M}),\\
C_M(S_{j_1,\ldots,j_M}^P, S_{j_1,\ldots,j_M}^P) =& C_0(S_{j_1,\ldots,j_M}^P,S_{j_1,\ldots,j_M}^P)\\
&-\sum_{k=0}^{M-1} B_{j_1,\ldots,j_M}^{k, P} \Phi_{j_1,\ldots,j_k}^{\top} \widehat{K}_{j_1,\ldots,j_k}^k \Phi_{j_1,\ldots,j_k} B_{j_1,\ldots,j_M}^{{k, P}^{\top}},\\
V_{j_1,\ldots,j_M}^{M, P} =& C_M(S_{j_1,\ldots,j_M}^P, S_{j_1,\ldots,j_M}^P) \circ T_{\gamma}(S_{j_1,\ldots,j_M}^P,S_{j_1,\ldots,j_M}^P).
\end{align*}

\noindent
\textit{Step} 3. Calculate $\widetilde{B}_{j_1,\ldots,j_M}^{M, k} = B_{j_1,\ldots,j_M}^{k, P} - L_{j_1,\ldots,j_M}^M \Sigma_{j_1,\ldots,j_M}^{-1} B_{j_1,\ldots,j_M}^{k}$ ($0 \le k < M$).

\noindent
\textit{Step} 4. For $0 \le k<l < M$, it follows that
\begin{align}
\label{eq:tilB_recursive}
\widetilde{B}_{j_1,\ldots,j_M}^{l, k} = \widetilde{B}_{j_1,\ldots,j_M}^{l+1, k} - \widetilde{B}_{j_1,\ldots,j_M}^{l+1, l} \Phi_{j_1,\ldots,j_l}^{\top} \widetilde{K}_{j_1,\ldots,j_l}^l \Phi_{j_1,\ldots,j_l} A_{j_1,\ldots,j_l}^{l,k},
\end{align}
where $A_{j_1,\ldots,j_l}^{l,k} = A_{j_1,\ldots,j_l}^{{k,l}^{\top}}$. Calculate $\widetilde{B}_{j_1,\ldots,j_M}^{l, k}$ ($0 \le k<l < M$) by using \eqref{eq:tilB_recursive} recursively from $\widetilde{B}_{j_1,\ldots,j_M}^{M, k}$ in Step 3 as the initial matrix.

\noindent
\textit{Step} 5. Calculate
\begin{align*}
\Psi_{j_1,\ldots,j_M}^0 =& V_{j_1,\ldots,j_M}^{M, P} 
- L_{j_1,\ldots,j_M}^M \Sigma_{j_1,\ldots,j_M}^{-1} L_{j_1,\ldots,j_M}^{M^{\top}}\\ 
&+ \sum_{k=0}^{M-1} \widetilde{B}_{j_1,\ldots,j_M}^{k+1, k} \Phi_{j_1,\ldots,j_k}^{\top} \widetilde{K}_{j_1,\ldots,j_k}^k \Phi_{j_1,\ldots,j_k} \widetilde{B}_{j_1,\ldots,j_M}^{{k+1, k}^{\top}}.
\end{align*}

\noindent
\textit{Step} 6. Conduct the procedure of Steps 4 and 5 in Algorithm \ref{alg:log-lik_m-ra-lp} and calculate $\omega_{j_1,\ldots,j_m}^m$ ($m = M-1,\ldots,0$) for the fixed $(j_1,\ldots,j_M)$.

\noindent
\textit{Step} 7. Calculate
\begin{align*}
\bm{\mu}_{j_1,\ldots,j_M}^0 =& L_{j_1,\ldots,j_M}^M \Sigma_{j_1,\ldots,j_M}^{-1} \bm{Z}(S_{j_1,\ldots,j_M})\\
&+ \sum_{k=0}^{M-1} \widetilde{B}_{j_1,\ldots,j_M}^{k+1, k} \Phi_{j_1,\ldots,j_k}^{\top} \widetilde{K}_{j_1,\ldots,j_k}^k \Phi_{j_1,\ldots,j_k} \omega_{j_1,\ldots,j_k}^k.
\end{align*}

\noindent
\textit{Step} 8. Output $\bm{\mu}_{j_1,\ldots,j_M}^0$ and $\Psi_{j_1,\ldots,j_M}^0$.
\end{algorithm}

Therefore, we can obtain \eqref{eq:pred_dist_m-ra-lp_mean} and \eqref{eq:pred_dist_m-ra-lp_matrix} from only Algorithm \ref{alg:pred_dist_m-ra-lp} without Algorithm \ref{alg:approx_cov}.
In steps except for Steps 5, 7, and 8 of Algorithm \ref{alg:pred_dist_m-ra-lp}, we calculate the matrices required to obtain $\bm{\mu}_{j_1,\ldots,j_M}^0$ and $\Psi_{j_1,\ldots,j_M}^0$. Note that if $M=1$, then \eqref{eq:tilB_recursive} and calculations corresponding to \eqref{eq:tilA_recursive} and \eqref{eq:tilomega_recursive} of Algorithm \ref{alg:log-lik_m-ra-lp} are unnecessary. Similar to Algorithm \ref{alg:log-lik_m-ra-lp}, all of $A_{j_1,\ldots,j_m}^{k,l}$ ($0 \le k \le l \le m$, $m = M-1,\ldots,0$) calculated in Step 1 of Algorithm \ref{alg:pred_dist_m-ra-lp} are not necessarily used in the subsequent steps.

We can efficiently implement Algorithm \ref{alg:pred_dist_m-ra-lp} because this algorithm includes not the inverse of an $n \times n$ matrix but the inverse of the $\left| S_{j_1,\ldots,j_M} \right| \times \left| S_{j_1,\ldots,j_M} \right|$ sparse matrices $\Sigma_{j_1,\ldots,j_M}$ and $r_{j_1,\ldots, j_m} \times r_{j_1,\ldots, j_m}$ matrices ($r_{j_1,\ldots, j_m} \ll n$), similar to Algorithm \ref{alg:log-lik_m-ra-lp}. 

For the operation count and storage of Algorithm \ref{alg:pred_dist_m-ra-lp}, we assume $\left| S_{j_1,\ldots,j_M}^P \right|=O(r)$ as well as the assumptions required in the derivation of the time and memory complexity of Algorithm \ref{alg:log-lik_m-ra-lp}. Note that these assumptions are available only in the discussion on the complexities. By an argument similar to the case of Algorithm \ref{alg:log-lik_m-ra-lp}, for the fixed index $(j_1,\ldots,j_M)$ of $S_{j_1,\ldots,j_M}^P$, we find that Algorithm \ref{alg:pred_dist_m-ra-lp} does not require $O(n^3)$ operations and $O(n^2)$ memory. 
%In this case, for the fixed index $(j_1,\ldots,j_M)$ of $S_{j_1,\ldots,j_M}^P$, the bottlenecks of the computation and memory consumption and their complexities are identical to those in Algorithm \ref{alg:log-lik_m-ra-lp}. 
Consequently, the $M$-RA-lp can handle massive spatial datasets such that the original model is computationally infeasible as shown in the fourth simulation of Section \ref{subsec:simulation_study}.

Similar to Algorithm \ref{alg:log-lik_m-ra-lp}, Algorithm \ref{alg:pred_dist_m-ra-lp} can also be parallelized except for a part of Steps 1 and 2. In particular, Algorithm \ref{alg:pred_dist_m-ra-lp} provides $\bm{\mu}_{j_1,\ldots,j_M}^0$ and $\Psi_{j_1,\ldots,j_M}^0$ only for the fixed index $(j_1,\ldots,j_M)$, while its parallel version can calculate those for any $(j_1,\ldots,j_M)$ in parallel (see Algorithm \ref{alg:pred_dist_m-ra-lp_parallel} in Appendix \ref{subappend:parallel_version_alg_pred_dist}).
%(see Appendix \ref{subappend:parallel_version_alg_pred_dist} for details). 

\subsection{Relationship with the MLP and $M$-RA\label{subsec:relation}}

First, we compare the $M$-RA-lp to the MLP proposed by \cite{Hirano_2017a}. From the viewpoint of the $M$-RA-lp, the MLP carries out the $M$-RA-lp with $M=1$, $Q_0 = S_0$, and $J_1=1$ ($D_0 =D_1$) in Algorithm \ref{alg:approx_cov}. 
In this sense, we regard the $M$-RA-lp as an extension of the MLP. The approximated covariance matrix by the MLP is given by
\begin{align}
\label{eq:cov_mlp}
C_{\text{MLP}}(S_0,S_0) &=  C_{\tau_0}(S_0,S_0) + \left\{ C_0(S_0,S_0) - C_{\tau_0}(S_0,S_0)  \right\} \circ T_{\gamma}(S_0,S_0) \notag \\
&=B_0^0 \Phi_0^{\top} \widehat{K}_{0}^0 \Phi_0 B_0^{{0}^{\top}} + C_{\text{sparse}}(S_0,S_0),
\end{align}
where $C_{\text{sparse}}(S_0,S_0) = \left\{ C_0(S_0,S_0) - C_{\tau_0}(S_0,S_0) \right\} \circ T_{\gamma}(S_0,S_0)$ is the sparse matrix, and $\Phi_0$ is obtained by using Algorithm \ref{alg:stochastic_matrix_approx}. Since $C_0(S_0,S_0)$ in \eqref{eq:log-lik} and \eqref{eq:pred_dist} is replaced with $C_{\text{MLP}}(S_0,S_0)$ in the estimation and prediction by the MLP, the inverse matrix and determinant of $C_{\text{MLP}}(S_0,S_0) + \tau^2 \bm{\mbox{I}}_n$ 
need to be calculated efficiently. As described in \cite{Hirano_2017a}, their calculations are achieved by using
\begin{align*}
&\left\{ C_{\text{MLP}}(S_0,S_0) + \tau^2 \bm{\mbox{I}}_n \right\}^{-1} = \left\{ C_{\text{sparse}}(S_0,S_0)  + \tau^2 \bm{\mbox{I}}_n \right\}^{-1} \\ 
&- \left\{ C_{\text{sparse}}(S_0,S_0)  + \tau^2 \bm{\mbox{I}}_n \right\}^{-1}  B_0^0 \Phi_0^{\top} \left[ \widehat{K}_{0}^{0^{-1}} + \Phi_0 B_0^{{0}^{\top}} \left\{ C_{\text{sparse}}(S_0,S_0)  + \tau^2 \bm{\mbox{I}}_n \right\}^{-1} \right. \\ 
&\times \left. B_0^0 \Phi_0^{\top} \right]^{-1} \Phi_0 B_0^{{0}^{\top}} \left\{ C_{\text{sparse}}(S_0,S_0)  + \tau^2 \bm{\mbox{I}}_n \right\}^{-1} 
\end{align*}
and
\begin{align*}
&\mbox{det} \left\{ C_{\text{MLP}}(S_0,S_0) + \tau^2 \bm{\mbox{I}}_n \right\} = \mbox{det} \left\{ C_{\text{sparse}}(S_0,S_0) + \tau^2 \bm{\mbox{I}}_n \right\} \mbox{det} \left( \widehat{K}_{0}^0 \right) \notag \\
 &\times \mbox{det} \left[ \widehat{K}_{0}^{0^{-1}} + \Phi_0 B_0^{{0}^{\top}} \left\{ C_{\text{sparse}}(S_0,S_0)  + \tau^2 \bm{\mbox{I}}_n \right\}^{-1} B_0^0 \Phi_0^{\top}  \right].
\end{align*}
These expansions are derived from Theorems 18.1.1 and 18.2.8 of \cite{Harville_1997}, and we can treat them rapidly because they contain only the inverse matrix and determinant of the $n \times n$ sparse matrix $C_{\text{sparse}}(S_0,S_0)  + \tau^2 \bm{\mbox{I}}_n$ and $r_0 \times r_0$ matrices. 

In comparison with \eqref{eq:cov_mat_m-ra-lp}, \eqref{eq:cov_mlp} includes only the linear projection term at resolution 0, whereas the $M$-RA-lp has the additional linear projection terms at higher resolutions. Since the linear projection at resolution 0 focuses on fitting the large-scale dependence structure, a modification is required to capture the small-scale spatial variations \citep[see][]{Hirano_2017a}. 
The second term in \eqref{eq:cov_mlp} shows the modification of the linear projection through  the covariance tapering. Although the modification in the MLP is conducted on the whole region $D_0$, this type of the modification in the $M$-RA-lp corresponds to the last term in \eqref{eq:cov_mat_m-ra-lp} and is conducted on each subregion $D_{j_1,\ldots,j_M}$ to elicit Algorithms \ref{alg:log-lik_m-ra-lp} and \ref{alg:pred_dist_m-ra-lp}. However, the overall modification of the $M$-RA-lp by adding the linear projection terms at higher resolutions can  more accurately approximate the small-scale dependence structure than that of the MLP (see Figure \ref{fig:cov_comparison} for details).

Second, we explain the relationship between the $M$-RA-lp and the $M$-RA. In Algorithms \ref{alg:approx_cov}, \ref{alg:log-lik_m-ra-lp}, and \ref{alg:pred_dist_m-ra-lp}, if we set $\Phi_{j_1,\ldots, j_m} = \bm{\mbox{I}}_{\left| Q_{j_1,\ldots, j_m} \right|}$ and $T_{\gamma}(\bm{s}_1,\bm{s}_2) = 1$, the $M$-RA-lp is identical with the $M$-RA. In this sense, the $M$-RA-lp is regarded as an extension of the $M$-RA. Unlike Proposition \ref{prop:m-ra-lp} (d), the approximated covariance function by the $M$-RA equals the original covariance function if the two locations belong to $D_{j_1,\ldots,j_M}$ because $T_{\gamma}(\bm{s}_1,\bm{s}_2) = 1$ \citep[see Section 2.4.5 of][]{Katzfuss_2017}. However, from $\Phi_{j_1,\ldots, j_m} = \bm{\mbox{I}}_{\left| Q_{j_1,\ldots, j_m} \right|}$, it might be necessary to pay attention to the knot selection. 

Furthermore, the introduction of $\Phi_{j_1,\ldots, j_m}$ can yield the stable numerical calculation. In several steps of Algorithms \ref{alg:log-lik_m-ra-lp} and \ref{alg:pred_dist_m-ra-lp}, we conduct the calculation related to the inverse matrix of $\widehat{K}_{j_1,\ldots,j_m}^{{m}^{-1}} = \Phi_{j_1,\ldots,j_m} K_{j_1,\ldots,j_m}^m  \Phi_{j_1,\ldots,j_m}^{\top}$ ($m=0,\ldots,M-1$). If a positive definite matrix is ill-conditioned, the calculation of the inverse matrix may be unstable with the propagation of round-off errors due to the finite precision arithmetic. How well the positive definite  matrix is conditioned can be evaluated by the condition number $\sigma_l/\sigma_s$ which means the ratio of the largest $\sigma_l$ and the smallest $\sigma_s$ eigenvalues of the positive definite matrix \citep[see][]{Dixon_1983}. The condition number closer to 1 indicates better numerical stability. 
The following simulation is similar to the one in Section 3.2 of \cite{Banerjee_2013} and empirically shows the smaller condition number of $\widehat{K}_{j_1,\ldots,j_m}^{{m}^{-1}}$ in the $M$-RA-lp over the $M$-RA.

We consider $M=5$ and the two covariance functions, that is, the exponential covariance function $C_0(\bm{s}_1,\bm{s}_2)=\exp(-5 \| \bm{s}_1 - \bm{s}_2 \|)$ and the Gaussian covariance function $C_0(\bm{s}_1,\bm{s}_2)=\exp(-2.5 (\| \bm{s}_1 - \bm{s}_2 \|/10^2)^2)$, generate locations in $[0,100]^2$ uniformly, and evaluate the average value of the logarithmic transformation of the condition numbers of $\widehat{K}_{1,\ldots,1}^{{m}^{-1}}$ ($m=0,\ldots,4$) for the ten datasets. Each domain is divided into two equal subregions, that is, $J_1=\cdots=J_5=2$. In the $M$-RA-lp, the sizes of $Q_{j_1,\ldots,j_m}$ were 300, 100, 50, 30, and  20 for $m=0,\ldots,4$, respectively, and we selected $\Phi_{j_1,\ldots,j_m}$ such that Algorithm \ref{alg:stochastic_matrix_approx} almost achieved some target values of the rank.

\begin{table}[t]
\caption{Comparative performance of the $M$-RA and $M$-RA-lp with respect to the logarithmic transformation of the condition number in $C_0(\bm{s}_1,\bm{s}_2)=\exp(-5 \| \bm{s}_1 - \bm{s}_2 \|)$.}
\label{tab:comparison_cn_exp}
\centering
\vspace{4pt}
\scalebox{0.77}{
\begin{tabular}{lllccccc} \toprule
$n$ & Approximation & Rank & $\widehat{K}_{0}^{{0}^{-1}}$ & $\widehat{K}_{1}^{{1}^{-1}}$ & $\widehat{K}_{1,1}^{{2}^{-1}}$ & $\widehat{K}_{1,1,1}^{{3}^{-1}}$ & $\widehat{K}_{1,1,1,1}^{{4}^{-1}}$ \\\midrule
5,000 & $M$-RA & 5 & 2.7457$\times 10^{-12}$ & 4.1968$\times 10^{-8}$ & 2.3786$\times 10^{-4}$ & 4.7041$\times 10^{-4}$ & 0.8624 \\
 &  & 10 & 0.0035 & 1.7480$\times 10^{-4}$ & 0.0040 & 1.8454 & 3.4786 \\
 &  & 15 & 2.3175$\times 10^{-4}$ & 0.0048 & 4.5554$\times 10^{-4}$ & 2.8131 & 4.4041 \\
 & $M$-RA-lp & 5 & 4.2507$\times 10^{-4}$ & 0.0031 & 0.0168 & 0.0154 & 0.0691 \\
 &  & 10 & 0.0248 & 0.0760 & 0.0391 & 0.0604 & 0.2241 \\
 &  & 15 & 0.0592 & 0.1244 & 0.1417 & 0.2650 & 0.5392 \\
10,000 & $M$-RA & 5 & 1.0819$\times 10^{-10}$ & 4.0632$\times 10^{-11}$ & 7.8879$\times 10^{-7}$ & 0.0440 & 8.0198$\times 10^{-5}$ \\
 &  & 10 & 0.0033 & 6.6363$\times 10^{-7}$ & 9.3172$\times 10^{-4}$ & 1.8893 & 3.5504 \\
 &  & 15 & 4.5924$\times 10^{-6}$ & 3.2506$\times 10^{-9}$ & 0.0038 & 0.0999 & 7.1757 \\
 & $M$-RA-lp & 5 & 0.0305 & 0.0069 & 0.0380 & 0.0261 & 0.0097 \\
 &  & 10 & 0.0545 & 0.0772 & 0.0366 & 0.0595 & 0.1512 \\
 &  & 15 & 0.1070 & 0.0879 & 0.1072 & 0.1564 & 0.2817 \\\bottomrule
\end{tabular}
}
\end{table}

\begin{table}[t]
\caption{Comparative performance of the $M$-RA and $M$-RA-lp with respect to the logarithmic transformation of the condition number in $C_0(\bm{s}_1,\bm{s}_2)=\exp(-2.5 (\| \bm{s}_1 - \bm{s}_2 \|/10^2)^2)$.}
\label{tab:comparison_cn_gauss}
\centering
\vspace{4pt}
\begin{tabular}{lllccccc} \toprule
$n$ & Approximation & Rank & $\widehat{K}_{0}^{{0}^{-1}}$ & $\widehat{K}_{1}^{{1}^{-1}}$ & $\widehat{K}_{1,1}^{{2}^{-1}}$ & $\widehat{K}_{1,1,1}^{{3}^{-1}}$ & $\widehat{K}_{1,1,1,1}^{{4}^{-1}}$ \\\midrule
5,000 & $M$-RA & 5 & 4.7199 & 7.6432 & 8.2238 & 9.2433 & 9.4813 \\
 &  & 10 & 9.8629 & 12.0785 & 15.9587 & 16.4546 & 14.3746 \\
 &  & 15 & 12.6046 & 16.9184 & 18.7201 & 15.4269 & 16.7445 \\
 & $M$-RA-lp & 5 & 3.2443 & 2.6734 & 2.7241 & 3.4194 & 4.2966 \\
 &  & 10 & 5.8569 & 4.8034 & 6.8916 & 7.2595 & 8.0247 \\
 &  & 15 & 8.2005 & 6.9981 & 7.6940 & 10.2207 & 12.2585 \\
10,000 & $M$-RA & 5 & 5.3286 & 7.3821 & 7.8050 & 9.3527 & 10.6376 \\
 &  & 10 & 10.6290 & 12.4498 & 14.0487 & 14.2829 & 13.4687 \\
 &  & 15 & 14.0862 & 17.2333 & 18.7836 & 15.6605 & 17.5230 \\
 & $M$-RA-lp & 5 & 4.3762 & 3.4216 & 3.5816 & 4.8419 & 4.8598 \\
 &  & 10 & 6.6606 & 5.9712 & 4.0011 & 7.5025 & 7.9609 \\
 &  & 15 & 7.5037 & 7.9213 & 6.6500 & 9.9056 & 11.8865 \\\bottomrule
\end{tabular}
\end{table}

Tables \ref{tab:comparison_cn_exp} and \ref{tab:comparison_cn_gauss} illustrate the comparison of the condition numbers between the $M$-RA and the $M$-RA-lp. Rank in Tables \ref{tab:comparison_cn_exp} and \ref{tab:comparison_cn_gauss} means the target value of the rank of $\Phi_{j_1,\ldots,j_m}$ in the $M$-RA-lp and the size of $Q_{j_1,\ldots,j_m}$ in the $M$-RA. 
From Tables \ref{tab:comparison_cn_exp} and \ref{tab:comparison_cn_gauss}, as the resolution and/or rank increased, the condition number tended to increase. Moreover, the smoothness of $C_0(\bm{s}_1,\bm{s}_2)=\exp(-2.5 (\| \bm{s}_1 - \bm{s}_2 \|/10^2)^2)$ caused the larger condition numbers as a whole. The condition numbers of the $M$-RA-lp were holistically smaller than those of the $M$-RA in similar situations. This may be because the $M$-RA-lp replaces the predictive process in the $M$-RA with the linear projection. Section 3.2 of \cite{Banerjee_2013} empirically showed the smaller condition number of the covariance matrix approximated by the linear projection than that by the predictive process. We also obtained similar results for different types of partitions and covariance functions, but these are not reported here.

Unlike the $M$-RA, the $M$-RA-lp needs to implement Algorithm  \ref{alg:stochastic_matrix_approx} for each subregion except for the one at the highest resolution. 
Although Algorithm \ref{alg:stochastic_matrix_approx} is implemented quickly, large $M$ and $J_i$ ($i=1,\ldots,M-1$) cause the unignorable computational cost due to a large number of implementations of Algorithm \ref{alg:stochastic_matrix_approx}. 
Therefore, we typically select low $M$ and $J_i$ ($i=1,\ldots,M-1$) in the $M$-RA-lp  compared to the $M$-RA, but the size of $\Sigma_{j_1,\ldots,j_M}$ is likely to become large and make it difficult computationally to conduct the calculation related to the inverse of $\Sigma_{j_1,\ldots,j_M}$ in Algorithms \ref{alg:log-lik_m-ra-lp} and \ref{alg:pred_dist_m-ra-lp}. To avoid this problem, we introduce $T_{\gamma}$ in Step 4 of Algorithm \ref{alg:approx_cov} and make $\Sigma_{j_1,\ldots,j_M}$ sparse unlike the $M$-RA. If we need large $M$ and $J_i$ in order to bypass the lack of memory, the total computational time of Algorithm \ref{alg:stochastic_matrix_approx} can be shortened by selecting small $\left| Q_{j_1,\ldots,j_m} \right|$ and $r_{j_1,\ldots, j_m}$.

%\textcolor{red}{
%%Moreover, 
%For large $M$, the overlap between the covariance tapering at the highest resolution and the effect of iterative approximation in the $M$-RA-lp can occur because we observe that  the linear projection at the high resolution approximates the small-scale spatial variations of the original covariance function. 
%By selecting low $M$, we may be able to bypass the redundant overlap.
%}

\begin{figure}[t]
\centering
 \subfigure[MLP (10.49)]{
  \includegraphics[width = 5.1cm,pagebox=artbox,clip]{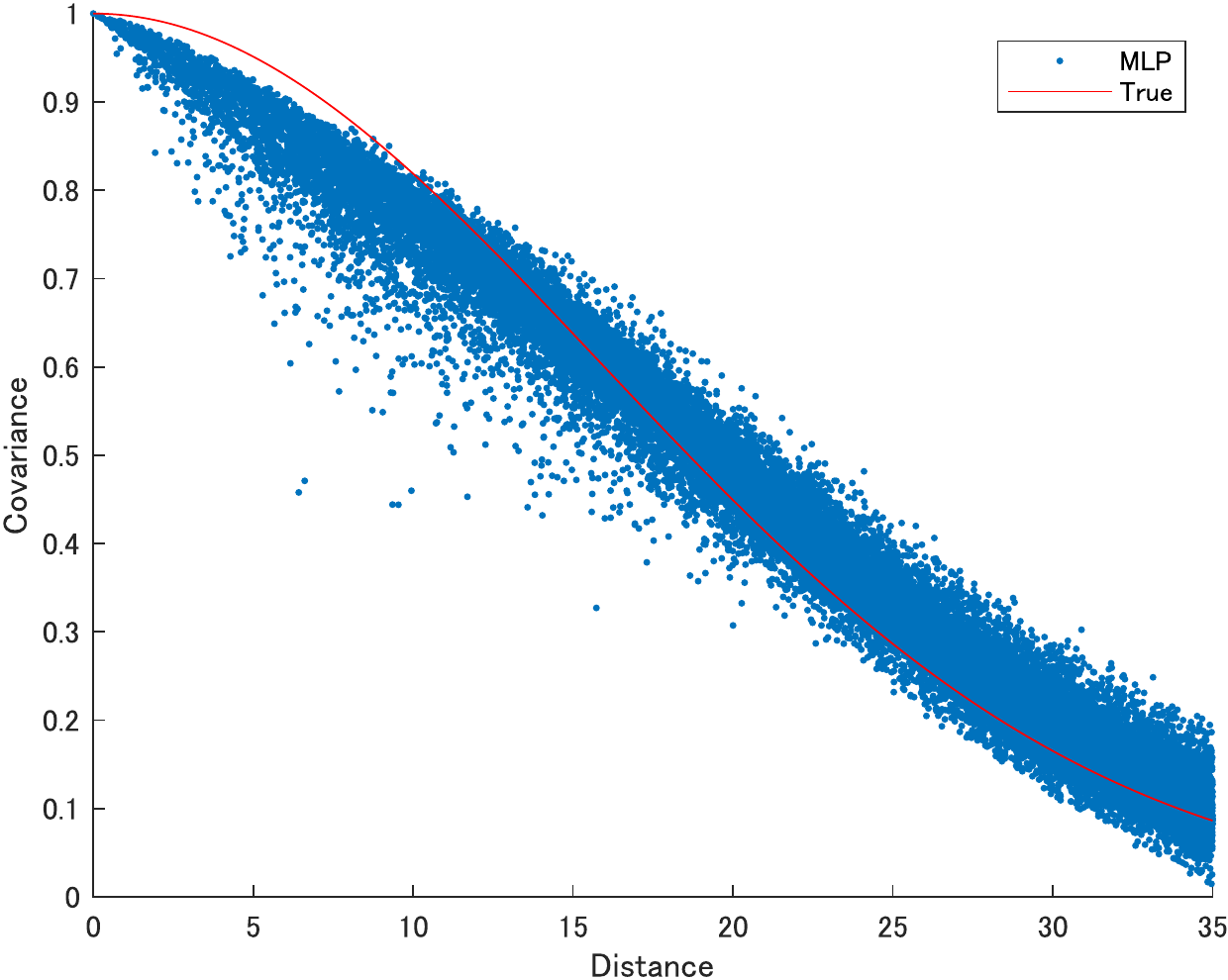}\label{fig:a}}
 \subfigure[$M$-RA (6.65)]{ 
  \includegraphics[width = 5.1cm,pagebox=artbox,clip]{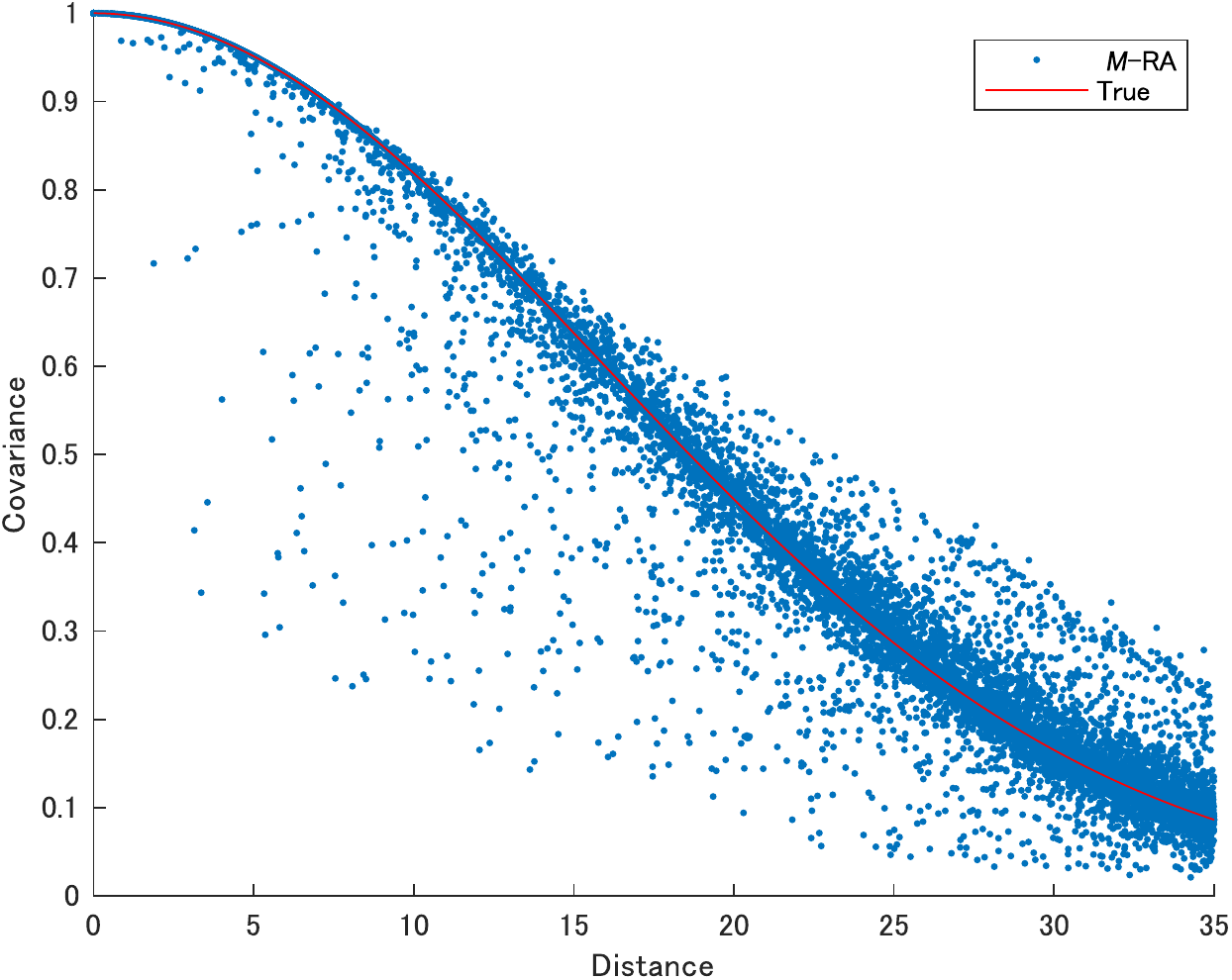}\label{fig:b}}
 \subfigure[$M$-RA-lp (5.37)]{
  \includegraphics[width = 5.1cm,pagebox=artbox,clip]{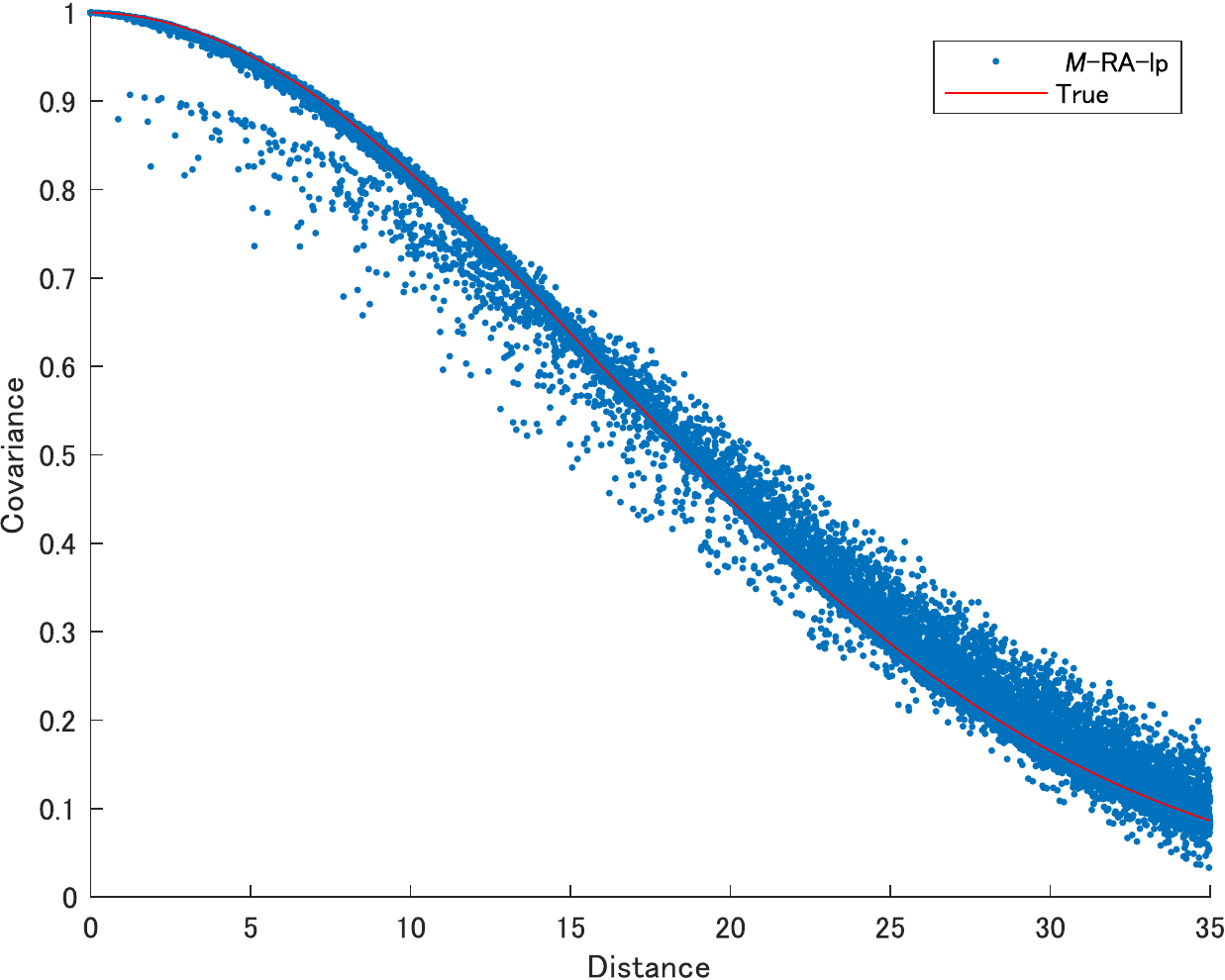}\label{fig:c}}
 \vspace{-4pt}
  \caption{$C_0(\bm{s}_1,\bm{s}_2)=\exp(-0.002 \| \bm{s}_1 - \bm{s}_2 \|^2)$ (solid line) and its three approximations (points). 500 random locations are uniformly generated in $[0,100]^2$. We set $\gamma=10$, $M=2$, and $J_1=J_2=2$, and each domain is divided into equal subregions. The Frobenius norms of the difference between the true covariance matrix and the approximated one are given within parenthesis. (a) MLP. The selected rank of $\Phi_0$ in Algorithm \ref{alg:stochastic_matrix_approx} was 15. (b) $M$-RA with $\left| Q_{j} \right|=16$ ($j=0,1,2$). 
(c) $M$-RA-lp with $Q_{j} = S_{j}$ ($j=0,1,2$). The selected ranks of $\Phi_j$ ($j=0,1,2$) in Algorithm \ref{alg:stochastic_matrix_approx} were 16, 14, and 15. 
  }
\label{fig:cov_comparison} 
\end{figure}

Figure \ref{fig:cov_comparison} describes the typical characteristics of the three approximation methods for the original covariance function. 
The fitting of the MLP to the original covariance function was worse around the origin than that of the $M$-RA-lp. This suggests that the modification by the covariance tapering on the whole spatial region can be insufficient unlike the linear projection at higher resolutions plus the modification by the covariance tapering on each subregion when the small-scale spatial correlation has smoothness such as the Gaussian covariance function. The $M$-RA-lp showed the best approximation accuracy with regard to the Frobenius norm. However, Figure \ref{fig:cov_comparison} (c) 
represents the partly mismatched fitting of the $M$-RA-lp around the origin because the approximation procedure of Algorithm \ref{alg:approx_cov} stops in the early stages if the two locations are not in the same subregion at the low resolution. This problem can also occur in the $M$-RA. A taper version of the $M$-RA-lp based on \cite{Katzfuss_2019} may resolve this artificiality.

\section{Illustrations\label{sec:illustrations}}

In this section, we compare our proposed $M$-RA-lp with the MLP and $M$-RA by using the simulated and real data. All computations were carried out by using MATLAB on a single core machine (4.20 GHz) with 64 GB RAM. For sparse matrix calculations and the optimization of the log-likelihood function, we used the MATLAB functions {\tt sparse} and {\tt fmincon}, respectively.

\subsection{Simulation study\label{subsec:simulation_study}}

We evaluated the performance of Algorithms \ref{alg:log-lik_m-ra-lp} and \ref{alg:pred_dist_m-ra-lp} in our proposed $M$-RA-lp through simulation studies. Let $D_0 = [0, 100]^2$ be the sampling domain, and observed locations were sampled from a uniform distribution over $D_0$. We considered the zero-mean Gaussian processes having the same covariance functions as those in simulations of Tables \ref{tab:comparison_cn_exp} and \ref{tab:comparison_cn_gauss} of Section \ref{subsec:relation}. When pairs of observations were more than 0.6 unit distant from each other in the case of the exponential covariance function, they had negligible ($< 0.05$) correlation. This distance is called the effective range, and 0.6 unit represents the Gaussian process with the weak spatial correlation. In contrast, the effective range in the case of the Gaussian covariance function was 110 unit, and it shows the strong spatial correlation. 
The measurement error variance $\tau^2$ was 0.5. In this subsection, we set $\left| Q_{j_1,\ldots,j_m} \right|=r$ ($m=0,\ldots,M-1$) in the $M$-RA and $J_{i}=J$ ($i=1,\ldots,M$) except for the fourth simulation, and the equal-area partitions were chosen when the resolution increased. 
The average value of total computational times required for one calculation of the evaluation measure in each iteration was recorded and scaled relative to that of the original model  except for the fourth simulation.

First, we compared the approximation accuracy of the original log-likelihood function in Algorithm \ref{alg:log-lik_m-ra-lp} of the $M$-RA-lp with those of the MLP and $M$-RA. All comparisons were conducted based on the log-score which is defined by the log-likelihood function at the true parameter values. The log-score indicates how well the original covariance function is approximated. Since this measure is maximized in the sense of the expectation by the original model \citep[see, e.g.,][]{Gneiting_2014}, the log-score by the original covariance function is expected to have the highest value on average. Thus, the higher log-score is better.

For a fixed configuration of 10,000 sampling locations, we calculated the sample mean of the log-scores of 50 simulations. This procedure was iterated 10 times, and we recorded their average value. For the $M$-RA-lp with $\gamma=1$ and $M=2$, we set $|Q_0|=300$, $|Q_1|=|Q_2|=100$, and $J=2$. In Algorithm \ref{alg:stochastic_matrix_approx}, we selected each $\varepsilon$ such that all of $r_{j}$ ($j=0,1,2$) over the iterations were nearly equal to the target values of 5, 10, and 20. For the MLP with $\gamma=1$, the target values of $r_0$ were 10, 20, and 40. For the $M$-RA with $r=10$ and $M=2, 4, 5$, we selected $J$ that almost satisfies  $r=n \big/ \left(J^M \right)$. This selection guideline $r=n \big/ \left(J^M \right)$ was often used in simulation studies of \cite{Katzfuss_2017}.

\begin{table}[t]
\caption{Comparison of the log-score.}
\label{tab:log-score}
\centering
\vspace{4pt}
\scalebox{0.86}{
\begin{tabular}{llcccccc} \toprule
 & Covariance &  & Original model &  &  & $M$-RA-lp &  \\ \cmidrule(l{5pt}r{5pt}){6-8}
 &  &  &  &  & 5 & 10 & 20 \\ \midrule
log-score ($\times 10^4$) & Exponential &  & -1.6148 &  & -1.6159 & -1.6155 & -1.6153 \\
 & Gaussian &  & -1.0763  &  & -1.0777 & -1.0769  & -1.0764 \\
Relative time &  &  & 1 &  & 0.2609  & 0.2640  & 0.2689  \\ \cmidrule[1pt]{1-8}
 & Covariance &  & MLP &  &  & $M$-RA &  \\ \cmidrule(l{5pt}r{5pt}){3-5} \cmidrule(l{5pt}r{5pt}){6-8}
 &  & 10 & 20 & 40 & M=2 & M=4 & M=5 \\ \midrule
log-score ($\times 10^4$) & Exponential & -1.6156 & -1.6154 & -1.6151 & -1.6162 & -1.6157 & -1.6155 \\
 & Gaussian & -1.0773  & -1.0770  & -1.0765  & -1.0799  & -1.0774  & -1.0770  \\
Relative time &  & 0.5142  & 0.6813  & 0.9658  & 0.2318  & 0.2620  & 0.2824  \\ \bottomrule
\end{tabular}
}

\vspace{4pt}
\hspace{-2pt}
Three cases in the $M$-RA-lp and MLP represent target values of $r_j$'s and $r_0$, respectively.
\end{table}

The results are summarized in Table \ref{tab:log-score}. We compared the three approximation methods on the basis of the $M$-RA-lp with $r_{j}$ ($j=0,1,2$) nearly equal to 10. The comparison of these methods showed common characteristics in both covariance functions. The MLP eventually indicated a similar log-score and larger computational time in $r_0 \approx 20$. 
Unlike the case of the Gaussian covariance function, the log-score of the MLP in the case of the exponential covariance function was better slightly in the sense of the magnitude relationship with those of the $M$-RA-lp and $M$-RA. This is because the exponential covariance function in this simulation has the weak spatial correlation and the modification by the covariance tapering works well. 
%The log-score of the MLP in the case of the exponential covariance function was better slightly in the sense of the magnitude relationship with those of the $M$-RA-lp and $M$-RA because the modification by the covariance tapering works well. 
For the $M$-RA with $M=2$, although the computational time was lower, the log-score was not good. In the case of $M=5$, the log-score was similar to that of the $M$-RA-lp, but the computational time was somewhat large. Also, in other cases, the log-scores and computational times of the $M$-RA-lp were not improved simultaneously compared with those of the MLP and $M$-RA. These results support the effectiveness of the $M$-RA-lp for efficiently approximating the log-likelihood function. 

Second, we assessed the prediction performance of Algorithm \ref{alg:pred_dist_m-ra-lp} with regard to the mean squared prediction error (MSPE) and the continuous ranked probability score (CRPS). The CRPS evaluates the fitting of the predictive distribution to the data \citep[see][]{Gneiting_2007,Gneiting_2014}. The lower MSPE and CRPS are better. The prediction point was $(1,1)^{\top}$. The tuning parameter settings in the three approximation methods and the iteration procedure were the same as those in the first simulation except that the MSPE and averaged CRPS were calculated from 100 simulations.

\begin{table}[t]
\caption{Comparison of the MSPE and CRPS.}
\label{tab:mspe_crps}
\centering
\vspace{4pt}
\scalebox{0.9}{
\begin{tabular}{llcccccc} \toprule
 & Covariance &  & Original model &  &  & $M$-RA-lp &  \\ \cmidrule(l{5pt}r{5pt}){6-8}
 &  &  &  &  & 5 & 10 & 20 \\ \midrule
MSPE & Exponential &  & 0.9300  &  & 0.9931 & 0.9771  & 0.9660 \\
 & Gaussian &  & 0.0056  &  & 0.0094 & 0.0062  & 0.0058 \\
CRPS & Exponential &  & 0.5351 &  & 0.5609 & 0.5529 & 0.5478  \\
 & Gaussian &  & 0.0423 &  & 0.0542 & 0.0447 & 0.0429  \\
Relative time &  &  & 1 &  & 0.4830  & 0.4879  & 0.4952  \\ \cmidrule[1pt]{1-8}
 & Covariance &  & MLP &  &  & $M$-RA &  \\ \cmidrule(l{5pt}r{5pt}){3-5} \cmidrule(l{5pt}r{5pt}){6-8}
 &  & 10 & 20 & 40 & M=2 & M=4 & M=5 \\ \midrule
MSPE & Exponential & 0.9694  & 0.9501 & 0.9385  & 1.0572 & 0.9969  & 0.9728 \\
 & Gaussian & 0.0139  & 0.0064 & 0.0059  & 0.0162 & 0.0074  & 0.0067 \\
CRPS & Exponential & 0.5570 & 0.5479  & 0.5417  & 0.5814  & 0.5642 & 0.5570 \\
 & Gaussian & 0.0656 & 0.0452  & 0.0433  & 0.0716  & 0.0486 & 0.0464 \\
Relative time &  & 0.7656  & 0.8936  & 1.2989  & 0.4403  & 0.5035  & 0.5339  \\ \bottomrule
\end{tabular}
}

\vspace{4pt}
\hspace{-2pt}
Three cases in the $M$-RA-lp and MLP represent target values of $r_j$'s and $r_0$, respectively.
\end{table}

The characteristics of the results in Table \ref{tab:mspe_crps} were similar to those of the first simulation. 
The weak spatial correlation in the case of the exponential covariance function gave rise to the better prediction accuracy of the MLP than that of the $M$-RA-lp in many cases 
because of the effective modification of the covariance tapering. However, the computational time of the MLP is larger than that of the $M$-RA-lp, and the prediction accuracy of the MLP in the case of the Gaussian covariance function degrades due to the strong spatial correlation unlike the $M$-RA-lp. 
Moreover, for the Gaussian covariance function, the $M$-RA-lp with $r_{j}$ ($j = 0,1,2$) nearly equal to 20 showed almost the same MSPE and CRPS as those of the original model despite half the computational time. These results demonstrate that the $M$-RA-lp can achieve better prediction accuracy rapidly than the MLP and $M$-RA. 

Through the first and second simulations in Section \ref{subsec:simulation_study}, we examined the effect of the covariance tapering in $\Sigma_{j_1,\ldots,j_M}$. The averaged percentage of non-zero entries in $\Sigma_{j_1,\ldots,j_M}$, that is, the averaged sparsity of $\Sigma_{j_1,\ldots,j_M}$, was 0.16\%. In the case where the covariance tapering was not used in $\Sigma_{j_1,\ldots,j_M}$ of the $M$-RA-lp, the relative computational times in the first simulation were 0.3123, 0.3154, and 0.3193 for $r_{j}$ ($j=0,1,2$) nearly equal to 5, 10, and 20, respectively. Similarly, the ones in the second simulation were 0.5783, 0.5832, and 0.5891 for $r_{j}$ ($j=0,1,2$) nearly equal to 5, 10, and 20, respectively. Thus, the covariance tapering reduced the computational time by approximately 19\% in the first and second simulations. On the other hand, in the case of the exponential covariance function, the relative Frobenius norms between the original covariance matrix and the approximated covariance matrix by the $M$-RA-lp without the covariance tapering, which were scaled relative to $\| C_{0}(S_0,S_0)-C_{M \text{-RA-lp}}(S_0,S_0) \|_F$, were 0.7748, 0.7989, and 0.8100 for $r_{j}$ ($j=0,1,2$) nearly equal to 5, 10, and 20, respectively. As $r_{j}$'s increase, the $M$-RA-lp improves the approximation of the small-scale spatial variations of the original covariance function. Consequently, since the effect of the covariance tapering decreases, the relative Frobenius norm is closer to 1. 
%Note that this Frobenius norm was scaled relative to the one between the original covariance matrix and the covariance matrix by the $M$-RA-lp. 
Considering that the averaged sparsity was 0.16\%, the reduction of the approximation accuracy for the original covariance matrix by using the covariance tapering was small. This is because 
the spatial correlation in this case was weak 
%the effective range in this case was small 
and the covariance tapering worked well. In the case of the Gaussian covariance function, the relative Frobenius norms were 0.9989, 0.9998, and 0.9998 for $r_{j}$ ($j=0,1,2$) nearly equal to 5, 10, and 20, respectively. 
%In this case, since the approximated covariance function up to  resolution $m=1$ in Algorithm 1 with $M=2$ approximated the small-scale spatial variations of the original covariance function, the effect of the covariance tapering is small.
In this case, since the small-scale spatial variations of the original covariance function are well approximated by the $M$-RA-lp up to resolution $m=1$ in Algorithm 1 with $M=2$, the reduction of the approximation accuracy by the covariance tapering was very small. 
%In this case, since the approximation by the $M$-RA-lp up to resolution $m=1$ in Algorithm \ref{alg:approx_cov} with $M=2$ extended to the small-scale spatial variations of the original covariance function, the reduction of the approximation accuracy by the covariance tapering was very small. 
%As a consequence, in the first and second simulations, the covariance tapering reduced the computational time efficiently.
As a consequence, it is suggested that the covariance tapering in the $M$-RA-lp reduced the computational time efficiently in the first and second simulations. 

Third, we investigated scalability of the $M$-RA-lp. 
The sample size $n$ was selected from 5,000 to 20,000, and the count of iterations for calculating the averaged total computational time of one log-score, MSPE, and CRPS by Algorithms  \ref{alg:log-lik_m-ra-lp} and \ref{alg:pred_dist_m-ra-lp} was 3. However, for $n=10,000$, we used the summation of the computational times in Tables \ref{tab:log-score} and \ref{tab:mspe_crps}. We employed tuning parameter settings under which the three approximation methods showed almost the same prediction accuracy for $n=10,000$ in the second simulation. Specifically, all of $r_{j}$ ($j = 0,1,2$) of the $M$-RA-lp were nearly equal to 10, and $r_0$ of the MLP was almost 20. For the $M$-RA, we set $r=10$, $M=5$, and $J=4$ for $n=10,000$ and selected $M$ to almost satisfy $r=n \big/ \left(J^M \right)$ as $r=10$ and $J=4$ for different values of $n$. 

\begin{table}[t]
\caption{Comparison of the computational time.}
\label{tab:comp_time}
\vspace{4pt}
\centering
\begin{tabular}{lcccc} \toprule
 & $n=5,000$ & $n=10,000$ & $n=15,000$ & $n=20,000$ \\ \midrule
Original model & 1 & 1 & 1 & 1 \\
$M$-RA-lp & 0.5558 & 0.3669 & 0.3628 & 0.3236 \\
MLP & 1.1347 & 0.7788 & 0.6870 & 0.6722 \\
$M$-RA & 0.6275 & 0.3979 & 0.3512 & 0.3003 \\ \bottomrule
\end{tabular}
\end{table}

Table \ref{tab:comp_time} displays better scalability of the $M$-RA-lp and $M$-RA than that of the MLP, and the $M$-RA showed a shorter computational time than that of the $M$-RA-lp for very large $n$. This is because the $M$-RA-lp requires the additional computational time by Algorithm \ref{alg:stochastic_matrix_approx} and matrix multiplication related to $\Phi_{j}$ ($j=0,1,2$). The third simulation indicates that the $M$-RA has competitive scalability. However, from the results in Tables \ref{tab:comparison_cn_exp}--\ref{tab:mspe_crps}, we believe that the $M$-RA-lp can attain a better and stable inference by the somewhat additional computational time.

Fourth, we examined computational feasibility of the $M$-RA-lp when $n=10^5$ or more. These kinds of massive spatial datasets are often obtained by sensing devices on satellites. 
For $n=100,000$, $120,000$, we recorded the averaged total computational time of one log-score, MSPE, and CRPS by Algorithms \ref{alg:log-lik_m-ra-lp} and \ref{alg:pred_dist_m-ra-lp} from the three iterations. In this case, it was necessary to pay attention to the memory burden as well as the expensive computational cost. Since the original model and MLP require the $n \times n$ covariance matrix, we experienced the lack of memory. 
%In this case, we needed to pay attention to the memory burden as well as the expensive computational cost. 
Similarly, the $M$-RA-lp with 
$M=2$, $(J_1, J_2) = (2, 2)$, $(|Q_0|,|Q_{j_1}|) = (300,100)$ ($1 \le j_1 \le 2$), $\gamma=1$ 
%$M=2$ and $J_1=J_2=2$ 
used in the third simulation
% of Section \ref{subsec:simulation_study} 
 also caused the lack of memory due to large $\left| S_{j_1,\ldots,j_M} \right|$. Hence, we needed to increase $M$ and/or $J_i$ in order to reduce the size of $S_{j_1,\ldots,j_M}$. Specifically, for the $M$-RA-lp, we considered two cases: $M=2$, $(J_1, J_2) = (2, 16)$, $(|Q_0|,|Q_{j_1}|) = (300,100)$ ($1 \le j_1 \le 2$) and $M=4$, $(J_1, J_2, J_3, J_4)= (2, 2, 8, 16)$, $(|Q_0|,|Q_{j_1}|,|Q_{j_1,j_2}|,|Q_{j_1,j_2,j_3}|) = (300,100,50,30)$ ($1 \le j_1 \le 2, \; 1 \le j_2 \le 2,\; 1 \le j_3 \le 8$). For Cases 1 and 2, the target values of $r_{j_1,\ldots, j_m}$'s were 10. 
%For both Cases 1 and 2, the target values of $r_0$, $r_{j_1}$, $r_{j_1, j_2}$, and $r_{j_1, j_2, j_3}$ were 10. 
%For both Cases 1 and 2, the target values of the ranks of $\Phi_{j_1,\ldots, j_m}$ were 10. 
Furthermore, we set the computational time of the $M$-RA with $M=5$ and $J_i=J=4$ ($i=1,\ldots,5$) used in the third simulation
% of Section \ref{subsec:simulation_study} 
 as the baseline for calculating the relative time. For the $M$-RA, $r$ was selected from $r=n \big/ \left(J^M \right)$.

\begin{table}[t]
\caption{Computational time for massive spatial datasets.}
\label{tab:comp_time_massive}
\centering
\begin{tabular}{lccc} \toprule
\multicolumn{1}{c}{} & $M$-RA & $M$-RA-lp (Case 1) & $M$-RA-lp (Case 2) \\ \midrule
$n=100,000$ & 1 & 1.0598 & 0.9327 \\
$n=120,000$ & 1 & 1.1047 & 0.9150 \\ \bottomrule
\end{tabular}

\vspace{4pt}
\raggedright
\hspace{40pt} Case 1: $M=2$, $(J_1, J_2) = (2, 16)$, $(|Q_0|,|Q_{j_1}|) = (300,100)$; Case 2:\\
\hspace{40pt} $M=4$, $(J_1, J_2, J_3, J_4)= (2, 2, 8, 16)$, $(|Q_0|,|Q_{j_1}|,|Q_{j_1,j_2}|,|Q_{j_1,j_2,j_3}|)$\\
\hspace{40pt} $= (300,100,50,30)$.
\end{table}

Table \ref{tab:comp_time_massive} shows the computational time for each method and suggests that it may be desirable 
to make both $M$ and $J_i$ large 
%to choose large $M$ and $J_i$ 
for the $M$-RA-lp in terms of the computational cost. Since the $M$-RA-lp with large $M$ and/or $J_i$ requires a large number of implementations of Algorithm \ref{alg:stochastic_matrix_approx}, we selected relatively low $\left| Q_{j_1,\ldots, j_m} \right|$ and $r_{j_1,\ldots, j_m}$ in spite of massive spatial datasets. This might lead to insufficient approximation of the small-scale spatial variation when the variation of the spatial correlation around the origin is smooth. Since the $M$-RA does not conduct Algorithm \ref{alg:stochastic_matrix_approx}, we can select large $M$, $J$, and $r$. As a consequence, the $M$-RA can likely avoid this problem, but Tables \ref{tab:comparison_cn_exp} and \ref{tab:comparison_cn_gauss} indicate that numerical instability might occur unlike the $M$-RA-lp.

\subsection{Real data analysis\label{subsec:real_data}}

In this subsection, we applied our proposed $M$-RA-lp to the air dose rates, that is, the amount of radiation per unit time in the air. The data were created by using the results of the vehicle-borne survey conducted by the Nuclear Regulation Authority (NRA) from November 2 to December 18, 2015, and are available at \url{https://emdb.jaea.go.jp/emdb/en/portals/b1010202/}. In particular, we focused on the air dose rates in Chiba prefecture, and this dataset includes the air dose rates (microsievert per hour), longitudes, and  latitudes at 39,553 sampling locations. Since they were observed on irregularly spaced locations at discrete time points, they were rigorously spatio-temporal data. However, we regarded the dataset as spatial data by assuming that the trend of the air dose rates does not fluctuate largely over a short period. Moreover, to satisfy the assumption of Gaussianity over the whole region, we selected 7,801 locations inside the rectangular region $[140.00, 140.15] \times [35.65, 35.87]$ and applied the logarithmic transformation to these air dose rates. 
Figure \ref{fig:log_original_data} shows an overview of the transformed data. 
After subtracting the sample mean from the transformed data, some exploratory analyses led us to use the zero-mean Gaussian process with an exponential covariance function $C_0(\bm{s}_1,\bm{s}_2)=\sigma^2 \exp(- \theta \| \bm{s}_1 - \bm{s}_2 \|)$.  Also, to check the predictive performance, we considered a test set of size 129 from 7,801 data points. The test set belonged to $[140.07, 140.08] \times [35.71, 35.74]$, and we designed the domain partitioning such that $[140.07, 140.08] \times [35.71, 35.74]$ was inside the subregion at the highest resolution.

\begin{figure}[t]
\centering
  \includegraphics[width = 9.2cm,pagebox=artbox,clip]{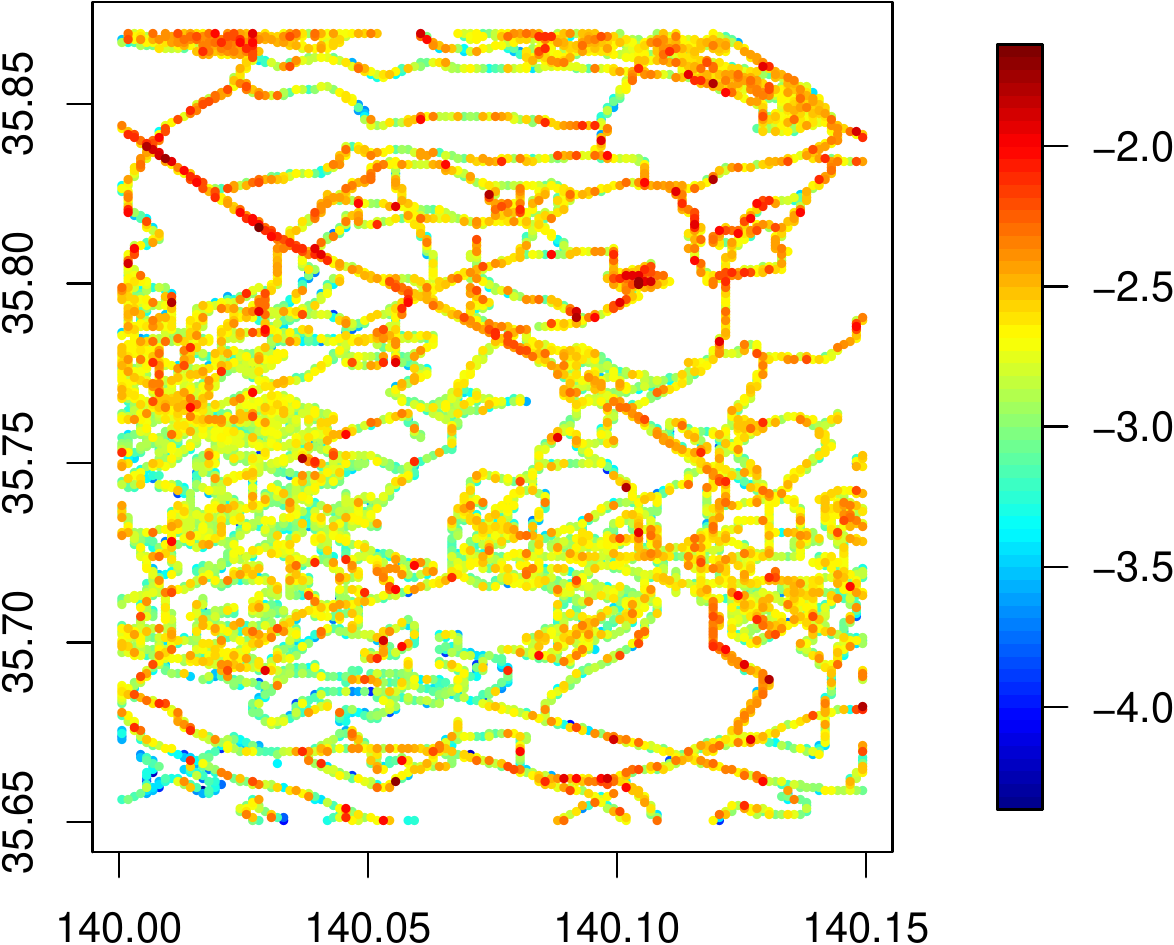}
  \caption{The logarithmic transformation of the air dose rates at 7,801 sampling locations in Chiba prefecture.}
\label{fig:log_original_data} 
\end{figure}

First, we estimated the unknown parameters $\sigma^2$, $\theta$, and $\tau^2$ by maximizing the approximated log-likelihood functions of the $M$-RA-lp, MLP, and $M$-RA. Then, we calculated the predictive distribution of the test set to compare the three methods by assessing the MSPE and averaged CRPS. In order to calculate the two prediction measures, we adopted the predictive distribution $\bm{Z} \left(S_0^P \right) |\bm{Z}(S_0)$ instead of \eqref{eq:pred_dist} because $\bm{Y}_0 \left(S_0^P \right)$ is not observed. $\bm{Z} \left(S_0^P \right) |\bm{Z}(S_0)$ is given just by adding $\tau^2 \bm{\mbox{I}}_{n'}$ to the covariance matrix in \eqref{eq:pred_dist}.

We compared the $M$-RA-lp with $M=2$ to the MLP and $M$-RA with $M=4$. $\gamma$ in $T_{\gamma}$ was 1. For the $M$-RA-lp, we set $(J_1, J_2)=(2, 2)$, $|Q_0|=140$, and $|Q_1|=|Q_2|=60$. We implemented Algorithm \ref{alg:stochastic_matrix_approx} such that $r_0$ and $r_{j_1}$ ($j_1=1,2$) were nearly equal to the target values 20 and 10, respectively. In the same way, the target value of $r_0$ in the MLP was set as 20. The $M$-RA had $(J_1, J_2, J_3, J_4) = (2, 2, 8, 16)$ because a few partitions at low resolution often improve the approximation of the original covariance function by avoiding the early stop of Algorithm \ref{alg:approx_cov}. For the number of knots, we considered two cases: $(|Q_0|,|Q_{j_1}|,|Q_{j_1,j_2}|,|Q_{j_1,j_2,j_3}|) = (20,10,5,5), (20,10,800,70)$ ($1 \le j_1 \le 2, \; 1 \le j_2 \le 2,\; 1 \le j_3 \le 8$). 

\begin{table}[t]
\caption{Results of the real data analysis.}
\label{tab:real_data_analysis}
\vspace{4pt}
\scalebox{0.95}{
\begin{tabular}{lccccccc} \toprule
 & $\hat{\sigma}^2$ & $\hat{\theta}$ & $\hat{\tau}^2$ & \begin{tabular}{c} Relative time \end{tabular} & loglik. & MSPE & CRPS \\ \midrule
Original model & 0.0416 & 1.6109 & 0.0581 & 1.0000  & -1051.7 & 0.0690  & 0.1464 \\
$M$-RA-lp & 0.0429 & 1.6073 & 0.0609 & 0.6952  & -1067.5 & 0.0691 & 0.1471 \\
MLP & 0.0445 & 1.6268 & 0.0598 & 1.6422  & -1101.5 & 0.0701 & 0.1477 \\
$M$-RA (Case 1) 
& 0.0409 & 0.7126 & 0.0635 & 0.2740  & -1119.1 & 0.0700  & 0.1499 \\
$M$-RA (Case 2) 
 & 0.0424 & 1.6402 & 0.0574 & 0.8952  & -1061.8 & 0.0691 & 0.1480  \\ \bottomrule
\end{tabular}
}

\vspace{4pt}
Case 1: $(|Q_0|,|Q_{j_1}|,|Q_{j_1,j_2}|,|Q_{j_1,j_2,j_3}|) = (20,10,5,5)$;~ Case~ 2:~ $(|Q_0|,|Q_{j_1}|,|Q_{j_1,j_2}|,|Q_{j_1,j_2,j_3}|)\\= (20,10,800,70)$; Relative time: relative time per likelihood function evaluation; loglik.: maximum log-likelihood value.
\end{table}

The results of the real data analysis are shown in Table \ref{tab:real_data_analysis}. 
The MLP and the $M$-RA (Case 1) showed the discrepancy from the results of the original model in terms of the maximum log-likelihood values and prediction measures. In particular, the computational time of the MLP was larger than that of the original model due to Algorithm \ref{alg:stochastic_matrix_approx}. Additionally, the estimated value of $\theta$ in the $M$-RA (Case 1) was very small because the approximation of the small-scale spatial variation is insufficient due to the small $|Q_{j_1,j_2}|$ and $|Q_{j_1,j_2,j_3}|$. Although the $M$-RA-lp had almost the same $r_0$ and $r_{j_1}$ as the corresponding rank and number of knots in the MLP and $M$-RA (Case 1), the $M$-RA-lp achieved results similar to the original model. Furthermore, the computational time of the $M$-RA-lp was smaller than that of the original model. By increasing the number of knots at higher resolutions in order to capture the small-scale spatial variation, the $M$-RA (Case 2) showed results close to the original model, while the $M$-RA-lp attained similar results in the shorter computational time. 
Therefore, it is suggested that the $M$-RA-lp can more rapidly realize results close to the  original model compared with the $M$-RA. 
%However, note that another valid statistical model model should be used to predict correctly the air dose rates.

Finally, we produced the prediction surfaces in the rectangular region $[140.07,$ $140.08] \times [35.71, 35.74]$ so as to examine how well the $M$-RA-lp, MLP, and $M$-RA perform in the prediction of a region. The prediction surfaces were generated by calculating the mean vector of the predictive distribution at $31 \times 31$ lattice points in $[140.07, 140.08] \times [35.71, 35.74]$ by using 7,801 observations with the test set of 129 observations. For $\sigma^2$, $\theta$, and $\tau^2$, we used the estimated values of each method in the results of Table \ref{tab:real_data_analysis}, and tuning parameter settings were also identical to those in the real data analysis of Table \ref{tab:real_data_analysis}. 

\begin{figure}[p]
%\centering
 \subfigure[Original model]{
  \includegraphics[width = 7.3cm,pagebox=artbox,clip]{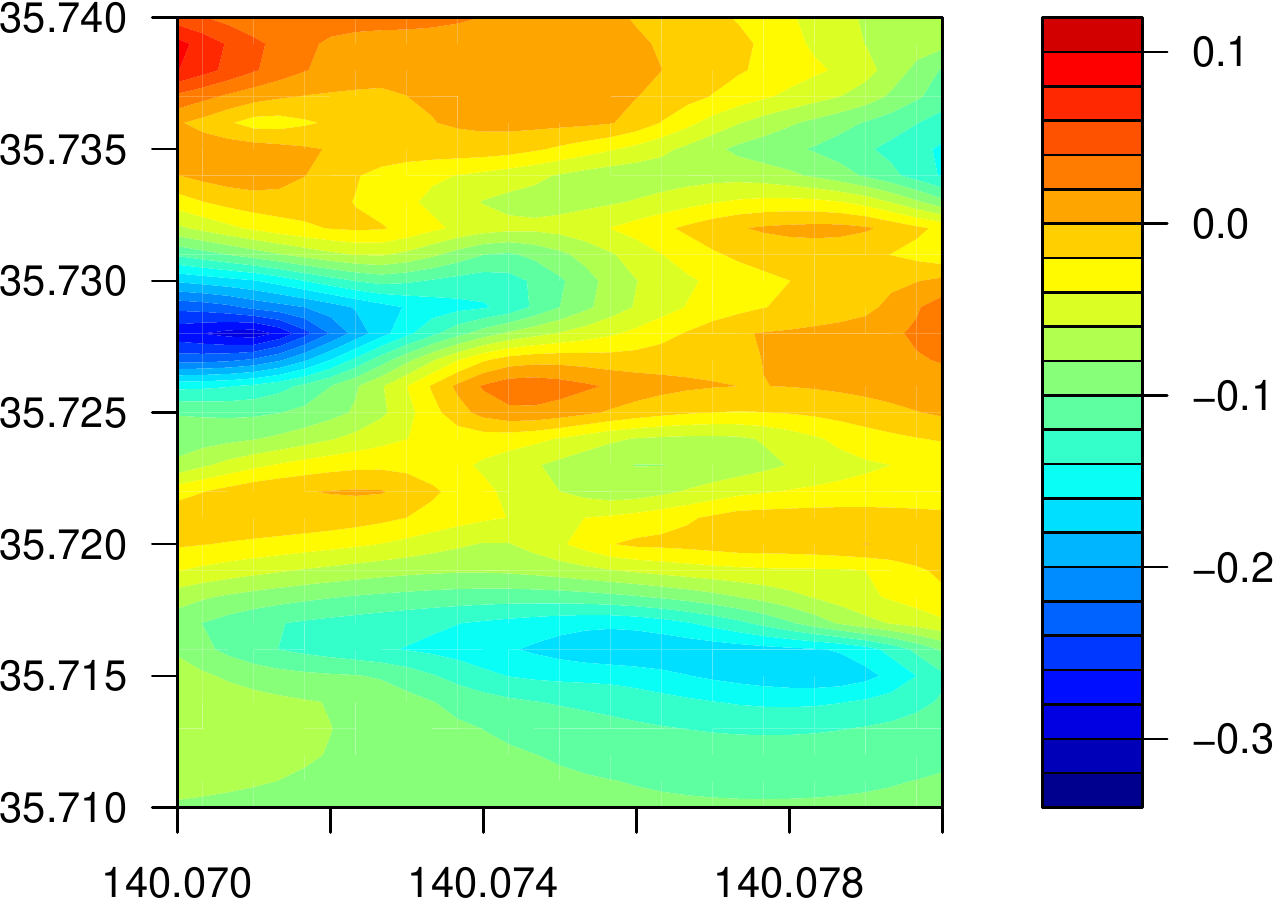}\label{fig:a}}
 \subfigure[$M$-RA-lp]{ 
  \includegraphics[width = 7.3cm,pagebox=artbox,clip]{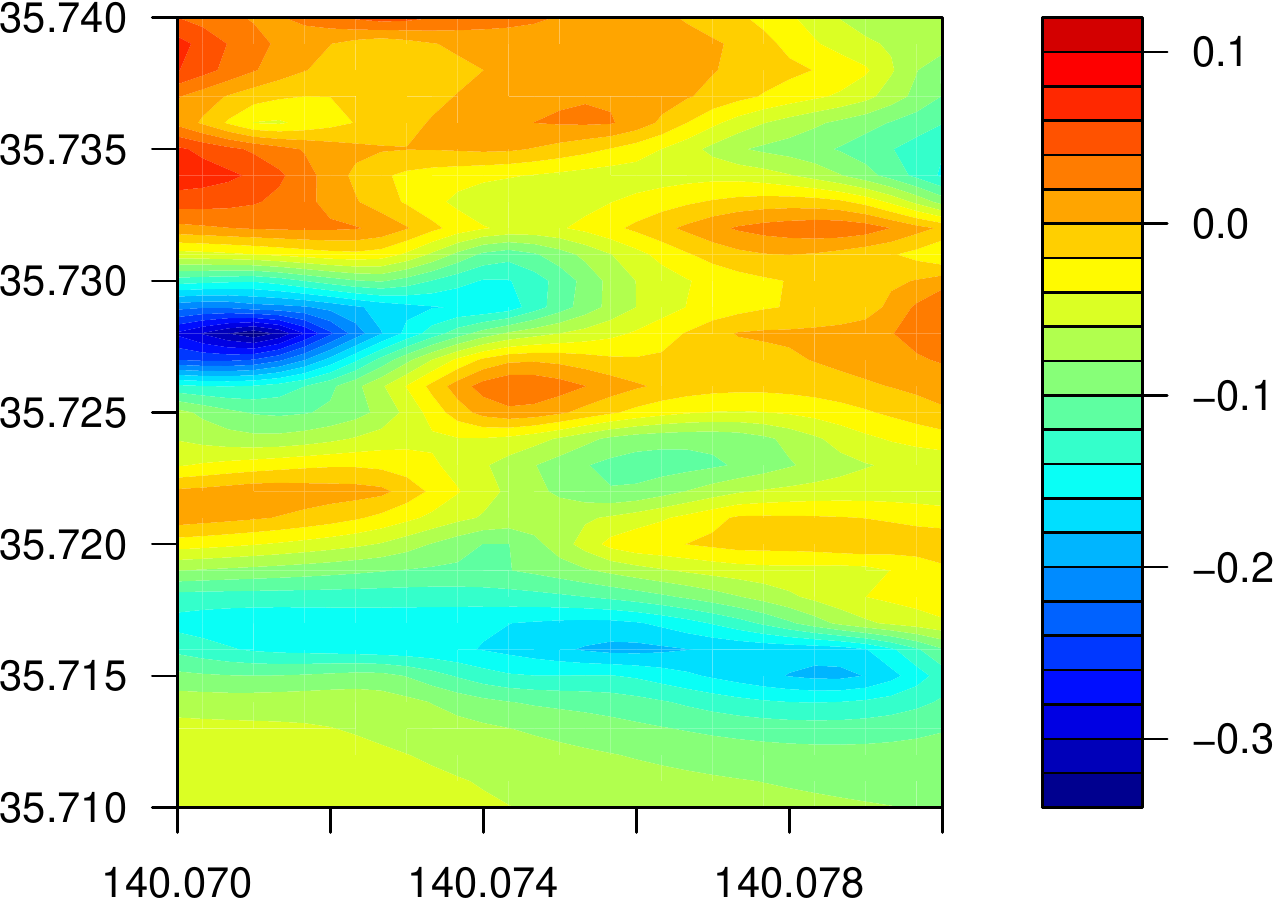}\label{fig:b}}
 \subfigure[MLP]{
  \includegraphics[width = 7.3cm,pagebox=artbox,clip]{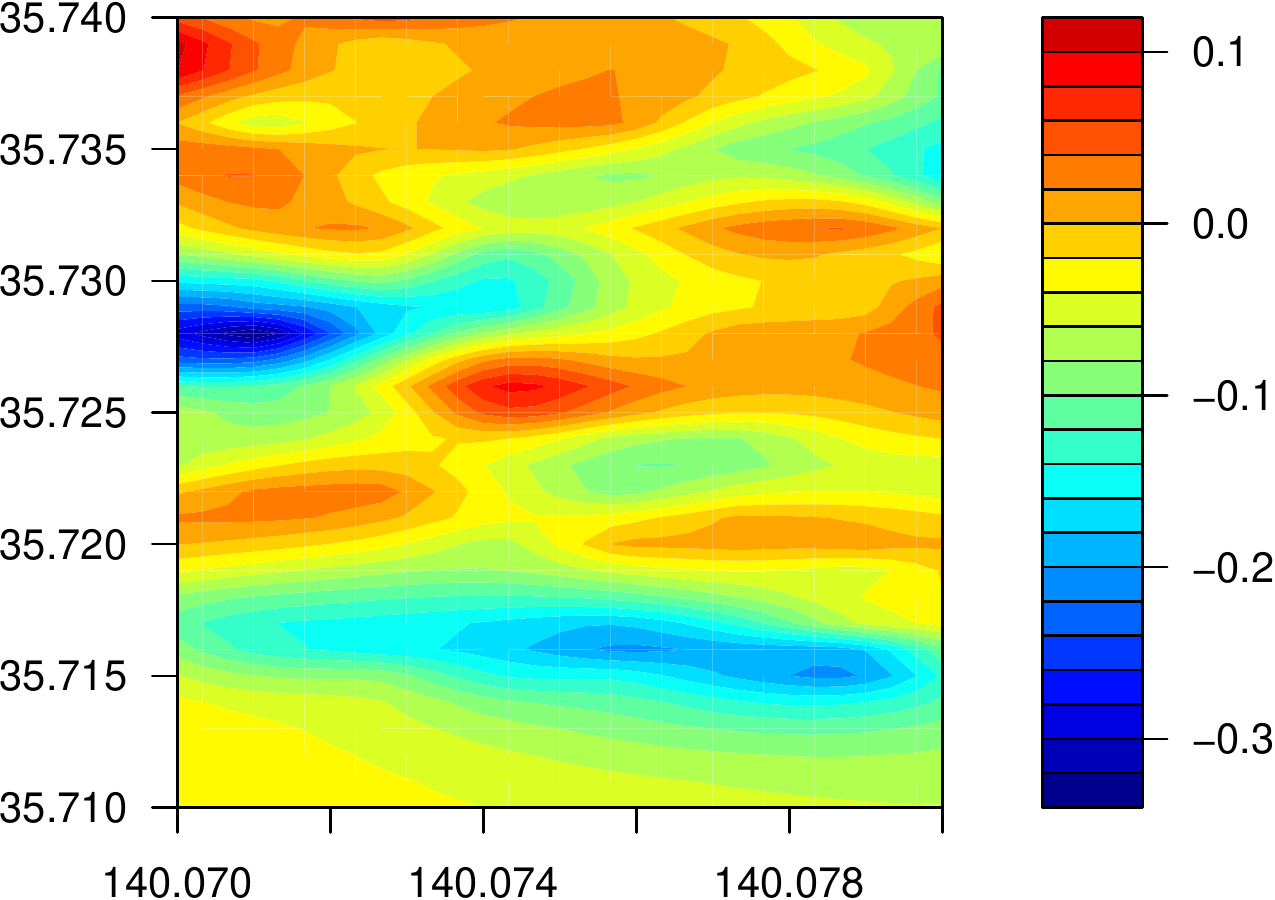}\label{fig:c}}
 \subfigure[$M$-RA (Case 1)]{
  \includegraphics[width = 7.3cm,pagebox=artbox,clip]{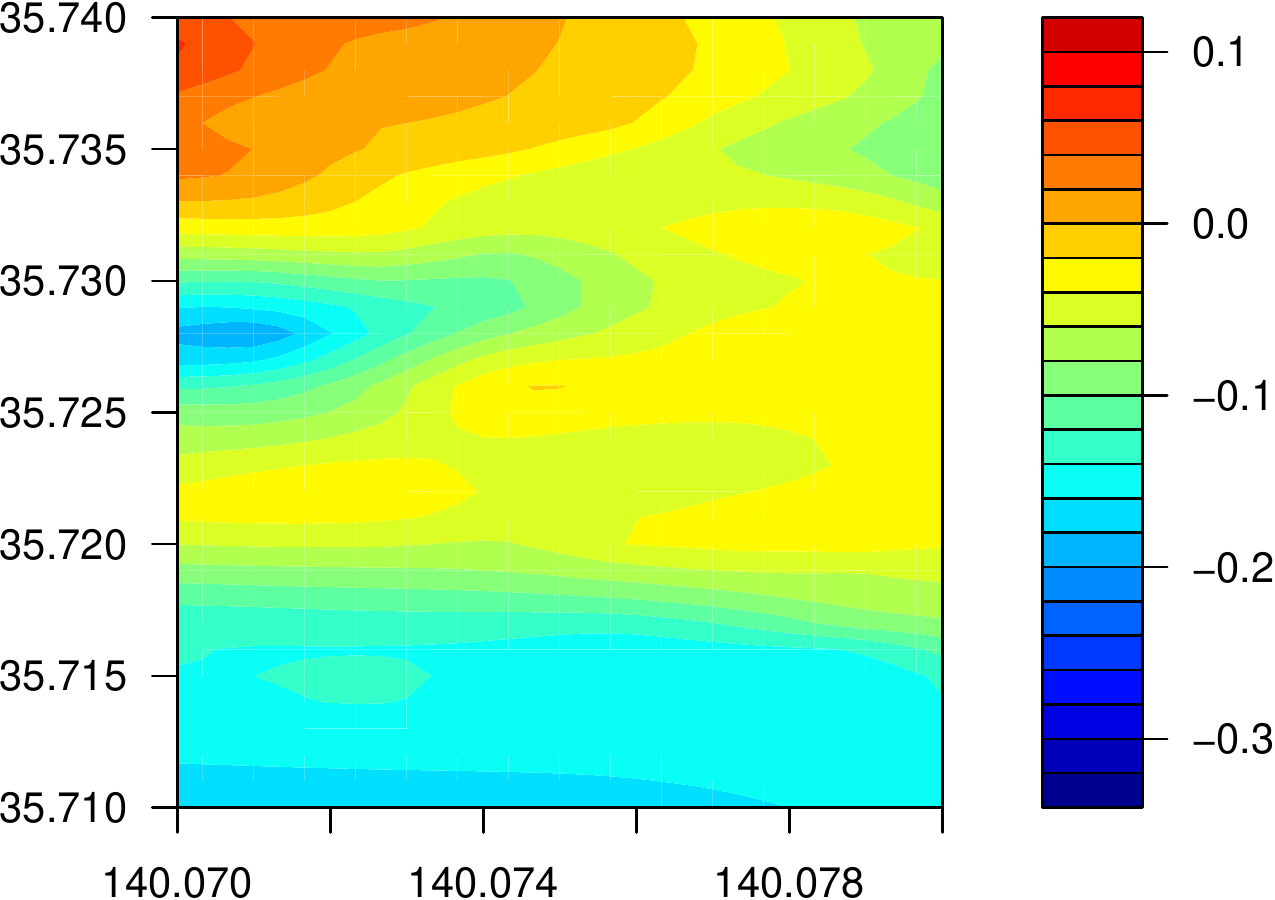}\label{fig:d}}
   \subfigure[$M$-RA (Case 2)]{
  \includegraphics[width = 7.3cm,pagebox=artbox,clip]{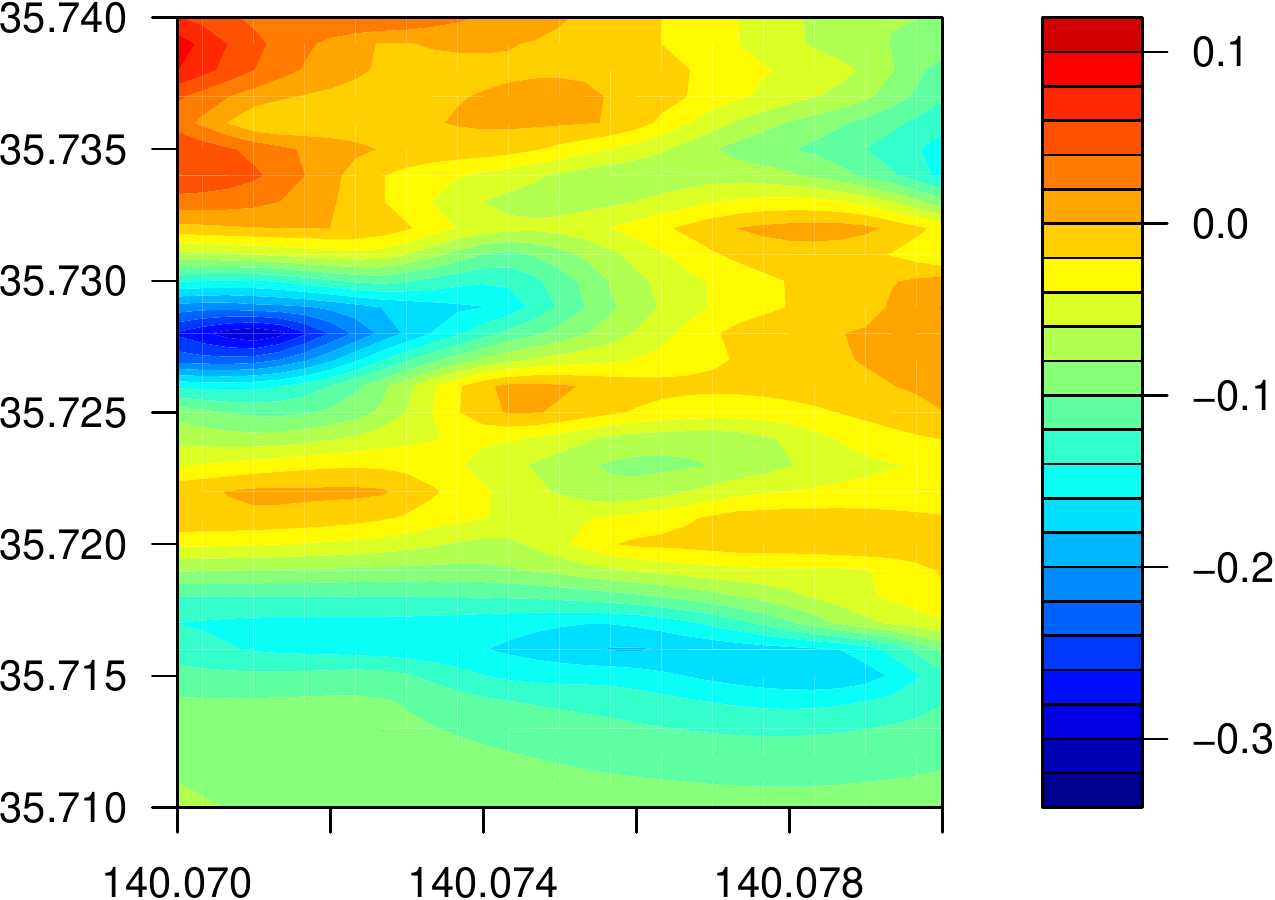}\label{fig:e}}
  \caption{%The mean vector of the predictive distribution by using the $M$-RA-lp, MLP, and $M$-RA.
  The prediction surfaces generated by the original model, $M$-RA-lp, MLP, and $M$-RA.}
\label{fig:surface_comparison} 
\end{figure}

Figure \ref{fig:surface_comparison} shows the prediction surfaces of the original model and three approximation methods. The prediction surfaces of the $M$-RA-lp, MLP, and $M$-RA (Case 2) were similar to that of the original model. For $\hat{\theta} = 1.6109$ in the original model of Table \ref{tab:real_data_analysis}, the effective range was 1.86 km, which shows the weak spatial correlation because the sides of the sampling region are 13.58 km and 24.41km. Consequently, the MLP depicted the good prediction surface because of the effectiveness of the covariance tapering. Also, some partial shapes in the prediction surface of the $M$-RA (Case 2) resembled those of the original model very well because the $M$-RA completely recovers the spatial correlation between observations in $[140.07, 140.08] \times [35.71, 35.74]$ as explained in Section \ref{subsec:relation}. However, the relative computational times of producing the prediction surfaces were 0.2063, 1.0299, 0.0715, and 0.3287 for the $M$-RA-lp, MLP, $M$-RA (Case 1), and $M$-RA (Case 2), respectively. The $M$-RA-lp generated the prediction surface in a shorter computation  time than the $M$-RA (Case 2), and this demonstrates that the $M$-RA-lp is the reasonable fast computation approach.

\section{Conclusion and future studies\label{sec:conclusion_future_studies}}

In this paper, we have described the multi-resolution approximation via linear projection ($M$-RA-lp). The proposed method refined the MLP of \cite{Hirano_2017a} by introducing the multiple resolutions and the recursive partitioning of the entire spatial domain based on the idea of \cite{Katzfuss_2017}. Also, the $M$-RA-lp can be regarded as an extension of the $M$-RA of \cite{Katzfuss_2017} by replacing the predictive process in the $M$-RA with the linear projection. Some simulations suggested that this replacement gave rise to better numerical stability by reducing the condition number, which is consistent with the results of \cite{Banerjee_2013}. In simulation studies and the real data analysis, the $M$-RA-lp was generally efficient compared with the MLP and $M$-RA in terms of the approximation of the log-likelihood function and predictive distribution at unobserved locations. 

Some issues are to be solved in the future. First, \cite{Katzfuss_2019} pointed out discontinuities of the $M$-RA  and proposed a taper version of the $M$-RA. In order to bypass the artificiality presented in Section \ref{subsec:relation}, we plan to derive a taper version of the $M$-RA-lp. 
Second, since the $M$-RA-lp has many tuning parameters, a comprehensive study on their selection is left for a future study. The faster selection method of $\Phi_{j_1,\ldots,j_m}$ should also be investigated. 
%it is also interesting to extend the current work to non-Gaussian data by using a generalized Gaussian process \citep[e.g.,][]{Zilber_2020}. 
Third, \cite{Jurek_2019} developed a multi-resolution filter for massive spatio-temporal data. Similarly, our proposed method might be extended to a spatio-temporal process. Finally, \cite{Katzfuss_2020a} proposed Vecchia approximations which contain many existing fast computation methods as well as the $M$-RA as special cases. 
This general Vecchia framework was applied to a variety of settings such as the prediction, non-Gaussian case, and computer experiments \citep[][]{Katzfuss_2020b, Zilber_2020, Katzfuss_2020c}. It is also interesting to investigate the relationship between the Vecchia approximations and the $M$-RA-lp. 

%\section{Section title}
%\label{sec:1}
%Text with citations \cite{RefB} and \cite{RefJ}.
%\subsection{Subsection title}
%\label{sec:2}
%as required. Don't forget to give each section
%and subsection a unique label (see Sect.~\ref{sec:1}).
%\paragraph{Paragraph headings} Use paragraph headings as needed.
%\begin{equation}
%\end{equation}
%
%% For one-column wide figures use
%\begin{figure}
%% Use the relevant command to insert your figure file.
%% For example, with the graphicx package use
%  \includegraphics{example.eps}
%\end{figure}
%%
%% For two-column wide figures use
%\begin{figure*}
%% Use the relevant command to insert your figure file.
%% For example, with the graphicx package use
%  \includegraphics[width=0.75\textwidth]{example.eps}
%\end{figure*}
%\begin{table}
%\begin{tabular}{lll}
%\hline\noalign{\smallskip}
%\end{tabular}
%\end{table}

%
%----------------------------------------------------------------------
% Appendix
%---------------------------------------------------------------------- 
%

\appendix

\def\thesection{Appendix \Alph{section}}
\section{Technical lemmas}
\def\thesection{\Alph{section}}
\label{append:lemmas}

This appendix collects some relevant results on matrix algebra.

\begin{lemma}[A part of Proposition 5.4 of \cite{Puntanen_2011}]
\label{lem:rank_mat}
Let A be a positive semidefinite  $n \times n$ matrix and B be an $n \times p$ matrix. If $\mbox{rank}(B)=p$ and $R(B) \cap R(A)^{\perp} = \{ \bm{0} \}$, 
then $\mbox{rank}(B^{\top} A B)=p$.
\end{lemma}

\begin{lemma}[Lemma 2 of \cite{Welling_2010}]
\label{lem:inv_formula_variant}
Let A be a positive definite $n \times n$ matrix, B be a positive definite $m \times m$ matrix, and C be a $m \times n$ matrix. Then, 
\begin{align*}
\left( A^{-1} + C^{\top} B^{-1} C \right)^{-1} C^{\top} B^{-1} = A C^{\top} \left( C A C^{\top} + B \right)^{-1}.
\end{align*}
\end{lemma}

We can prove Lemma \ref{lem:inv_formula_variant} by using Theorem 18.2.8 of \cite{Harville_1997} known as the Sherman--Morrison--Woodbury formula.

\def\thesection{Appendix \Alph{section}}
\section{Proof of Proposition \ref{prop:m-ra-lp}}
\def\thesection{\Alph{section}}
\label{append:prop_proof}

\begin{proof}[Proof of Proposition \ref{prop:m-ra-lp}]
(a) Consider $\bm{a}^{\top} C_m(S_0,S_0) \bm{a}$ for any $\bm{a} \in \mathbb{R}^n$ and any set of locations $S_0 \subset D_0$. We will show the assertion by mathematical induction in the same way as the proof of Proposition 1 of \cite{Katzfuss_2017}. For $m=1$, we have
\begin{align*}
C_0(S_0,S_0) - C_{\tau_0}(S_0,S_0)
%&= C_0(S_0,S_0) - C_0(S_0,Q_0) \Phi_0^{\top} (\Phi_0  K_0^0 \Phi_0^{\top})^{-1}\Phi_0 C_0(S_0,Q_0)^{\top} \\
&=\mbox{Var}(\bm{Y}_0(S_0)) - \mbox{Var}(\bm{\tau}_0 (S_0))\\
&=\mbox{Var}(\bm{Y}_0(S_0)) - \mbox{Var}(E[\bm{Y}_0(S_0) | \Phi_0 \bm{Y}_0(Q_0)]) \\
%&=E[\mbox{Var}(Y_0(S_0) | \Phi_0 \bm{Y}_0(Q_0)) ],
&=\mbox{Var}(\bm{Y}_0(S_0) | \Phi_0 \bm{Y}_0(Q_0)),
\end{align*}
where the third equation holds by using the law of total variance, $Y_0(s) \sim \mbox{GP}(0,C_0)$, and the fact that 
%$\mbox{Var}(Y_0(S_0) | \Phi_0 \bm{Y}_0(Q_0))$ does not include a random variable from 
$\mbox{Var}(\Phi_0 \bm{Y}_0(Q_0)) = \Phi_0 K^0_0 \Phi_0^{\top}$ is positive definite from the proof of Proposition 4.1 (a) of \cite{Hirano_2017a}. Thus, $C_0(S_0,S_0) - C_{\tau_0}(S_0,S_0)$ is positive semidefinite. For any $\bm{a} = (\bm{a}_1^{\top}, \ldots, \bm{a}_{J_1}^{\top})^{\top} \in \mathbb{R}^n$ such that $|\bm{a}_{j_1}| = |S_{j_1}|$ and $\bm{a}_{j_1} = \emptyset$ if $|S_{j_1}|=0$, 
\begin{align*}
\bm{a}^{\top} C_1(S_0,S_0) \bm{a} = \sum_{\substack{1 \le j_1 \le J_1, j_1 \notin A_1, \\A_1 = \left\{ j'_1 \middle| \left|S_{j'_1} \right|=0 \right\} }} \bm{a}_{j_1}^{\top} \left\{ C_0(S_{j_1}, S_{j_1}) - C_{\tau_0}(S_{j_1}, S_{j_1}) \right\} \bm{a}_{j_1} \ge 0.
\end{align*}
Therefore, the result holds for $m=1$. 

Next, assume that the result holds for $m=l$. Since $\mbox{rank} \left( \Phi_{j_1,\ldots, j_l} K_{j_1,\ldots,j_l}^l \Phi_{j_1,\ldots, j_l}^{\top} \right) = r_{j_1,\ldots, j_l}$ by the assumption and Lemma \ref{lem:rank_mat}, $\Phi_{j_1,\ldots, j_l} K_{j_1,\ldots,j_l}^l \Phi_{j_1,\ldots, j_l}^{\top}$ is positive definite, 
so that $\Phi^{(l)} C_l (Q^{(l)},Q^{(l)}) \Phi^{{(l)}^{\top}}$ is also positive definite. Therefore, for $\delta_l(s) \sim \mbox{GP}(0,C_l)$, we can define
\begin{align*}
\tau_l(\bm{s}) &= E[\delta_l(\bm{s}) | \Phi^{(l)} \bm{\delta}_l (Q^{(l)})]\\
&=C_l(\bm{s},Q^{(l)}) \Phi^{{(l)}^{\top}} \left\{\Phi^{(l)} C_l (Q^{(l)},Q^{(l)}) \Phi^{{(l)}^{\top}} \right\}^{-1} \Phi^{(l)}\bm{\delta}_l (Q^{(l)}).
\end{align*}
Then, in the same way as the argument of $m=1$, 
%\begin{align*}
%C_l(S_0,S_0) - C_{\tau_l}(S_0,S_0)
%&= C_0(S_0,S_0) - C_0(S_0,Q_0) \Phi_0^{\top} (\Phi_0  K_0^0 \Phi_0^{\top})^{-1}\Phi_0 C_0(S_0,Q_0)^{\top} \\
%&=\mbox{Var}(\bm{\delta}_l (S_0)) - \mbox{Var}(\bm{\tau}_l (S_0))\\
%&=\mbox{Var}(\bm{\delta}_l (S_0)) - \mbox{Var}(E[\bm{\delta}_l(S_0) | \Phi^{(l)} \bm{\delta}_l (Q^{(l)})]) \\
%&=E[\mbox{Var}(Y_0(S_0) | \Phi_0 \bm{Y}_0(Q_0)) ],
%&=\mbox{Var}(\bm{\delta}_l (S_0) | \Phi^{(l)} \bm{\delta}_l (Q^{(l)})),
%\end{align*}
%where the third equation holds by using the law of total variance and the property that 
%$\mbox{Var}(Y_0(S_0) | \Phi_0 \bm{Y}_0(Q_0))$ does not include a random variable from 
%$\delta_l(s) \sim GP(0,C_l)$ and $\mbox{Var}(\Phi^{(l)} \bm{\delta}_l (Q^{(l)})) = \Phi^{(l)} C_l (Q^{(l)},Q^{(l)}) \Phi^{{(l)}^{\top}}$ is nonsingular. Consequently, 
$C_l(S_0,S_0) - C_{\tau_l}(S_0,S_0)$ is positive semidefinite because $C_l(S_0,S_0) - C_{\tau_l}(S_0,S_0)=\mbox{Var}(\bm{\delta}_l (S_0) | \Phi^{(l)} \bm{\delta}_l (Q^{(l)}))$.

Consider $m=l+1$. For any $\bm{a} = (\bm{a}_{1,\ldots,1}^{\top}, \ldots, \bm{a}_{J_1, \ldots, J_{l+1}}^{\top})^{\top} \in \mathbb{R}^n$ such that $|\bm{a}_{j_1,\ldots,j_{l+1}}| = |S_{j_1,\ldots,j_{l+1}}|$ and $\bm{a}_{j_1,\ldots, j_{l+1}} = \emptyset$ if $|S_{j_1,\ldots,j_{l+1}}|=0$, 
%\begin{align*}
%\bm{a}^{\top} C_{l+1}(S_0,S_0) \bm{a} =& \sum_{\substack{1 \le j_1 \le J_1, \ldots, 1 \le j_{l+1} \le J_{l+1}, \\ 
%(j_1,\ldots,j_{l+1}) \notin A_{l+1},\\ A_{l+1} = \left\{ (j'_1,\ldots,j'_{l+1}) \middle|  \left|S_{j'_1,\ldots,j'_{l+1}} \right|=0 \right\} }} \bm{a}_{j_1,\ldots,j_{l+1}}^{\top} \left\{ C_l(S_{j_1,\ldots,j_{l+1}},S_{j_1,\ldots,j_{l+1}}) \right. \\
%& \left. \hspace{4cm} - C_{\tau_l}(S_{j_1,\ldots,j_{l+1}},S_{j_1,\ldots,j_{l+1}}) \right\} \bm{a}_{j_1,\ldots,j_{l+1}} \\
%\ge& 0.
%\end{align*}
\begin{align*}
\bm{a}^{\top} & C_{l+1}(S_0,S_0) \bm{a} = \sum_{\substack{1 \le j_1 \le J_1, \ldots, 1 \le j_{l+1} \le J_{l+1}, \\ 
(j_1,\ldots,j_{l+1}) \notin A_{l+1},\\ A_{l+1} = \left\{ (j'_1,\ldots,j'_{l+1}) \middle|  \left|S_{j'_1,\ldots,j'_{l+1}} \right|=0 \right\} }} \bm{a}_{j_1,\ldots,j_{l+1}}^{\top} \left\{ C_l(S_{j_1,\ldots,j_{l+1}},S_{j_1,\ldots,j_{l+1}}) \right. \\
& \left. - C_{\tau_l}(S_{j_1,\ldots,j_{l+1}},S_{j_1,\ldots,j_{l+1}}) \right\} \bm{a}_{j_1,\ldots,j_{l+1}} 
\\
\ge& 0.
\end{align*}
The result holds for $m=l+1$. The proof is completed.
\par
\bigskip
\noindent
(b) 
As shown in the proof of Proposition \ref{prop:m-ra-lp} (a), $\Phi_0 K^0_0 \Phi_0^{\top}$ is positive definite. For $m =1,\ldots,M-1$, $K_{j_1,\ldots,j_m}^m$ is positive semidefinite from Proposition \ref{prop:m-ra-lp} (a). Since $\Phi_{j_1,\ldots,j_m} K_{j_1,\ldots,j_m}^m \Phi_{j_1,\ldots,j_m}^{\top}$ is nonsingular by using Lemma \ref{lem:rank_mat}, $\Phi_{j_1,\ldots,j_m} K_{j_1,\ldots,j_m}^m \Phi_{j_1,\ldots,j_m}^{\top}$ is positive definite.
\par
\bigskip
\noindent
(c) For any set of locations $S_0 \subset D_0$, we have
\begin{align*}
C_{M \text{-RA-lp}}(S_0,S_0) =  \sum_{m=0}^{M-1} C_{\tau_{m}}(S_0,S_0) + C_M(S_0,S_0) \circ T_{\gamma}(S_0,S_0),
\end{align*}
where $C_{\tau_{m}}(S_0,S_0)$ is the block diagonal matrix whose diagonal element is $B_{j_1,\ldots,j_m}^m \Phi_{j_1,\ldots,j_m}^{\top} \\
\times \widehat{K}_{j_1,\ldots, j_m}^m \Phi_{j_1,\ldots,j_m} B_{j_1,\ldots,j_m}^{m^{\top}}$ if $|S_{j_1,\ldots,j_m}| \neq 0$. From Proposition \ref{prop:m-ra-lp} (b), $\widehat{K}_{j_1,\ldots, j_m}^m$ is positive definite, so that $C_{\tau_{m}}(S_0,S_0)$ is positive semidefinite. In addition, from Theorem 5.2.1 of \cite{Horn_1991} and  Proposition \ref{prop:m-ra-lp} (a), $C_M(S_0,S_0) \circ T_{\gamma}(S_0,S_0)$ is also positive semidefinite. 
\par
\bigskip
\noindent
(d) Since $\bm{s}_1$ and $\bm{s}_2$ are always in the same subregion $D_{j_1,\ldots,j_m}$ at each $m$th resolution ($m=0, \ldots, M$), 
\begin{align}
\label{eq:cov_s_s}
C_{M \text{-RA-lp}}(\bm{s}_1, \bm{s}_2) &= \sum_{m=0}^{M-1} C_{\tau_{m}}(\bm{s}_1, \bm{s}_2) + C_M(\bm{s}_1, \bm{s}_2)
%&= C_{\tau_0}(\bm{s}_1, \bm{s}_2) + \cdots + C_{\tau_{M-1}}(\bm{s}_1, \bm{s}_2) + C_M(\bm{s}_1, \bm{s}_2)
\end{align}
and
\begin{align}
\label{eq:c_m}
C_{m}(\bm{s}_1, \bm{s}_2) &= C_{m-1}(\bm{s}_1, \bm{s}_2) - C_{\tau_{m-1}}(\bm{s}_1, \bm{s}_2) \quad (m=1,\ldots,M).
\end{align}
By substituting \eqref{eq:c_m} into \eqref{eq:cov_s_s} recursively, the assertion is obtained.
\end{proof}

\def\thesection{Appendix \Alph{section}}
\section{Derivation of Algorithms \ref{alg:log-lik_m-ra-lp} and \ref{alg:pred_dist_m-ra-lp}}
\def\thesection{\Alph{section}}
\label{append:algorithms}

\subsection{Expansion of the inversion and determinant}
\label{subappend:expansion_inversion_determinant}

The following lemma is used when deriving Algorithms \ref{alg:log-lik_m-ra-lp} and \ref{alg:pred_dist_m-ra-lp}.

\begin{lemma}
\label{lem:log-lik_m-ra-lp}
Suppose that the assumption of Proposition \ref{prop:m-ra-lp} holds.
\par
\noindent
(a) $\Sigma_{j_1,\ldots,j_m}$ is positive definite for $m=0, \ldots, M$.
\par
\noindent
(b) For $m=0,\ldots,M-1$, 
\begin{align*}
\Sigma_{j_1,\ldots,j_m}^{-1} = &V_{j_1,\ldots,j_m}^{-1} - V_{j_1,\ldots,j_m}^{-1} B_{j_1,\ldots,j_m}^m \Phi_{j_1,\ldots, j_m}^{\top} \left( \widehat{K}_{j_1,\ldots,j_m}^{m^{-1}} 
+ \Phi_{j_1,\ldots, j_m} B_{j_1,\ldots,j_m}^{m^{\top}} \right. \\ 
&\times \left. V_{j_1,\ldots,j_m}^{-1} B_{j_1,\ldots,j_m}^m \Phi_{j_1,\ldots, j_m}^{\top} \right)^{-1}  \Phi_{j_1,\ldots, j_m} B_{j_1,\ldots,j_m}^{m^{\top}} V_{j_1,\ldots,j_m}^{-1}, \\
\det \left( \Sigma_{j_1,\ldots,j_m} \right) =& \det \left( \widehat{K}_{j_1,\ldots,j_m}^{m^{-1}} + \Phi_{j_1,\ldots, j_m} B_{j_1,\ldots,j_m}^{m^{\top}} V_{j_1,\ldots,j_m}^{-1} B_{j_1,\ldots,j_m}^m \Phi_{j_1,\ldots, j_m}^{\top} \right) \\
&\times \det \left( \widehat{K}_{j_1,\ldots,j_m}^{m}  \right) \det \left( V_{j_1,\ldots,j_m} \right).
\end{align*}
\end{lemma}

\begin{proof}[Proof of Lemma \ref{lem:log-lik_m-ra-lp}]
(a) For $m=M$, $\Sigma_{j_1,\ldots,j_M}=K_{j_1,\ldots,j_M}^M \circ T_{\gamma}(S_{j_1,\ldots,j_M},S_{j_1,\ldots,j_M}) + \tau^2 \bm{\mbox{I}}_{\left| S_{j_1,\ldots, j_M}\right|}$ is positive definite from the proof of Proposition \ref{prop:m-ra-lp} (c). 
%For $m=M-1$, since $\Sigma_{j_1,\ldots,j_M}$ is positive definite, $V_{j_1,\ldots,j_{M-1}}$ is also positive definite. Additionally, from Proposition \ref{prop:m-ra-lp} (b), $\widehat{K}_{j_1,\ldots,j_m}^{m}$ is positive definite. Thus, $\Sigma_{j_1,\ldots,j_{M-1}}$ is positive definite.
For $m=l+1$, assume that $\Sigma_{j_1,\ldots,j_{l+1}}$ is positive definite. For $m=l$, $V_{j_1,\ldots,j_l}$ is positive definite from the definition of $V_{j_1,\ldots,j_l}$. Additionally, from Proposition \ref{prop:m-ra-lp} (b), $\widehat{K}_{j_1,\ldots,j_l}^{l}$ is positive definite. Thus, $\Sigma_{j_1,\ldots,j_l}$ is positive definite. The proof is completed by mathematical induction.
\par
\bigskip
\noindent
(b) From the proof of Lemma \ref{lem:log-lik_m-ra-lp} (a), $V_{j_1,\ldots,j_m}$, $\widehat{K}_{j_1,\ldots,j_m}^{m}$, and 
$\widehat{K}_{j_1,\ldots,j_m}^{m^{-1}} + \Phi_{j_1,\ldots, j_m} B_{j_1,\ldots,j_m}^{m^{\top}} \\ \times  V_{j_1,\ldots,j_m}^{-1} B_{j_1,\ldots,j_m}^m \Phi_{j_1,\ldots, j_m}^{\top}$ 
are positive definite for $m=0,\ldots, M-1$. From Theorems 18.1.1 and 18.2.8 of \cite{Harville_1997}, the assertion is obtained.
\end{proof}

\subsection{Derivation of Algorithm \ref{alg:log-lik_m-ra-lp}}
\label{subappend:algorithm_log-lik} 

From \eqref{eq:cov_resolution_m}, \eqref{eq:lp_resolution_m}, and \eqref{eq:cov_resolution_M}, for $1 \le k \le m$, $m=1,\ldots,M$, we obtain
\begin{align}
K_{j_1,\ldots,j_m}^k &= C_k(Q_{j_1,\ldots,j_k},Q_{j_1,\ldots,j_m}) \notag \\
&= C_{k-1}(Q_{j_1,\ldots,j_k},Q_{j_1,\ldots,j_m}) - K_{j_1,\ldots,j_k}^{{k-1}^{\top}} \Phi_{j_1,\ldots,j_{k-1}}^{\top}
\widehat{K}_{j_1,\ldots,j_{k-1}}^{k-1} \Phi_{j_1,\ldots,j_{k-1}} K_{j_1,\ldots,j_m}^{k-1}. \label{eq:K_m_recursive}
\end{align}
By applying \eqref{eq:cov_resolution_m} and \eqref{eq:lp_resolution_m} to \eqref{eq:K_m_recursive} recursively, \eqref{eq:K_expansion} is obtained.

Next, we define $A_{j_1,\ldots,j_m}^{k,l} = B_{j_1,\ldots,j_m}^{k^{\top}} V_{j_1,\ldots,j_m}^{-1} B_{j_1,\ldots,j_m}^l$ ($0 \le k \le l \le m,\; m=0,\ldots,M-1$), $A_{j_1,\ldots,j_m}^{l,k} = A_{j_1,\ldots,j_m}^{{k,l}^{\top}}$, and $\widetilde{A}_{j_1,\ldots,j_m}^{k,l} = B_{j_1,\ldots,j_m}^{k^{\top}} \Sigma_{j_1,\ldots,j_m}^{-1} B_{j_1,\ldots,j_m}^l$ ($0 \le k \le l < m,\; m=1,\ldots,M$). Since $B_{j_1,\ldots,j_m}^{k} = (B_{j_1,\ldots,j_m,1}^{k^{\top}}, \ldots, B_{j_1,\ldots,j_m,J_{m+1}}^{k^{\top}}  )^{\top}$, we have 
\begin{align*}
A_{j_1,\ldots,j_m}^{k,l} &= \sum _{j_{m+1}=1}^{J_{m+1}} B_{j_1,\ldots,j_m,j_{m+1}}^{k^{\top}} \Sigma_{j_1,\ldots,j_m,j_{m+1}}^{-1} B_{j_1,\ldots,j_m,j_{m+1}}^{l} \\
&=\sum _{j_{m+1}=1}^{J_{m+1}} \widetilde{A}_{j_1,\ldots,j_m,j_{m+1}}^{k,l}.
\end{align*}
For $m=1,\ldots,M-1$, it follows from Lemma \ref{lem:log-lik_m-ra-lp} (b) that
\begin{align*}
\widetilde{A}_{j_1,\ldots,j_m}^{k,l} =& B_{j_1,\ldots,j_m}^{k^{\top}} V_{j_1,\ldots,j_m}^{-1} B_{j_1,\ldots,j_m}^l - B_{j_1,\ldots,j_m}^{k^{\top}} V_{j_1,\ldots,j_m}^{-1} B_{j_1,\ldots,j_m}^m \Phi_{j_1,\ldots, j_m}^{\top} \left( \widehat{K}_{j_1,\ldots,j_m}^{m^{-1}} \right.\\
&\left. +  \Phi_{j_1,\ldots, j_m} A_{j_1,\ldots,j_m}^{m,m} \Phi_{j_1,\ldots, j_m}^{\top} \right)^{-1} \Phi_{j_1,\ldots, j_m} B_{j_1,\ldots,j_m}^{m^{\top}} V_{j_1,\ldots,j_m}^{-1} B_{j_1,\ldots,j_m}^{l}\\
=& A_{j_1,\ldots,j_m}^{k,l} - A_{j_1,\ldots,j_m}^{k,m} \Phi_{j_1,\ldots,j_m}^{\top} \widetilde{K}_{j_1,\ldots,j_m}^m \Phi_{j_1,\ldots,j_m} A_{j_1,\ldots,j_m}^{m,l}.
\end{align*}
Thus, \eqref{eq:A_recursive} and \eqref{eq:tilA_recursive} hold.

Also, we define $\omega_{j_1,\ldots,j_m}^{k} = B_{j_1,\ldots,j_m}^{k^{\top}} V_{j_1,\ldots,j_m}^{-1} \bm{Z}(S_{j_1,\ldots,j_m})$ ($0 \le k \le m,\; m=0,\ldots,M-1$) and $\widetilde{\omega}_{j_1,\ldots,j_m}^{k} = B_{j_1,\ldots,j_m}^{k^{\top}} \Sigma_{j_1,\ldots,j_m}^{-1} \bm{Z}(S_{j_1,\ldots,j_m})$ ($0 \le k < m,\; m=1,\ldots,M$). By the same argument as the case of $A_{j_1,\ldots,j_m}^{k,l}$ and $\widetilde{A}_{j_1,\ldots,j_m}^{k,l}$, we can obtain \eqref{eq:omega_recursive} and \eqref{eq:tilomega_recursive}.
%\begin{align*}
%\omega_{j_1,\ldots,j_m}^{k} &= \sum _{j_{m+1}=1}^{J_{m+1}} B_{j_1,\ldots,j_m,j_{m+1}}^{k^{\top}} \Sigma_{j_1,\ldots,j_m,j_{m+1}}^{-1} \bm{Z}(S_{j_1,\ldots,j_m,j_{m+1}}) \\
%&=\sum _{j_{m+1}=1}^{J_{m+1}} \widetilde{\omega}_{j_1,\ldots,j_m,j_{m+1}}^{k}
%\end{align*}

Finally, we define $d_{j_1,\ldots,j_m} = \log \left\{ \mbox{det} \left( \Sigma_{j_1,\ldots,j_m} \right) \right\}$ and $u_{j_1,\ldots,j_m} = \bm{Z}(S_{j_1,\ldots,j_m})^{\top} \Sigma_{j_1,\ldots,j_m}^{-1}\\ \times \bm{Z}(S_{j_1,\ldots,j_m})$ ($m=0,\ldots,M$). From $\Sigma_0 = C_{M \text{-RA-lp}}(S_0,S_0) + \tau^2 \bm{\mbox{I}}_n$, it follows that $d_0 = \log \left[ \mbox{det} \left\{ C_{M \text{-RA-lp}}(S_0,S_0) \right. \right. \\ \left. \left.+ \tau^2 \bm{\mbox{I}}_n \right\} \right]$ and $u_0 = \bm{Z}(S_0)^{\top} \left\{ C_{M \text{-RA-lp}}(S_0,S_0) + \tau^2 \bm{\mbox{I}}_n \right\}^{-1} \bm{Z}(S_0)$. By using Lemma \ref{lem:log-lik_m-ra-lp} (b), for $m=0,\ldots,M-1$, we have
\begin{align*}
d_{j_1,\ldots,j_m} =& \log \left\{ \mbox{det} \left( \widehat{K}_{j_1,\ldots,j_m}^{m^{-1}} + \Phi_{j_1,\ldots, j_m} B_{j_1,\ldots,j_m}^{m^{\top}} V_{j_1,\ldots,j_m}^{-1} B_{j_1,\ldots,j_m}^m \Phi_{j_1,\ldots, j_m}^{\top} \right) \right.\\
& \left. \times \mbox{det} \left( \widehat{K}_{j_1,\ldots,j_m}^{m}  \right) \mbox{det} \left( V_{j_1,\ldots,j_m} \right) \right\} \\
=&-\log \left\{ \mbox{det} \left( \widetilde{K}_{j_1,\ldots,j_m}^m \right) \right\} + \log \left\{ \mbox{det} \left( \widehat{K}_{j_1,\ldots,j_m}^m \right) \right\}+ \sum_{j_{m+1}=1}^{J_{m+1}} d_{j_1,\ldots,j_m,j_{m+1}}, \\
u_{j_1,\ldots,j_m} =& \bm{Z}(S_{j_1,\ldots,j_m})^{\top} V_{j_1,\ldots,j_m}^{-1} \bm{Z}(S_{j_1,\ldots,j_m}) - \bm{Z}(S_{j_1,\ldots,j_m})^{\top} V_{j_1,\ldots,j_m}^{-1} B_{j_1,\ldots,j_m}^m \Phi_{j_1,\ldots, j_m}^{\top} \\
&\times \left( \widehat{K}_{j_1,\ldots,j_m}^{m^{-1}} 
+ \Phi_{j_1,\ldots, j_m} B_{j_1,\ldots,j_m}^{m^{\top}} V_{j_1,\ldots,j_m}^{-1} B_{j_1,\ldots,j_m}^m \Phi_{j_1,\ldots, j_m}^{\top} \right)^{-1}  \Phi_{j_1,\ldots, j_m} \\
&\times B_{j_1,\ldots,j_m}^{m^{\top}} V_{j_1,\ldots,j_m}^{-1}\bm{Z}(S_{j_1,\ldots,j_m})\\
=& - \omega_{j_1,\ldots,j_m}^{m^{\top}} \Phi_{j_1,\ldots,j_m}^{\top} \widetilde{K}_{j_1,\ldots,j_m}^m \Phi_{j_1,\ldots,j_m} \omega_{j_1,\ldots,j_m}^m + \sum_{j_{m+1}=1}^{J_{m+1}} u_{j_1,\ldots,j_m,j_{m+1}}.
\end{align*}

\subsection{Derivation of Algorithm \ref{alg:pred_dist_m-ra-lp}}
\label{subappend:algorithm_pred_dist} 

We will derive Algorithm \ref{alg:pred_dist_m-ra-lp} by an argument similar to that used in the proof of Proposition 2 in \cite{Katzfuss_2017}. Let $Y_{M \text{-RA-lp}}(\bm{s}) \sim \mbox{GP}(0,C_{M \text{-RA-lp}})$ be a zero-mean Gaussian process with the degenerate covariance function $C_{M \text{-RA-lp}}$. Then, we can write
\begin{align*}
Y_{M \text{-RA-lp}}(\bm{s}) = \sum_{m=0}^{M-1} C_m(\bm{s}, Q^{(m)}) \Phi^{{(m)}^{\top}} \bm{\eta}^{(m)} + \delta_M (\bm{s}),
\end{align*}
where $\bm{\eta}^{(m)} = (\bm{\eta}_{1,\ldots,1}^{\top}, \ldots, \bm{\eta}_{J_1,\ldots,J_m}^{\top})^{\top}$, $\bm{\eta}^{(0)} = \bm{\eta}_0$, $\bm{\eta}_{j_1,\ldots,j_m} \sim \mathcal{N} \left(\bm{0}, \widehat{K}_{j_1,\ldots,j_m}^m \right)$ ($m=0,\ldots,M-1$), $\delta_M (\bm{s}) \sim \mbox{GP}(0,C_M \times T_{\gamma})$, and $\bm{\eta}_{j_1,\ldots,j_m}$'s are independent of each other and of $\{ \delta_M (\bm{s}) \}$. 

We define $\mathcal{E}_l = \left\{ \bm{\eta}_{j_1,\ldots,j_a} \middle| 1 \le j_1 \le J_1, \ldots, 1 \le j_a \le J_a, \; a = 0,\ldots,l \right\}$ ($l=0,\ldots,M-1$), $\mathcal{E}_{-1} = \emptyset$, and $L_{j_1,\ldots,j_M}^l = \mbox{Cov} \left( \bm{Y}_{M \text{-RA-lp}}(S_{j_1,\ldots,j_M}^P), \bm{Y}_{M \text{-RA-lp}}(S_{j_1,\ldots,j_l}) \middle| \mathcal{E}_{l-1} \right)$ ($0 \le l \le M$). 
Note that the index $(j_1,\ldots,j_M)$ is fixed. Then, for $0 \le l < M$,
\begin{align}
L_{j_1,\ldots,j_M}^l 
%=& E \left[ \left\{ \sum_{m=l}^{M-1} C_m(S_{j_1,\ldots,j_M}^P, Q^{(m)}) \Phi^{{(m)}^{\top}} \bm{\eta}^{(m)} + \delta_M (S_{j_1,\ldots,j_M}^P)  \right\} \right. \notag \\
%& \left. \times \left\{ \sum_{m=l}^{M-1} C_m(S_{j_1,\ldots,j_l}, Q^{(m)}) \Phi^{{(m)}^{\top}} \bm{\eta}^{(m)} + \delta_M (S_{j_1,\ldots,j_l})  \right\}^{\top} \right] \notag \\
=& C_l (S_{j_1,\ldots,j_M}^P, Q^{(l)}) \Phi^{{(l)}^{\top}} 
\begin{pmatrix}
\widehat{K}_{1,\ldots,1}^l & & O \\
 & \ddots & \\
 O & & \widehat{K}_{J_1,\ldots,J_l}^l
\end{pmatrix}
\Phi^{(l)} C_l(S_{j_1,\ldots,j_l}, Q^{(l)})^{\top} \notag \\
&+ \mbox{Cov} \left( \bm{Y}_{M \text{-RA-lp}}(S_{j_1,\ldots,j_M}^P), \bm{Y}_{M \text{-RA-lp}}(S_{j_1,\ldots,j_l}) \middle| \mathcal{E}_l \right). \label{eq:L_cov_pred}
\end{align}
When we define $B_{j_1,\ldots,j_M}^{l,P} = C_l (S_{j_1,\ldots,j_M}^P, Q_{j_1,\ldots,j_l})$ for $0 \le l \le M$, the first term of \eqref{eq:L_cov_pred} is
\begin{align*}
C_l (S_{j_1,\ldots,j_M}^P, Q_{j_1,\ldots,j_l}) \Phi_{j_1,\ldots,j_l}^{\top} \widehat{K}_{j_1,\ldots,j_l}^l \Phi_{j_1,\ldots,j_l} C_l(S_{j_1,\ldots,j_l}, Q_{j_1,\ldots,j_l})^{\top} \\
= B_{j_1,\ldots,j_M}^{l,P} \Phi_{j_1,\ldots,j_l}^{\top} \widehat{K}_{j_1,\ldots,j_l}^l \Phi_{j_1,\ldots,j_l} B_{j_1,\ldots,j_l}^{l^{\top}}.
\end{align*}
By noting that $\mbox{Cov} \left( \bm{Y}_{M \text{-RA-lp}}(S_{j_1,\ldots,j_M}^P), \bm{Y}_{M \text{-RA-lp}}(S_{j_1,\ldots,j_l, j^{\prime}_{l+1}}) \middle| \mathcal{E}_l \right) = O$ for $j^{\prime}_{l+1} \neq j_{l+1}$ and $\mbox{Cov} \left( \bm{Y}_{M \text{-RA-lp}}(S_{j_1,\ldots,j_M}^P), \bm{Y}_{M \text{-RA-lp}}(S_{j_1,\ldots,j_l, j_{l+1}}) \middle| \mathcal{E}_l \right) = L_{j_1,\ldots,j_M}^{l+1}$,
%\begin{align*}
%\mbox{Cov} \left( \bm{Y}_{M \mathchar`- RA \mathchar`- lp}(S_{j_1,\ldots,j_M}^P), \bm{Y}_{M \mathchar`- RA \mathchar`- lp}(S_{j_1,\ldots,j_l, j^{\prime}_{l+1}}) \middle| \mathcal{E}_l \right) = 0 
%\end{align*}
%and
%\begin{align*}
%\mbox{Cov} \left( \bm{Y}_{M \mathchar`- RA \mathchar`- lp}(S_{j_1,\ldots,j_M}^P), \bm{Y}_{M \mathchar`- RA \mathchar`- lp}(S_{j_1,\ldots,j_l, j_{l+1}}) \middle| \mathcal{E}_l \right) = L_{j_1,\ldots,j_M}^{l+1},
%\end{align*}
the second term of \eqref{eq:L_cov_pred} is expressed as
\begin{align*}
\mbox{Cov} \left( \bm{Y}_{M \text{-RA-lp}}(S_{j_1,\ldots,j_M}^P), \bm{Y}_{M \text{-RA-lp}}(S_{j_1,\ldots,j_l}) \middle| \mathcal{E}_l \right) =& 
\overbrace{
\begin{pmatrix}
O & \smash[b]{\underbrace{L_{j_1,\ldots,j_M}^{l+1}}_{\left| S_{j_1,\ldots,j_l,j_{l+1}} \right|}} & O
\end{pmatrix}
}^{\left|S_{j_1,\ldots,j_l} \right|}
\\ \\ \\
=&\widetilde{L}_{j_1,\ldots,j_M}^l, \quad (\text{say}).
\end{align*}
Therefore, for $0 \le l \le M$, we have
\begin{align}
\label{eq:L_def}
L_{j_1,\ldots,j_M}^l =
\begin{cases}
B_{j_1,\ldots,j_M}^{l,P} \Phi_{j_1,\ldots,j_l}^{\top} \widehat{K}_{j_1,\ldots,j_l}^l \Phi_{j_1,\ldots,j_l} B_{j_1,\ldots,j_l}^{l^{\top}} + \widetilde{L}_{j_1,\ldots,j_M}^l, \quad 0 \le l < M,\\
C_M(S_{j_1,\ldots,j_M}^P, S_{j_1,\ldots,j_M}) \circ T_{\gamma}(S_{j_1,\ldots,j_M}^P, S_{j_1,\ldots,j_M}), \quad l=M.
\end{cases}
\end{align}
From the definition of $\widetilde{L}_{j_1,\ldots,j_M}^l$,
\begin{align}
\widetilde{L}_{j_1,\ldots,j_M}^l V_{j_1,\ldots,j_l}^{-1} B_{j_1,\ldots,j_l}^k =& L_{j_1,\ldots,j_M}^{l+1} \Sigma_{j_1,\ldots,j_l,j_{l+1}}^{-1} B_{j_1,\ldots,j_l,j_{l+1}}^k, \quad 0 \le k \le l, \label{eq:relation_L_til_L_B} \\
\widetilde{L}_{j_1,\ldots,j_M}^l V_{j_1,\ldots,j_l}^{-1} \bm{Z}(S_{j_1,\ldots,j_l}) =& L_{j_1,\ldots,j_M}^{l+1} \Sigma_{j_1,\ldots,j_l,j_{l+1}}^{-1} \bm{Z}(S_{j_1,\ldots,j_l,j_{l+1}}), \label{eq:relation_L_til_L_Z} \\
\widetilde{L}_{j_1,\ldots,j_M}^l V_{j_1,\ldots,j_l}^{-1} \widetilde{L}_{j_1,\ldots,j_M}^{l^{\top}} =& L_{j_1,\ldots,j_M}^{l+1} \Sigma_{j_1,\ldots,j_l,j_{l+1}}^{-1}  L_{j_1,\ldots,j_M}^{{l+1}^{\top}}. \label{eq:relation_L_til_L_L}
\end{align}

Now, it follows from Lemma \ref{lem:inv_formula_variant} that
\begin{align}
\label{eq:relation_K_til_K_hat}
\widetilde{K}_{j_1,\ldots,j_l}^l \Phi_{j_1,\ldots,j_l} B_{j_1,\ldots,j_l}^{l^{\top}} V_{j_1,\ldots,j_l}^{-1} = \widehat{K}_{j_1,\ldots,j_l}^l \Phi_{j_1,\ldots,j_l} B_{j_1,\ldots,j_l}^{l^{\top}} \Sigma_{j_1,\ldots,j_l}^{-1}.
\end{align}
%From the definition of $\widetilde{L}_{j_1,\ldots,j_M}^l$,
%\begin{align}
%\widetilde{L}_{j_1,\ldots,j_M}^l V_{j_1,\ldots,j_l}^{-1} B_{j_1,\ldots,j_l}^k =& L_{j_1,\ldots,j_M}^{l+1} \Sigma_{j_1,\ldots,j_l,j_{l+1}}^{-1} B_{j_1,\ldots,j_l,j_{l+1}}^k, \quad 0 \le k \le l, \label{eq:relation_L_til_L_B} \\
%\widetilde{L}_{j_1,\ldots,j_M}^l V_{j_1,\ldots,j_l}^{-1} \bm{Z}(S_{j_1,\ldots,j_l}) =& L_{j_1,\ldots,j_M}^{l+1} \Sigma_{j_1,\ldots,j_l,j_{l+1}}^{-1} \bm{Z}(S_{j_1,\ldots,j_l,j_{l+1}}), \label{eq:relation_L_til_L_Z} \\
%\widetilde{L}_{j_1,\ldots,j_M}^l V_{j_1,\ldots,j_l}^{-1} \widetilde{L}_{j_1,\ldots,j_M}^{l^{\top}} =& L_{j_1,\ldots,j_M}^{l+1} \Sigma_{j_1,\ldots,j_l,j_{l+1}}^{-1}  L_{j_1,\ldots,j_M}^{{l+1}^{\top}}. \label{eq:relation_L_til_L_L}
%\end{align}
Moreover, from Lemma \ref{lem:log-lik_m-ra-lp} (b) and the definition of $A_{j_1,\ldots,j_m}^{k,l}$, $\Sigma_{j_1,\ldots,j_l}^{-1}$ is expressed as
\begin{align}
\label{eq:Sigma_inv_another_expression}
\Sigma_{j_1,\ldots,j_l}^{-1} = V_{j_1,\ldots,j_l}^{-1} - V_{j_1,\ldots,j_l}^{-1} B_{j_1,\ldots,j_l}^l \Phi_{j_1,\ldots, j_l}^{\top} \widetilde{K}_{j_1,\ldots,j_l}^l \Phi_{j_1,\ldots, j_l} B_{j_1,\ldots,j_l}^{l^{\top}} V_{j_1,\ldots,j_l}^{-1}.
\end{align}

Define $\bm{\mu}_{j_1,\ldots,j_M}^l = L_{j_1,\ldots,j_M}^l \Sigma_{j_1,\ldots,j_l}^{-1} \bm{Z}(S_{j_1,\ldots,j_l})$ for $0 \le l \le M$. From the definition of $L_{j_1,\ldots,j_M}^l$, $\bm{\mu}_{j_1,\ldots,j_M}^0 = C_{M \text{-RA-lp}}(S_{j_1,\ldots,j_M}^P,S_0) \left\{ C_{M \text{-RA-lp}}(S_0,S_0) + \tau^2 \bm{\mbox{I}}_n \right\}^{-1} \bm{Z}(S_0)$. For $0 \le l < M$, it follows from \eqref{eq:L_def}, \eqref{eq:relation_L_til_L_B}, \eqref{eq:relation_L_til_L_Z}, \eqref{eq:relation_K_til_K_hat}, and \eqref{eq:Sigma_inv_another_expression} that
\begin{align}
\bm{\mu}_{j_1,\ldots,j_M}^l =& B_{j_1,\ldots,j_M}^{l,P} \Phi_{j_1,\ldots,j_l}^{\top} \widehat{K}_{j_1,\ldots,j_l}^l \Phi_{j_1,\ldots,j_l} B_{j_1,\ldots,j_l}^{l^{\top}} \Sigma_{j_1,\ldots,j_l}^{-1} \bm{Z}(S_{j_1,\ldots,j_l}) \notag\\
&+ \widetilde{L}_{j_1,\ldots,j_M}^l \Sigma_{j_1,\ldots,j_l}^{-1} \bm{Z}(S_{j_1,\ldots,j_l}) \notag \\
=& B_{j_1,\ldots,j_M}^{l,P} \Phi_{j_1,\ldots,j_l}^{\top} \widetilde{K}_{j_1,\ldots,j_l}^l \Phi_{j_1,\ldots,j_l} B_{j_1,\ldots,j_l}^{l^{\top}} V_{j_1,\ldots,j_l}^{-1} \bm{Z}(S_{j_1,\ldots,j_l}) \notag\\
&+ \widetilde{L}_{j_1,\ldots,j_M}^l V_{j_1,\ldots,j_l}^{-1} \bm{Z}(S_{j_1,\ldots,j_l}) \notag \\
&-  \widetilde{L}_{j_1,\ldots,j_M}^l V_{j_1,\ldots,j_l}^{-1} B_{j_1,\ldots,j_l}^l \Phi_{j_1,\ldots, j_l}^{\top} \widetilde{K}_{j_1,\ldots,j_l}^l \Phi_{j_1,\ldots, j_l} B_{j_1,\ldots,j_l}^{l^{\top}} V_{j_1,\ldots,j_l}^{-1} \bm{Z}(S_{j_1,\ldots,j_l}) \notag \\
=& B_{j_1,\ldots,j_M}^{l,P} \Phi_{j_1,\ldots,j_l}^{\top} \widetilde{K}_{j_1,\ldots,j_l}^l \Phi_{j_1,\ldots,j_l} \omega_{j_1,\ldots,j_l}^l + L_{j_1,\ldots,j_M}^{l+1} \Sigma_{j_1,\ldots,j_{l+1}}^{-1} \bm{Z}(S_{j_1,\ldots,j_{l+1}}) \notag \\
&- L_{j_1,\ldots,j_M}^{l+1} \Sigma_{j_1,\ldots,j_{l+1}}^{-1} B_{j_1,\ldots,j_{l+1}}^l \Phi_{j_1,\ldots, j_l}^{\top} \widetilde{K}_{j_1,\ldots,j_l}^l \Phi_{j_1,\ldots, j_l} \omega_{j_1,\ldots,j_l}^l \notag \\
=& \bm{\mu}_{j_1,\ldots,j_M}^{l+1} +  \widetilde{B}_{j_1,\ldots,j_M}^{l+1, l} \Phi_{j_1,\ldots, j_l}^{\top} \widetilde{K}_{j_1,\ldots,j_l}^l \Phi_{j_1,\ldots, j_l} \omega_{j_1,\ldots,j_l}^l, \label{eq:mu_recursive_l}
\end{align}
where $\widetilde{B}_{j_1,\ldots,j_M}^{l+1, l} = B_{j_1,\ldots,j_M}^{l, P} - L_{j_1,\ldots,j_M}^{l+1} \Sigma_{j_1,\ldots,j_{l+1}}^{-1} B_{j_1,\ldots,j_{l+1}}^{l}$. By applying \eqref{eq:mu_recursive_l} to $\bm{\mu}_{j_1,\ldots,j_M}^0$ recursively, we can obtain
\begin{align*}
\bm{\mu}_{j_1,\ldots,j_M}^0 = & L_{j_1,\ldots,j_M}^M \Sigma_{j_1,\ldots,j_M}^{-1} \bm{Z}(S_{j_1,\ldots,j_M}) \\
&+ \sum_{k=0}^{M-1} \widetilde{B}_{j_1,\ldots,j_M}^{k+1, k} \Phi_{j_1,\ldots,j_k}^{\top} \widetilde{K}_{j_1,\ldots,j_k}^k \Phi_{j_1,\ldots,j_k} \omega_{j_1,\ldots,j_k}^k.
\end{align*}

Next, define $V_{j_1,\ldots,j_M}^{l,P} = \mbox{Var} \left( \bm{Y}_{M \text{-RA-lp}}(S_{j_1,\ldots,j_M}^P) \middle| \mathcal{E}_{l-1} \right)$ for $0 \le l \le M$. By a derivation 
similar to that of \eqref{eq:L_def},
\begin{align}
\label{eq:V_def}
V_{j_1,\ldots,j_M}^{l,P} =
\begin{cases}
B_{j_1,\ldots,j_M}^{l,P} \Phi_{j_1,\ldots,j_l}^{\top} \widehat{K}_{j_1,\ldots,j_l}^l \Phi_{j_1,\ldots,j_l} B_{j_1,\ldots,j_M}^{{l,P}^{\top}} + V_{j_1,\ldots,j_M}^{l+1,P}, \quad 0 \le l < M,\\
C_M(S_{j_1,\ldots,j_M}^P, S_{j_1,\ldots,j_M}^P) \circ T_{\gamma}(S_{j_1,\ldots,j_M}^P, S_{j_1,\ldots,j_M}^P), \quad l=M.
\end{cases}
\end{align}

For $0 \le l < M$, it follows from \eqref{eq:L_def}, \eqref{eq:relation_L_til_L_B}, \eqref{eq:relation_L_til_L_L}, \eqref{eq:relation_K_til_K_hat}, and \eqref{eq:Sigma_inv_another_expression} that
\begin{align}
L_{j_1,\ldots,j_M}^l \Sigma_{j_1,\ldots,j_l}^{-1} L_{j_1,\ldots,j_M}^{l^{\top}} =& B_{j_1,\ldots,j_M}^{l,P} \Phi_{j_1,\ldots,j_l}^{\top} \widehat{K}_{j_1,\ldots,j_l}^l \Phi_{j_1,\ldots,j_l} B_{j_1,\ldots,j_l}^{l^{\top}} \Sigma_{j_1,\ldots,j_l}^{-1} B_{j_1,\ldots,j_l}^l \notag\\
&\times \Phi_{j_1,\ldots,j_l}^{\top} \widehat{K}_{j_1,\ldots,j_l}^l \Phi_{j_1,\ldots,j_l} B_{j_1,\ldots,j_M}^{{l,P}^{\top}} \notag \\
&+ \widetilde{L}_{j_1,\ldots,j_M}^l \Sigma_{j_1,\ldots,j_l}^{-1} B_{j_1,\ldots,j_l}^l \Phi_{j_1,\ldots,j_l}^{\top} \widehat{K}_{j_1,\ldots,j_l}^l \Phi_{j_1,\ldots,j_l} B_{j_1,\ldots,j_M}^{{l,P}^{\top}} \notag \\
&+ B_{j_1,\ldots,j_M}^{l,P} \Phi_{j_1,\ldots,j_l}^{\top} \widehat{K}_{j_1,\ldots,j_l}^l \Phi_{j_1,\ldots,j_l} B_{j_1,\ldots,j_l}^{l^{\top}} \Sigma_{j_1,\ldots,j_l}^{-1} \widetilde{L}_{j_1,\ldots,j_M}^{l^{\top}} \notag \\
&+ \widetilde{L}_{j_1,\ldots,j_M}^l \Sigma_{j_1,\ldots,j_l}^{-1} \widetilde{L}_{j_1,\ldots,j_M}^{l^{\top}} \notag \\
=& B_{j_1,\ldots,j_M}^{l,P} \Phi_{j_1,\ldots,j_l}^{\top} \widetilde{K}_{j_1,\ldots,j_l}^l \Phi_{j_1,\ldots,j_l} A_{j_1,\ldots,j_l}^{l,l} \Phi_{j_1,\ldots,j_l}^{\top} \widehat{K}_{j_1,\ldots,j_l}^l \notag\\
&\times \Phi_{j_1,\ldots,j_l} B_{j_1,\ldots,j_M}^{{l,P}^{\top}} \notag \\
&+ L_{j_1,\ldots,j_M}^{l+1} \Sigma_{j_1,\ldots,j_{l+1}}^{-1} B_{j_1,\ldots,j_{l+1}}^l \Phi_{j_1,\ldots,j_l}^{\top} \widetilde{K}_{j_1,\ldots,j_l}^l \Phi_{j_1,\ldots,j_l} B_{j_1,\ldots,j_M}^{{l,P}^{\top}} \notag \\
&+ B_{j_1,\ldots,j_M}^{l,P} \Phi_{j_1,\ldots,j_l}^{\top} \widetilde{K}_{j_1,\ldots,j_l}^l \Phi_{j_1,\ldots,j_l} B_{j_1,\ldots,j_{l+1}}^{l^{\top}} \Sigma_{j_1,\ldots,j_{l+1}}^{-1} L_{j_1,\ldots,j_M}^{{l+1}^{\top}} \notag \\
&+ L_{j_1,\ldots,j_M}^{l+1} \Sigma_{j_1,\ldots,j_{l+1}}^{-1} L_{j_1,\ldots,j_M}^{{l+1}^{\top}} \notag \\
&- L_{j_1,\ldots,j_M}^{l+1} \Sigma_{j_1,\ldots,j_{l+1}}^{-1} B_{j_1,\ldots,j_{l+1}}^l \Phi_{j_1,\ldots, j_l}^{\top} \widetilde{K}_{j_1,\ldots,j_l}^l \Phi_{j_1,\ldots, j_l} \notag \\
&\times B_{j_1,\ldots,j_{l+1}}^{l^{\top}} \Sigma_{j_1,\ldots,j_{l+1}}^{-1} L_{j_1,\ldots,j_M}^{{l+1}^{\top}} \notag \\
=& L_{j_1,\ldots,j_M}^{l+1} \Sigma_{j_1,\ldots,j_{l+1}}^{-1} L_{j_1,\ldots,j_M}^{{l+1}^{\top}} \notag \\
&+ B_{j_1,\ldots,j_M}^{l,P} \Phi_{j_1,\ldots,j_l}^{\top} \widetilde{K}_{j_1,\ldots,j_l}^l \Phi_{j_1,\ldots,j_l} A_{j_1,\ldots,j_l}^{l,l} \Phi_{j_1,\ldots,j_l}^{\top} \widehat{K}_{j_1,\ldots,j_l}^l \notag\\
&\times \Phi_{j_1,\ldots,j_l} B_{j_1,\ldots,j_M}^{{l,P}^{\top}} \notag \\
&- \widetilde{B}_{j_1,\ldots,j_M}^{l+1, l} \Phi_{j_1,\ldots,j_l}^{\top} \widetilde{K}_{j_1,\ldots,j_l}^l \Phi_{j_1,\ldots,j_l} \widetilde{B}_{j_1,\ldots,j_M}^{{l+1, l}^{\top}} \notag \\
&+ B_{j_1,\ldots,j_M}^{l,P} \Phi_{j_1,\ldots,j_l}^{\top} \widetilde{K}_{j_1,\ldots,j_l}^l \Phi_{j_1,\ldots,j_l} B_{j_1,\ldots,j_M}^{{l,P}^{\top}}. \label{eq:L_Sigma_L}
\end{align}
Also, since $\Phi_{j_1,\ldots,j_l} A_{j_1,\ldots,j_l}^{l,l} \Phi_{j_1,\ldots,j_l}^{\top} = \widetilde{K}_{j_1,\ldots,j_l}^{l^{-1}} - \widehat{K}_{j_1,\ldots,j_l}^{l^{-1}}$ for $0 \le l \le M-1$, we can show that
\begin{align}
&B_{j_1,\ldots,j_M}^{l,P} \Phi_{j_1,\ldots,j_l}^{\top} \widetilde{K}_{j_1,\ldots,j_l}^l \Phi_{j_1,\ldots,j_l} A_{j_1,\ldots,j_l}^{l,l} \Phi_{j_1,\ldots,j_l}^{\top} \widehat{K}_{j_1,\ldots,j_l}^l \Phi_{j_1,\ldots,j_l} B_{j_1,\ldots,j_M}^{{l,P}^{\top}} \notag \\
&+ B_{j_1,\ldots,j_M}^{l,P} \Phi_{j_1,\ldots,j_l}^{\top} \widetilde{K}_{j_1,\ldots,j_l}^l \Phi_{j_1,\ldots,j_l} B_{j_1,\ldots,j_M}^{{l,P}^{\top}} \notag \\
=& B_{j_1,\ldots,j_M}^{l,P} \Phi_{j_1,\ldots,j_l}^{\top} \widehat{K}_{j_1,\ldots,j_l}^l \Phi_{j_1,\ldots,j_l} B_{j_1,\ldots,j_M}^{{l,P}^{\top}}. \label{eq:K_til_K_hat_minus}
\end{align}

Now, define $\Psi_{j_1,\ldots,j_M}^l = V_{j_1,\ldots,j_M}^{l, P} - L_{j_1,\ldots,j_M}^l \Sigma_{j_1,\ldots,j_l}^{-1} L_{j_1,\ldots,j_M}^{l^{\top}}$ for $0 \le l \le M$. From the definition of $V_{j_1,\ldots,j_M}^{l, P}$, we have 
\begin{align*}
\Psi_{j_1,\ldots,j_M}^0 =& C_{M \text{-RA-lp}}(S_{j_1,\ldots,j_M}^P,S_{j_1,\ldots,j_M}^P)\\
&-C_{M \text{-RA-lp}}(S_{j_1,\ldots,j_M}^P,S_0) \left\{ C_{M \text{-RA-lp}}(S_0,S_0) + \tau^2 \bm{\mbox{I}}_n \right\}^{-1} C_{M \text{-RA-lp}}(S_{j_1,\ldots,j_M}^P,S_0)^{\top}.
\end{align*}
Furthermore, from \eqref{eq:V_def}, \eqref{eq:L_Sigma_L}, and \eqref{eq:K_til_K_hat_minus},  
\begin{align}
\Psi_{j_1,\ldots,j_M}^l 
%=& V_{j_1,\ldots,j_M}^{l+1, P} - L_{j_1,\ldots,j_M}^{l+1} \Sigma_{j_1,\ldots,j_{l+1}}^{-1} L_{j_1,\ldots,j_M}^{{l+1}^{\top}} + \widetilde{B}_{j_1,\ldots,j_M}^{l+1, l} \Phi_{j_1,\ldots,j_l}^{\top} \widetilde{K}_{j_1,\ldots,j_l}^l \Phi_{j_1,\ldots,j_l} \widetilde{B}_{j_1,\ldots,j_M}^{{l+1, l}^{\top}} \notag \\
= \Psi_{j_1,\ldots,j_M}^{l+1}+ \widetilde{B}_{j_1,\ldots,j_M}^{l+1, l} \Phi_{j_1,\ldots,j_l}^{\top} \widetilde{K}_{j_1,\ldots,j_l}^l \Phi_{j_1,\ldots,j_l} \widetilde{B}_{j_1,\ldots,j_M}^{{l+1, l}^{\top}}, \label{eq:Psi_recursive}
\end{align}
for $0 \le l < M$. Thus, by applying \eqref{eq:Psi_recursive} to $\Psi_{j_1,\ldots,j_M}^0$ recursively, it follows that
\begin{align*}
\Psi_{j_1,\ldots,j_M}^0 =& V_{j_1,\ldots,j_M}^{M, P} 
- L_{j_1,\ldots,j_M}^M \Sigma_{j_1,\ldots,j_M}^{-1} L_{j_1,\ldots,j_M}^{M^{\top}}\\
& + \sum_{k=0}^{M-1} \widetilde{B}_{j_1,\ldots,j_M}^{k+1, k} \Phi_{j_1,\ldots,j_k}^{\top} \widetilde{K}_{j_1,\ldots,j_k}^k \Phi_{j_1,\ldots,j_k} \widetilde{B}_{j_1,\ldots,j_M}^{{k+1, k}^{\top}}.
\end{align*}

Lastly, we will derive the calculation method of the matrices required to obtain $\bm{\mu}_{j_1,\ldots,j_M}^0$ and $\Psi_{j_1,\ldots,j_M}^0$. We define $\widetilde{B}_{j_1,\ldots,j_M}^{l, k} = B_{j_1,\ldots,j_M}^{k, P} - L_{j_1,\ldots,j_M}^{l} \Sigma_{j_1,\ldots,j_l}^{-1} B_{j_1,\ldots,j_l}^k$ for $0 \le k < l \le M$. For $0 \le k<l < M$, by using \eqref{eq:L_def}, \eqref{eq:relation_L_til_L_B},  \eqref{eq:relation_K_til_K_hat}, and \eqref{eq:Sigma_inv_another_expression},
\begin{align*}
\widetilde{B}_{j_1,\ldots,j_M}^{l, k} =& B_{j_1,\ldots,j_M}^{k, P} 
- B_{j_1,\ldots,j_M}^{l,P} \Phi_{j_1,\ldots,j_l}^{\top} \widehat{K}_{j_1,\ldots,j_l}^l \Phi_{j_1,\ldots,j_l} B_{j_1,\ldots,j_l}^{l^{\top}} \Sigma_{j_1,\ldots,j_l}^{-1} B_{j_1,\ldots,j_l}^k \\
&-\widetilde{L}_{j_1,\ldots,j_M}^{l} \Sigma_{j_1,\ldots,j_l}^{-1} B_{j_1,\ldots,j_l}^k \\
=& B_{j_1,\ldots,j_M}^{k,P} 
- B_{j_1,\ldots,j_M}^{l,P} \Phi_{j_1,\ldots,j_l}^{\top} \widetilde{K}_{j_1,\ldots,j_l}^l \Phi_{j_1,\ldots,j_l} B_{j_1,\ldots,j_l}^{l^{\top}} V_{j_1,\ldots,j_l}^{-1} B_{j_1,\ldots,j_l}^k \\
&-\widetilde{L}_{j_1,\ldots,j_M}^{l} V_{j_1,\ldots,j_l}^{-1} B_{j_1,\ldots,j_l}^k \\
&+ \widetilde{L}_{j_1,\ldots,j_M}^{l} V_{j_1,\ldots,j_l}^{-1} B_{j_1,\ldots,j_l}^l \Phi_{j_1,\ldots,j_l}^{\top} \widetilde{K}_{j_1,\ldots,j_l}^l \Phi_{j_1,\ldots,j_l} B_{j_1,\ldots,j_l}^{l^{\top}} V_{j_1,\ldots,j_l}^{-1} B_{j_1,\ldots,j_l}^k \\
=& B_{j_1,\ldots,j_M}^{k,P} 
- B_{j_1,\ldots,j_M}^{l,P} \Phi_{j_1,\ldots,j_l}^{\top} \widetilde{K}_{j_1,\ldots,j_l}^l \Phi_{j_1,\ldots,j_l} A_{j_1,\ldots,j_l}^{l,k} \\
&-L_{j_1,\ldots,j_M}^{l+1} \Sigma_{j_1,\ldots,j_{l+1}}^{-1} B_{j_1,\ldots,j_{l+1}}^k \\
&+ L_{j_1,\ldots,j_M}^{l+1} \Sigma_{j_1,\ldots,j_{l+1}}^{-1} B_{j_1,\ldots,j_{l+1}}^l \Phi_{j_1,\ldots,j_l}^{\top} \widetilde{K}_{j_1,\ldots,j_l}^l \Phi_{j_1,\ldots,j_l} A_{j_1,\ldots,j_l}^{l,k} \\
=& \widetilde{B}_{j_1,\ldots,j_M}^{l+1, k} - \widetilde{B}_{j_1,\ldots,j_M}^{l+1, l} \Phi_{j_1,\ldots,j_l}^{\top} \widetilde{K}_{j_1,\ldots,j_l}^l \Phi_{j_1,\ldots,j_l} A_{j_1,\ldots,j_l}^{l,k}.
\end{align*}
Next, for $1 \le l \le M$, it follows from \eqref{eq:cov_resolution_m}, \eqref{eq:lp_resolution_m}, and \eqref{eq:cov_resolution_M} that 
\begin{align}
\label{eq:B_P_recursive}
B_{j_1,\ldots,j_M}^{l,P} =& C_l (S_{j_1,\ldots,j_M}^P, Q_{j_1,\ldots,j_l}) \notag \\
=& C_{l-1} (S_{j_1,\ldots,j_M}^P, Q_{j_1,\ldots,j_l}) - B_{j_1,\ldots,j_M}^{l-1,P} \Phi_{j_1,\ldots,j_{l-1}}^{\top} \widehat{K}_{j_1,\ldots,j_{l-1}}^{l-1} \Phi_{j_1,\ldots,j_{l-1}} K_{j_1,\ldots,j_l}^{l-1}.
\end{align}
By using \eqref{eq:B_P_recursive} for $B_{j_1,\ldots,j_M}^{l,P}$ recursively, 
\begin{align*}
B_{j_1,\ldots,j_M}^{l,P} = C_0(S_{j_1,\ldots,j_M}^P,Q_{j_1,\ldots,j_l})-\sum_{k=0}^{l-1} B_{j_1,\ldots,j_M}^{k, P} \Phi_{j_1,\ldots,j_k}^{\top}
\widehat{K}_{j_1,\ldots,j_k}^k \Phi_{j_1,\ldots,j_k} K_{j_1,\ldots,j_l}^k.
\end{align*}
Note that $B_{j_1,\ldots,j_M}^{M,P} = C_M (S_{j_1,\ldots,j_M}^P, S_{j_1,\ldots,j_M})$ in Step 2 of Algorithm \ref{alg:pred_dist_m-ra-lp} because $Q_{j_1,\ldots,j_M}=S_{j_1,\ldots,j_M}$. Similarly, since
\begin{align*}
C_M (S_{j_1,\ldots,j_M}^P, S_{j_1,\ldots,j_M}^P) =& C_{M-1} (S_{j_1,\ldots,j_M}^P, S_{j_1,\ldots,j_M}^P)\\ 
&- B_{j_1,\ldots,j_M}^{M-1,P} \Phi_{j_1,\ldots,j_{M-1}}^{\top} \widehat{K}_{j_1,\ldots,j_{M-1}}^{M-1} \Phi_{j_1,\ldots,j_{M-1}} B_{j_1,\ldots,j_M}^{{M-1,P}^{\top}}
\end{align*}
from \eqref{eq:lp_resolution_m} and \eqref{eq:cov_resolution_M}, we obtain
\begin{align*}
C_M(S_{j_1,\ldots,j_M}^P, S_{j_1,\ldots,j_M}^P) =& C_0(S_{j_1,\ldots,j_M}^P,S_{j_1,\ldots,j_M}^P)\\
&-\sum_{k=0}^{M-1} B_{j_1,\ldots,j_M}^{k, P} \Phi_{j_1,\ldots,j_k}^{\top} \widehat{K}_{j_1,\ldots,j_k}^k \Phi_{j_1,\ldots,j_k} B_{j_1,\ldots,j_M}^{{k, P}^{\top}}.
\end{align*}

\def\thesection{Appendix \Alph{section}}
\section{Distributed Computing}
\def\thesection{\Alph{section}}
\label{append:distributed_computing}

It is assumed that we have nodes $\mathcal{N}_{j_1, \ldots,j_m}$ ($m = 0,\ldots, M$, $1 \le j_1 \le J_1, \ldots, 1 \le j_{M} \le J_{M}$) with a tree-like structure where $\mathcal{N}_0$ represents a root node, and children of $\mathcal{N}_{j^{\prime}_1, \ldots,j^{\prime}_m}$ are $\mathcal{N}_{j^{\prime}_1, \ldots,j^{\prime}_m,i}$ $(i=1,\ldots, J_{m+1})$ for a fixed index $(j^{\prime}_1, \ldots,j^{\prime}_m)$ ($m=0,\ldots,M-1$). It is left to a future study to investigate how much the parallelization reduces the computational time beyond the cost of the communication. 
%$\mathcal{N}_{j^{\prime}_1, \ldots,j^{\prime}_m,1}, \ldots, \mathcal{N}_{j^{\prime}_1, \ldots,j^{\prime}_m,J_{m+1}}$ for a fixed index $(j^{\prime}_1, \ldots,j^{\prime}_m)$ ($m=0,\ldots,M-1$) are  children of $\mathcal{N}_{j^{\prime}_1, \ldots,j^{\prime}_m}$.

\subsection{A parallel version of Algorithm \ref{alg:log-lik_m-ra-lp}}
\label{subappend:parallel_version_alg_log-lik}

The following algorithm enables us to calculate some quantities of Algorithm \ref{alg:log-lik_m-ra-lp} in parallel.

\begin{algorithm}[A parallel version of Algorithm \ref{alg:log-lik_m-ra-lp}]
\label{alg:log-lik_m-ra-lp_parallel}
Given $M > 1$, $D_{j_1,\ldots,j_m}$ ($m = 1,\ldots, M$, $1 \le j_1 \le J_1, \ldots, 1 \le j_M \le J_M$), $Q_{j_1,\ldots,j_m}$ ($m = 0,\ldots, M-1$, $1 \le j_1 \le J_1, \ldots, 1 \le j_{M-1} \le J_{M-1}$),  and $\gamma > 0$, find $d_0 = \log \left[ \mbox{det} \left\{ C_{M \text{-RA-lp}}(S_0,S_0) + \tau^2 \bm{\mbox{I}}_n \right\} \right]$ and $u_0 = \bm{Z}(S_0)^{\top} \left\{ C_{M \text{-RA-lp}}(S_0,S_0) + \tau^2 \bm{\mbox{I}}_n \right\}^{-1} \bm{Z}(S_0)$.

\bigskip
\noindent
\textit{Step} 1. Conduct Step 1 in Algorithm \ref{alg:log-lik_m-ra-lp}.
%Each node $\mathcal{N}_{j_1, \ldots,j_m}$ holds

\noindent
\textit{Step} 2. In each node $\mathcal{N}_{j_1, \ldots,j_M}$, calculate 
\begin{align*}
\widetilde{A}_{j_1,\ldots,j_M}^{k,l} =& B_{j_1,\ldots,j_M}^{k^{\top}} \Sigma_{j_1,\ldots,j_M}^{-1} B_{j_1,\ldots,j_M}^l, \quad 0 \le k \le l < M,\\
\widetilde{\omega}_{j_1,\ldots,j_M}^{k} =& B_{j_1,\ldots,j_M}^{k^{\top}} \Sigma_{j_1,\ldots,j_M}^{-1} \bm{Z}(S_{j_1,\ldots,j_M}), \quad 0 \le k < M,\\
d_{j_1,\ldots,j_M} =& \log \left\{ \mbox{det} \left( \Sigma_{j_1,\ldots,j_M} \right) \right\},\\
u_{j_1,\ldots,j_M} =& \bm{Z}(S_{j_1,\ldots,j_M})^{\top}\Sigma_{j_1,\ldots,j_M}^{-1} \bm{Z}(S_{j_1,\ldots,j_M}).
\end{align*}
Send $\widetilde{A}_{j_1,\ldots,j_M}^{k,l}$, $\widetilde{\omega}_{j_1,\ldots,j_M}^{k}$, $d_{j_1,\ldots,j_M}$, and $u_{j_1,\ldots,j_M}$ to its parent, that is, $\mathcal{N}_{j_1, \ldots,j_{M-1}}$.

\noindent
\textit{Step} 3. In each node $\mathcal{N}_{j_1, \ldots,j_m}$ ($m=1,\ldots,M-1$), calculate
\begin{align*}
A_{j_1,\ldots,j_m}^{k,l} =& \sum_{j_{m+1}=1}^{J_{m+1}} \widetilde{A}_{j_1,\ldots,j_m,j_{m+1}}^{k,l}, \quad 0 \le k \le l \le m,\\
\widetilde{K}_{j_1,\ldots,j_m}^m =& \left( \widehat{K}_{j_1,\ldots,j_m}^{m^{-1}} + \Phi_{j_1,\ldots,j_m} A_{j_1,\ldots,j_m}^{m,m} \Phi_{j_1,\ldots,j_m}^{\top} \right)^{-1},\\
\widetilde{A}_{j_1,\ldots,j_m}^{k,l} =& A_{j_1,\ldots,j_m}^{k,l} - A_{j_1,\ldots,j_m}^{k,m} \Phi_{j_1,\ldots,j_m}^{\top} \widetilde{K}_{j_1,\ldots,j_m}^m \Phi_{j_1,\ldots,j_m} A_{j_1,\ldots,j_m}^{m,l}, \quad 0 \le k \le l < m,\\
\omega_{j_1,\ldots,j_m}^k =& \sum_{j_{m+1}=1}^{J_{m+1}} \widetilde{\omega}_{j_1,\ldots,j_m,j_{m+1}}^k, \quad 0 \le k \le m,\\
\widetilde{\omega}_{j_1,\ldots,j_m}^k =& \omega_{j_1,\ldots,j_m}^k - A_{j_1,\ldots,j_m}^{k,m} \Phi_{j_1,\ldots,j_m}^{\top} \widetilde{K}_{j_1,\ldots,j_m}^m \Phi_{j_1,\ldots,j_m} \omega_{j_1,\ldots,j_m}^m, \quad 0 \le k < m,\\
d_{j_1,\ldots,j_m} =& -\log \left\{ \mbox{det} \left( \widetilde{K}_{j_1,\ldots,j_m}^m \right) \right\} + \log \left\{ \mbox{det} \left( \widehat{K}_{j_1,\ldots,j_m}^m \right) \right\} + \sum_{j_{m+1}=1}^{J_{m+1}} d_{j_1,\ldots,j_m,j_{m+1}},\\
u_{j_1,\ldots,j_m} =& - \omega_{j_1,\ldots,j_m}^{m^{\top}} \Phi_{j_1,\ldots,j_m}^{\top} \widetilde{K}_{j_1,\ldots,j_m}^m \Phi_{j_1,\ldots,j_m} \omega_{j_1,\ldots,j_m}^m + \sum_{j_{m+1}=1}^{J_{m+1}} u_{j_1,\ldots,j_m,j_{m+1}}.
\end{align*}
Send $\widetilde{A}_{j_1,\ldots,j_m}^{k,l}$, $\widetilde{\omega}_{j_1,\ldots,j_m}^{k}$, $d_{j_1,\ldots,j_m}$, and $u_{j_1,\ldots,j_m}$ to its parent, that is, $\mathcal{N}_{j_1, \ldots,j_{m-1}}$.

\noindent
\textit{Step} 4. In $\mathcal{N}_0$, calculate
\begin{align*}
A_{0}^{0,0} =& \sum_{j_{1}=1}^{J_{1}} \widetilde{A}_{j_1}^{0,0}, \\
\widetilde{K}_{0}^0 =& \left( \widehat{K}_{0}^{0^{-1}} + \Phi_{0} A_{0}^{0,0} \Phi_{0}^{\top} \right)^{-1},\\
\omega_{0}^0 =& \sum_{j_{1}=1}^{J_{1}} \widetilde{\omega}_{j_1}^0, \\
d_{0} =& -\log \left\{ \mbox{det} \left( \widetilde{K}_{0}^0 \right) \right\} + \log \left\{ \mbox{det} \left( \widehat{K}_{0}^0 \right) \right\} + \sum_{j_{1}=1}^{J_{1}} d_{j_1},\\
u_{0} =& - \omega_{0}^{0^{\top}} \Phi_{0}^{\top} \widetilde{K}_{0}^0 \Phi_{0} \omega_{0}^0 + \sum_{j_{1}=1}^{J_{1}} u_{j_1}.
\end{align*}

\noindent
\textit{Step} 5. Output $d_0$ and $u_0$.
\end{algorithm}

In Steps 2 and 3, the calculations at the nodes for each resolution can be conducted in parallel. Furthermore, if each node $\mathcal{N}_{j_1, \ldots,j_M}$ has multiple cores, $\widetilde{A}_{j_1,\ldots,j_M}^{k,l}$, $\widetilde{\omega}_{j_1,\ldots,j_M}^{k}$, $d_{j_1,\ldots,j_M}$, and $u_{j_1,\ldots,j_M}$ can also be calculated in parallel. Thus, we can conduct the efficient computation, but 
the communication of sending the matrices to the parent is required in Steps 2 and 3.

\subsection{A parallel version of Algorithm \ref{alg:pred_dist_m-ra-lp}}
\label{subappend:parallel_version_alg_pred_dist}
The following algorithm allows us to calculate $\bm{\mu}_{j_1,\ldots,j_M}^0$ and $\Psi_{j_1,\ldots,j_M}^0$ ($1 \le j_1 \le J_1, \ldots, 1 \le j_{M} \le J_{M}$) in parallel.

\begin{algorithm}[A parallel version of Algorithm \ref{alg:pred_dist_m-ra-lp}]
\label{alg:pred_dist_m-ra-lp_parallel}
Given $M > 1$, $D_{j_1,\ldots,j_m}$ ($m = 1,\ldots, M$, $1 \le j_1 \le J_1, \ldots, 1 \le j_M \le J_M$), $Q_{j_1,\ldots,j_m}$ ($m = 0,\ldots, M-1$, $1 \le j_1 \le J_1, \ldots, 1 \le j_{M-1} \le J_{M-1}$), $\gamma > 0$, and $S_{j_1,\ldots,j_M}^P$ ($1 \le j_1 \le J_1, \ldots, 1 \le j_{M} \le J_{M}$), find $\bm{\mu}_{j_1,\ldots,j_M}^0 = C_{M \text{-RA-lp}}(S_{j_1,\ldots,j_M}^P,S_0) \left\{ C_{M \text{-RA-lp}}(S_0,S_0) + \tau^2 \bm{\mbox{I}}_n \right\}^{-1} \bm{Z}(S_0)$ and $\Psi_{j_1,\ldots,j_M}^0 =\\ 
C_{M \text{-RA-lp}}(S_{j_1,\ldots,j_M}^P,S_{j_1,\ldots,j_M}^P)-C_{M \text{-RA-lp}}(S_{j_1,\ldots,j_M}^P,S_0) \left\{ C_{M \text{-RA-lp}}(S_0,S_0) + \tau^2 \bm{\mbox{I}}_n \right\}^{-1} \\
\times C_{M \text{-RA-lp}}(S_{j_1,\ldots,j_M}^P,S_0)^{\top}$ ($1 \le j_1 \le J_1, \ldots, 1 \le j_{M} \le J_{M}$).

\bigskip
\noindent
\textit{Step} 1. Conduct Step 1 in Algorithm \ref{alg:log-lik_m-ra-lp}. 

\par
\noindent
\textit{Step} 2. For the indices $(j_1,\ldots,j_M)$ ($1 \le j_1 \le J_1, \ldots, 1 \le j_{M} \le J_{M}$), calculate $K_{j_1,\ldots,j_l}^k$ ($0 \le k \le l-1$, $1 \le l \le M-1$) and conduct Step 2 in Algorithm \ref{alg:pred_dist_m-ra-lp}.

\noindent
\textit{Step} 3. In each node $\mathcal{N}_{j_1, \ldots,j_M}$, calculate 
\begin{align*}
\widetilde{A}_{j_1,\ldots,j_M}^{k,l} =& B_{j_1,\ldots,j_M}^{k^{\top}} \Sigma_{j_1,\ldots,j_M}^{-1} B_{j_1,\ldots,j_M}^l, \quad 0 \le k \le l < M,\\
\widetilde{\omega}_{j_1,\ldots,j_M}^{k} =& B_{j_1,\ldots,j_M}^{k^{\top}} \Sigma_{j_1,\ldots,j_M}^{-1} \bm{Z}(S_{j_1,\ldots,j_M}), \quad 0 \le k < M,\\
\widetilde{B}_{j_1,\ldots,j_M}^{M, k} =& B_{j_1,\ldots,j_M}^{k, P} - L_{j_1,\ldots,j_M}^M \Sigma_{j_1,\ldots,j_M}^{-1} B_{j_1,\ldots,j_M}^{k}, \quad 0 \le k < M,\\
\Psi_{j_1,\ldots,j_M}^0 =& V_{j_1,\ldots,j_M}^{M, P} - L_{j_1,\ldots,j_M}^M \Sigma_{j_1,\ldots,j_M}^{-1} L_{j_1,\ldots,j_M}^{M^{\top}}, \\ %\quad 1 \le j_i \le J_i,\; i=1,\ldots,M,\\ 
\bm{\mu}_{j_1,\ldots,j_M}^0 =& L_{j_1,\ldots,j_M}^M \Sigma_{j_1,\ldots,j_M}^{-1} \bm{Z}(S_{j_1,\ldots,j_M}). % \quad 1 \le j_i \le J_i,\; i=1,\ldots,M.
\end{align*}
Send $\widetilde{A}_{j_1,\ldots,j_M}^{k,l}$, $\widetilde{\omega}_{j_1,\ldots,j_M}^{k}$, $\widetilde{B}_{j_1,\ldots,j_M}^{M, k}$, $\Psi_{j_1,\ldots,j_M}^0$, and $\bm{\mu}_{j_1,\ldots,j_M}^0$ to its parent, that is, $\mathcal{N}_{j_1, \ldots,j_{M-1}}$.

\noindent
\textit{Step} 4. In each node $\mathcal{N}_{j_1, \ldots,j_m}$ ($m=1,\ldots,M-1$), calculate
\begin{align*}
A_{j_1,\ldots,j_m}^{k,l} =& \sum_{j_{m+1}=1}^{J_{m+1}} \widetilde{A}_{j_1,\ldots,j_m,j_{m+1}}^{k,l}, \quad 0 \le k \le l \le m,\\
\widetilde{K}_{j_1,\ldots,j_m}^m =& \left( \widehat{K}_{j_1,\ldots,j_m}^{m^{-1}} + \Phi_{j_1,\ldots,j_m} A_{j_1,\ldots,j_m}^{m,m} \Phi_{j_1,\ldots,j_m}^{\top} \right)^{-1},\\
\widetilde{A}_{j_1,\ldots,j_m}^{k,l} =& A_{j_1,\ldots,j_m}^{k,l} - A_{j_1,\ldots,j_m}^{k,m} \Phi_{j_1,\ldots,j_m}^{\top} \widetilde{K}_{j_1,\ldots,j_m}^m \Phi_{j_1,\ldots,j_m} A_{j_1,\ldots,j_m}^{m,l}, \quad 0 \le k \le l < m,\\
\omega_{j_1,\ldots,j_m}^k =& \sum_{j_{m+1}=1}^{J_{m+1}} \widetilde{\omega}_{j_1,\ldots,j_m,j_{m+1}}^k, \quad 0 \le k \le m,\\
\widetilde{\omega}_{j_1,\ldots,j_m}^k =& \omega_{j_1,\ldots,j_m}^k - A_{j_1,\ldots,j_m}^{k,m} \Phi_{j_1,\ldots,j_m}^{\top} \widetilde{K}_{j_1,\ldots,j_m}^m \Phi_{j_1,\ldots,j_m} \omega_{j_1,\ldots,j_m}^m, \quad 0 \le k < m,\\
\widetilde{B}_{j_1,\ldots,j_M}^{m, k} =& \widetilde{B}_{j_1,\ldots,j_M}^{m+1, k} - \widetilde{B}_{j_1,\ldots,j_M}^{m+1, m} \Phi_{j_1,\ldots,j_m}^{\top} \widetilde{K}_{j_1,\ldots,j_m}^m \Phi_{j_1,\ldots,j_m} A_{j_1,\ldots,j_m}^{m,k}, \quad 0 \le k < m,\\
& \hspace{5.55cm} 1 \le j_i \le J_i,\; i=m+1,\ldots,M,\\
\Psi_{j_1,\ldots,j_M}^0 =& \Psi_{j_1,\ldots,j_M}^0 + \widetilde{B}_{j_1,\ldots,j_M}^{m+1, m} \Phi_{j_1,\ldots,j_m}^{\top} \widetilde{K}_{j_1,\ldots,j_m}^m \Phi_{j_1,\ldots,j_m} \widetilde{B}_{j_1,\ldots,j_M}^{{m+1, m}^{\top}}, \quad 1 \le j_i \le J_i,\\
& \hspace{7.48cm} i=m+1,\ldots,M,\\
\end{align*}
\begin{align*}
\bm{\mu}_{j_1,\ldots,j_M}^0 =& \bm{\mu}_{j_1,\ldots,j_M}^0 + \widetilde{B}_{j_1,\ldots,j_M}^{m+1, m} \Phi_{j_1,\ldots,j_m}^{\top} \widetilde{K}_{j_1,\ldots,j_m}^m \Phi_{j_1,\ldots,j_m} \omega_{j_1,\ldots,j_m}^m, \quad 1 \le j_i \le J_i,\\
& \hspace{7.38cm} i=m+1,\ldots,M.
\end{align*}
Send $\widetilde{A}_{j_1,\ldots,j_m}^{k,l}$, $\widetilde{\omega}_{j_1,\ldots,j_m}^{k}$, $\widetilde{B}_{j_1,\ldots,j_M}^{m, k}$, $\Psi_{j_1,\ldots,j_M}^0$, and $\bm{\mu}_{j_1,\ldots,j_M}^0$ to its parent, that is, $\mathcal{N}_{j_1, \ldots,j_{m-1}}$.

\noindent
\textit{Step} 5. In $\mathcal{N}_0$, calculate
\begin{align*}
A_{0}^{0,0} =& \sum_{j_{1}=1}^{J_{1}} \widetilde{A}_{j_1}^{0,0}, \\
\widetilde{K}_{0}^0 =& \left( \widehat{K}_{0}^{0^{-1}} + \Phi_{0} A_{0}^{0,0} \Phi_{0}^{\top} \right)^{-1},\\
\omega_{0}^0 =& \sum_{j_{1}=1}^{J_{1}} \widetilde{\omega}_{j_1}^0, \\
\Psi_{j_1,\ldots,j_M}^0 =& \Psi_{j_1,\ldots,j_M}^0 + \widetilde{B}_{j_1,\ldots,j_M}^{1, 0} \Phi_{0}^{\top} \widetilde{K}_{0}^0 \Phi_{0} \widetilde{B}_{j_1,\ldots,j_M}^{{1, 0}^{\top}}, \quad 1 \le j_i \le J_i,\; i=1,\ldots,M,\\
\bm{\mu}_{j_1,\ldots,j_M}^0 =& \bm{\mu}_{j_1,\ldots,j_M}^0 + \widetilde{B}_{j_1,\ldots,j_M}^{1, 0} \Phi_{0}^{\top} \widetilde{K}_{0}^0 \Phi_{0} \omega_{0}^0, \quad 1 \le j_i \le J_i,\; i=1,\ldots,M.
\end{align*}

\noindent
\textit{Step} 6. Output $\bm{\mu}_{j_1,\ldots,j_M}^0$ and $\Psi_{j_1,\ldots,j_M}^0$ ($1 \le j_i \le J_i$, $i=1,\ldots,M$).
\end{algorithm}

%Similar to Algorithm \ref{alg:log-lik_m-ra-lp_parallel}, not only Steps 3 and 4 but also each quantities of Step 3 can be parallelized in Algorithm \ref{alg:pred_dist_m-ra-lp_parallel}. 
Similar to Algorithm \ref{alg:log-lik_m-ra-lp_parallel}, in addition to Steps 3 and 4, quantities in each node $\mathcal{N}_{j_1, \ldots,j_M}$ of Step 3 can also be parallelized. 
Additionally, unlike Algorithm \ref{alg:pred_dist_m-ra-lp}, Algorithm \ref{alg:pred_dist_m-ra-lp_parallel} can calculate $\bm{\mu}_{j_1,\ldots,j_M}^0$ and $\Psi_{j_1,\ldots,j_M}^0$ ($1 \le j_1 \le J_1, \ldots, 1 \le j_{M} \le J_{M}$) in parallel. However, if the size of the prediction locations is large, the communication of sending the matrices to the parent in Steps 3 and 4 of Algorithm \ref{alg:pred_dist_m-ra-lp_parallel} is likely to cause a nonnegligible  computational burden.

%
%----------------------------------------------------------------------
% Acknowledgements
%---------------------------------------------------------------------- 
%
\section*{Acknowledgments}
%\begin{acknowledgements}
The author gratefully acknowledges the helpful comments and suggestions from the 
editor and two anonymous referees that refined the manuscript. 
This work is supported financially by JSPS KAKENHI Grant Number 18K12755.
%\end{acknowledgements}

%\section*{Conflict of interest}
%
%The corresponding author states that there is no conflict of interest.
%The authors declare that they have no conflict of interest.

%
%----------------------------------------------------------------------
% References
%---------------------------------------------------------------------- 
%

%\bibliographystyle{chicago}
\bibliographystyle{apalike}
\bibliography{m-ra-lp_paper}   % name your BibTeX data base

%% Non-BibTeX users please use
%\begin{thebibliography}{}
%%
%% and use \bibitem to create references. Consult the Instructions
%% for authors for reference list style.
%\end{thebibliography}

\end{document}